\documentclass[letterpaper,11pt]{article}
\usepackage{microtype}
\usepackage[letterpaper,margin=1in]{geometry}
\usepackage[numbers,sort&compress]{natbib}
\usepackage[table]{xcolor}
\usepackage{amsmath}
\usepackage{amssymb}
\usepackage{mathrsfs}
\usepackage{amsthm}
\usepackage{enumitem}
\usepackage{textgreek}
\usepackage{graphicx}
\usepackage{framed}
\usepackage[small]{caption}
\usepackage{booktabs}
\usepackage{tikz-cd}
\usepackage{hyperref}
\usepackage{doi}
\usepackage[textsize=tiny]{todonotes}
\usepackage[framemethod=tikz]{mdframed}
\usepackage{thm-restate}
\usepackage{mathtools}
\usepackage{xspace}
\usepackage[capitalize, nameinlink]{cleveref}
\usepackage{multirow}
\usepackage{cellspace}
\usepackage{array}
\usepackage{algorithm}
\usepackage{algpseudocodex}

\Crefname{remark}{Remark}{Remarks}
\Crefname{observation}{Observation}{Observations}

\theoremstyle{plain}
\newtheorem{theorem}{Theorem}[section]
\newtheorem{lemma}[theorem]{Lemma}

\newtheorem{corollary}[theorem]{Corollary}

\newtheorem{observation}[theorem]{Observation}

\theoremstyle{definition}

\theoremstyle{plain}

\newcounter{open}
\newtheorem{oq}[open]{Open Problem}

\theoremstyle{remark}

\newcommand{\s}{~\mspace{1mu}\allowbreak}
\newcommand{\mybox}[1]{\mspace{2mu}{\setlength{\fboxsep}{1.5pt}\color{lightgray}\boxed{\color{black}\scriptstyle #1\vphantom{+}}}\mspace{2mu}}
\newcommand{\A}{\mathsf{A}}
\newcommand{\B}{\mathsf{B}}
\newcommand{\C}{\mathsf{C}}
\newcommand{\D}{\mathsf{D}}
\newcommand{\E}{\mathsf{E}}
\newcommand{\F}{\mathsf{F}}
\renewcommand{\L}{\mathsf{L}}
\newcommand{\T}{\mathsf{T}}
\renewcommand{\c}{\mathsf{c}}
\renewcommand{\a}{\mathsf{a}}
\renewcommand{\b}{\mathsf{b}}

\renewcommand{\S}{\mathsf{S}}

\renewcommand{\O}{\mathsf{O}}

\newcommand{\I}{\mathsf{I}}

\newcommand{\X}{\mathsf{X}}
\newcommand{\Y}{\mathsf{Y}}

\newcommand{\bX}{\mybox{\X}}

\newcommand{\bY}{\mybox{\Y}}
\newcommand{\bA}{\mybox{\A}}
\newcommand{\bB}{\mybox{\B}}
\newcommand{\bXY}{\mybox{\X\Y}}
\newcommand{\bAXY}{\mybox{\A\X\Y}}
\newcommand{\bBXY}{\mybox{\B\X\Y}}
\newcommand{\bABXY}{\mybox{\A\B\X\Y}}

\renewcommand{\S}{\mathsf{S}}

\newcommand{\fC}{\mathcal{C}}

\newcommand{\fF}{\mathcal{F}}

\newcommand{\fL}{\mathcal{L}}

\newcommand{\fT}{\mathcal{T}}
\newcommand{\fZ}{\mathcal{Z}}

\newcommand{\bbempty}{\mybox{\phantom{A}}}
\newcommand{\bbX}{\mybox{\X}}
\newcommand{\bbY}{\mybox{\Y}}
\newcommand{\bbAX}{\mybox{\A\X}}
\newcommand{\bbBY}{\mybox{\B\Y}}
\newcommand{\bbXY}{\mybox{\X\Y}}
\newcommand{\bbXYp}{\mybox{\X\Y+}}
\newcommand{\bbAXYp}{\mybox{\A\X\Y+}}
\newcommand{\bbBXYp}{\mybox{\B\X\Y+}}
\newcommand{\bbABXYp}{\mybox{\A\B\X\Y+}}

\newcommand{\bbC}{\mybox{\C}}
\newcommand{\bbCX}{\mybox{\C\X}}
\newcommand{\bbCY}{\mybox{\C\Y}}
\newcommand{\bbACX}{\mybox{\A\C\X}}
\newcommand{\bbBCY}{\mybox{\B\C\Y}}
\newcommand{\bbCXY}{\mybox{\C\X\Y}}
\newcommand{\bbCXYp}{\mybox{\C\X\Y+}}
\newcommand{\bbACXYp}{\mybox{\A\C\X\Y+}}
\newcommand{\bbBCXYp}{\mybox{\B\C\X\Y+}}
\newcommand{\bbABCXYp}{\mybox{\A\B\C\X\Y+}}

\DeclareMathOperator{\re}{\mathcal R}
\DeclareMathOperator{\rere}{\overline{\mathcal R}}

\DeclareMathOperator{\rank}{\operatorname{rank}}

\newcommand{\nodeconst}{\ensuremath{\mathcal{N}}}
\newcommand{\edgeconst}{\ensuremath{\mathcal{E}}}
\newcommand{\gen}[1]{\langle #1 \rangle}

\DeclareMathOperator{\poly}{poly}

\newcommand{\LOCAL}{\ensuremath{\mathsf{LOCAL}}\xspace}

\newcommand{\set}[1]{\left\{#1\right\}}

\newcommand{\ccs}{\mathfrak{C}}
\newcommand{\ccc}{\mathcal{C}}

\newcommand{\diaginf}[2]{\mathrm{Inf}(#1,#2)}
\newcommand{\diagsup}[2]{\mathrm{Sup}(#1,#2)}

\newcommand{\fpprocname}[0]{\mathrm{SubFP}}
\newcommand{\fpproc}[2]{\fpprocname(#1,#2)}
\newcommand{\fpp}{\mathrm{FixedPoint}}

\newenvironment{mycover}
{\list{}{\listparindent 0pt
        \itemindent    \listparindent
        \leftmargin    1cm
        \rightmargin   1cm
        \parsep        0pt}%
    \raggedright
    \item\relax}
{\endlist}

\newcommand{\myemail}[1]{\,$\cdot$\, {\small #1}}
\newcommand{\myaff}[1]{\,$\cdot$\, {\small #1}\par\smallskip}

\definecolor{darkgreen}{rgb}{0,0.5,0}
\definecolor{darkred}{rgb}{0.4,0,0}
\hypersetup{
    colorlinks=true,
    linkcolor=darkred,
    citecolor=darkgreen,
    filecolor=black,
    urlcolor=[rgb]{0,0.1,0.5},
    pdftitle={Towards Fully Automatic Distributed Lower Bounds},
    pdfauthor={Alkida Balliu, Sebastian Brandt, Fabian Kuhn, Dennis Olivetti, Joonatan Saarhelo}
}

\begin{document}

\setcounter{page}{0}
\thispagestyle{empty}

\begin{mycover}
    {\huge\bfseries\boldmath Towards Fully Automatic Distributed Lower Bounds \par}
    \bigskip
    \bigskip
    \bigskip

    \textbf{Alkida Balliu}
    \myemail{alkida.balliu@gssi.it}
    \myaff{Gran Sasso Science Institute}

    \textbf{Sebastian Brandt}
    \myemail{brandt@cispa.de}
    \myaff{CISPA Helmholtz Center for Information Security}
    
    \textbf{Fabian Kuhn}
    \myemail{kuhn@cs.uni-freiburg.de}
    \myaff{University of Freiburg}

    \textbf{Dennis Olivetti}
    \myemail{dennis.olivetti@gssi.it}
    \myaff{Gran Sasso Science Institute}
    
    \textbf{Joonatan Saarhelo}
    \myemail{joon.saar@gmail.com}
    \myaff{Unaffiliated}

\end{mycover}
\bigskip

\begin{abstract}
	In the past few years, a successful line of research has lead
        to lower bounds for several fundamental local graph problems
        in the distributed setting. These results were obtained via a
        technique called \emph{round elimination}. On a high level,
        the round elimination technique can be seen as a recursive
        application of a function that takes as input a problem $\Pi$
        and outputs a problem $\Pi'$ that is one round easier than
        $\Pi$. Applying this function recursively to concrete problems
        of interest can be highly nontrivial, which is one of the reasons that has made the technique difficult to approach. The contribution of our paper is threefold. 
	
	Firstly, we develop a new and fully automatic method for finding so-called \emph{fixed point relaxations} under round elimination. The detection of a non-$0$-round solvable fixed point relaxation of a problem
        $\Pi$ immediately implies lower bounds of $\Omega(\log_\Delta
        n)$ and $\Omega(\log_\Delta \log n)$ rounds for deterministic
        and randomized algorithms for $\Pi$, respectively.

	Secondly, we show that this automatic method is indeed useful, by obtaining lower bounds for defective coloring problems. More precisely, as an application of our procedure, we show that the problem of coloring the nodes of a graph with $3$ colors and defect at most $(\Delta - 3)/2$ requires  $\Omega(\log_\Delta n)$ rounds for deterministic algorithms and  $\Omega(\log_\Delta \log n)$ rounds for randomized ones.
        Additionally, we provide a simplified proof for an existing defective coloring lower bound.
        We note that lower bounds for coloring problems are notoriously challenging to obtain, both in general, and via the round elimination technique.

	Both the first and (indirectly) the second contribution build on our third contribution---a new and conceptually simple way to compute the one-round easier problem $\Pi'$ in the round elimination framework. This new procedure provides a clear and easy recipe for applying round elimination, thereby making a substantial step towards the greater goal of having a fully automatic procedure for obtaining lower bounds in the distributed setting.
\end{abstract}

\clearpage
\setcounter{page}{0}
\thispagestyle{empty}

\tableofcontents
\clearpage

\section{Introduction}

In the standard setting of distributed graph algorithms, known as the \LOCAL model~\cite{Linial1992,Peleg2000}, the nodes $V$ of a graph $G=(V,E)$ communicate over the edges $E$ of $G$ in synchronous rounds. Initially, the nodes do not know anything about $G$ (except for their own unique identifier and possibly some global parameters such as the number of nodes $n$ or the maximum degree $\Delta$) and at the end, each node must output its local part of the solution for the graph problem that needs to be solved. For example, if we intend to compute a vertex coloring of $G$, at the end, every node must output its own color in the final coloring. The time complexity of such a distributed algorithm is then measured as the number of rounds needed from the start until all nodes have terminated.

The study of the complexity of solving graph problems in the \LOCAL model and in related distributed models has been a highly active area of research with a variety of substantial results over the last years. Apart from very significant and insightful new algorithmic results for distributed graph problems (e.g., \cite{chang16exponential,ChangLP18,FOCS18-derand,Rozhon2020,FGGKR23,GGHIR23}), the last ten years in particular also brought astonishing progress on proving lower bounds for distributed graph problems in the \LOCAL model (e.g., \cite{Brandt2016,chang16exponential,Balliu2019,hideandseek}). Essentially all of this recent progress on lower bounds has been obtained by a technique known as \emph{round elimination}. The technique works for a class of problems known as \emph{locally checkable} problems~\cite{NaorS95,Brandt2019}, which encompasses many of the most fundamental problems studied in the context of the \LOCAL model.

\paragraph{Round Elimination.}
On a very high level, round elimination works as follows. Given a problem $\Pi$ provided in the proper language, the round elimination framework provides a way to mechanically construct a problem $\Pi'=\hat{\re}(\Pi)$ that is exactly one round easier (under some mild assumptions). That is, if $\Pi$ can be solved in $R$ rounds, then $\Pi'$ can be solved in $R-1$ rounds (and vice versa).\footnote{Formally, round elimination has to be performed on a weaker version of the \LOCAL model, which is known as the port numbering model. In the port numbering model, nodes do not have unique IDs, but they can distinguish their neighbors through different port numbers. Round elimination lower bounds in the port numbering model can then be lifted to lower bounds in the standard \LOCAL model~\cite{hideandseek}.} For proving an $R$-round lower bound on problem $\Pi$, one then has to show that the problem $\hat{\re}^{(R-1)}(\Pi)$, or a relaxation of it, is not trivial, i.e., it cannot be solved in $0$ rounds.

In its modern form, round elimination has first been used to show that the problems of computing a sinkless edge orientation or a $\Delta$-vertex coloring of $G$ require $\Omega(\log\log n)$ rounds with randomization and $\Omega(\log n)$ rounds deterministically~\cite{Brandt2016,chang16exponential}.\footnote{We remark that although phrased differently, the classic proofs that $3$-coloring a ring requires $\Omega(\log^* n)$ rounds~\cite{Naor1991,Linial1992} can also be seen as round elimination proofs.} Subsequently Brandt~\cite{Brandt2019} showed that round elimination can be applied to essentially every locally checkable problem and if a problem $\Pi$ is specified in the right language, the problem $\hat{\re}(\Pi)$ can be computed in a fully automatic way. \emph{Automatic round elimination} in the following lead to a plethora of new distributed lower bounds. We next list some of the highlights. In \cite{Balliu2019}, it was shown that even in regular trees, computing a maximal matching requires $\Omega(\min\set{\Delta, \log_\Delta\log n})$ rounds with randomized algorithms and $\Omega(\min\set{\Delta, \log_\Delta n})$ rounds with deterministic algorithms. Previously, the best known lower bound as a function of $\Delta$ for this problem was only $\Omega(\log\Delta/\log\log\Delta)$~\cite{KuhnMW16}. By a simple reduction, the same lower bound as for maximal matching also holds for computing a maximal independent set (MIS). In later work, the same lower bound was also proven directly for the MIS problem on trees and it was generalized in particular to the problems of computing ruling sets and of computing maximal matchings in hypergraphs, leading to tight (as a function of $\Delta$) lower bounds for those problems~\cite{balliurules,BBKOmis,hideandseek,Balliu0KO23}.

While round elimination has been extremely successful for proving many new lower bounds for computing locally checkable graph problems, the method has so far not been able to provide new lower bounds for many of the standard variants of \emph{distributed graph coloring} and thus for some of the most important and most well-studied locally checkable problems. When applying round elimination to standard $(\Delta+1)$-coloring and related graph coloring problems, the descriptions of the problems in the sequence obtained by applying $\hat{\re}(\cdot)$ iteratively grow doubly exponential in \emph{each} round elimination step (i.e., with each application of $\hat{\re}(\cdot)$) and thus even the one round easier problem $\hat{\re}(\Pi)$ often becomes too complex to understand.
We emphasize that all the recent progress on developing new lower bounds for locally checkable problems in the \LOCAL model has only been possible because the work of Brandt~\cite{Brandt2019} describes an \emph{automatic and generic} way to turn any locally checkable problem (given in the right formalism) into a locally checkable problem that is exactly one round easier.
Moreover, for much of the progress, it was crucial that there exists efficient software as described by Olivetti in \cite{Olivetti2019} that can be used to apply round elimination to concrete locally checkable problems.
We are convinced that in order to continue the present success story, further developing the existing automatic techniques will be indispensible and the main objective of this paper is to provide more efficient and more powerful methods for finding distributed lower bounds in an automatic fashion.

\paragraph{Distributed Coloring.}
As a concrete application, we aim to make
progress towards obtaining lower bounds for distributed coloring problems.
To achieve this, we consider the problem of computing a \emph{$d$-defective $c$-coloring}.
For two parameters $c$ and $d$, a $d$-defective $c$-coloring of a graph $G=(V,E)$ is a partition of $V$ into $c$ color classes so that every node $v\in V$ has at most $d$ neighbors of the same color.
Such colorings have become an important tool in many recent distributed coloring algorithms~\cite{barenboim14distributed,BarenboimE10,BarenboimE11,barenboim16sublinear,BEG18,Kuhn20,Balliu0KO22a,BalliuKO20,FuchsK23}. In \cite{FuchsK23}, it is also argued that further progress on defective coloring algorithms might be key towards obtaining faster distributed $(\Delta+1)$-coloring algorithms and proving hardness results on distributed defective coloring algorithms might therefore also provide insights into understanding the hardness of the standard $(\Delta+1)$-coloring problem.
To obtain proper colorings, defective colorings are commonly used as a subroutine in a recursive manner and to obtain efficient coloring algorithms using few colors, it would be particularly convenient to have algorithms that efficiently compute defective colorings with $c$ colors and defect only $(1+o(1))\Delta/c$. Such defective colorings always exist~\cite{lovasz66} and efficient distributed algorithms for computing such colorings would immediately lead to faster $O(\Delta)$-coloring algorithms and potentially also to faster $(\Delta+1)$-coloring algorithms. In fact, a generalized variant of $(1+o(1))\Delta/2$-defective $2$-colorings of line graphs have recently been used in a breakthrough result that obtains the first $\poly\log\Delta +O(\log^* n)$-round algorithm for computing a $(2\Delta-1)$-edge coloring of a graph~\cite{Balliu0KO22a}.

In contrast, the best known algorithms for computing an $O(\Delta)$ or $(\Delta+1)$-vertex coloring require time polynomial in $\Delta$~\cite{barenboim16sublinear,fraigniaud16local,BEG18,MausTonoyan20}. For vertex coloring, it is already known that computing $(1+o(1))\Delta/2$-defective $2$-colorings requires $\Omega(\log n)$ rounds even in bounded-degree graphs~\cite{BalliuHLOS19}. This raises the important question whether an increased number of $c > 2$ colors can admit the desired efficient $(1+o(1))\Delta/c$-defective $c$-colorings. Already the case of $c = 3$ was wide open previous to our work and an important open problem in its own right: obtaining the desired efficient defective coloring algorithm for $c = 3$ would have fundamental consequences by improving the complexity of $O(\Delta)$-coloring (and of $(\Delta + 1)$-coloring if extendable to list defective colorings~\cite{FuchsK23}), while proving a substantial lower bound for any such algorithm might pave the way for proving similar lower bounds for larger $c$ in the future. As one of the main technical results of this paper, we show that computing $(1+o(1))\Delta/3$-defective colorings (and in fact $(1-o(1))\Delta/2$-defective colorings) with $3$ colors requires $\Omega(\log n)$ rounds. We conjecture that a similar result should also hold for more than $3$ colors and we hope that such a result can be proven by extending the techniques that we introduce in this paper.

\subsection{Our Contributions}
\label{sec:ourcongen}

In the present paper, we take the task of automating round elimination and thus automating the search for distributed lower bounds one step further. In the following, we provide a high-level discussion of the contributions of the paper.

\subsubsection{An Automatic Way of Generating Round Elimination Fixed Points}\label{sec:contribution:fpprocedure}

Chang, Kopelowitz, and Pettie~\cite{chang16exponential} showed that in the \LOCAL model, every locally checkable problem $\Pi$ can either be solved deterministically in $f(\Delta) \cdot O(\log^* n)$ rounds (for some function $f(\cdot)$) or $\Pi$ has a deterministic $\Omega(\log_\Delta n)$ and a randomized $\Omega(\log_\Delta\log n)$ lower bounds. In the following, we call problems of the first type easy problems and problems of the second type hard problems.

\paragraph{Fixed Points Imply Hardness Results.}
A particularly elegant way to prove that a problem is of the second type is through round elimination fixed points. A locally checkable problem $\Pi$ is called a round elimination fixed point if $\hat{\re}(\Pi)=\Pi$, i.e., if the problem that is ``one round easier'' than $\Pi$ is $\Pi$ itself. We say that a problem $\Pi$ is a \emph{non-trivial fixed point} if $\Pi$ is a round elimination fixed point that cannot be solved in $0$ rounds. If a problem $\Pi$ is a non-trivial fixed point, existing standard techniques directly imply that $\Pi$ is a hard problem, i.e., that any deterministic \LOCAL algorithm to solve $\Pi$ requires at least $\Omega(\log_\Delta n)$ rounds and every randomized such algorithm requires at least $\Omega(\log_\Delta\log n)$ rounds (see, e.g., \cite{hideandseek}).
Moreover, we obtain the same lower bounds for $\Pi$ if $\Pi$ is not a fixed point itself but can be relaxed to a non-trivial fixed point $\tilde{\Pi}$.
In fact, while interesting problems exist that are non-trivial fixed points themselves (see, e.g., \cite{Brandt2016}), finding a non-trivial \emph{fixed point relaxation} $\tilde{\Pi}$ for $\Pi$ (which we may simply call a \emph{fixed point for $\Pi$}) is a more common way to prove lower bounds for a given problem $\Pi$ (see, e.g., \cite{binary,hideandseek,Balliu0KO23}).
Furthermore, as shown in~\cite{hideandseek,Balliu0KO23}, surprisingly, fixed points can also be used to prove lower bounds on the $\Delta$-dependency of easy problems, i.e., problems that can be solved in time $f(\Delta) \cdot O(\log^* n)$.

\paragraph{Fixed Points Can Be Large.}
In order to understand the distributed complexity of locally checkable problems, we therefore need methods to find non-trivial fixed points for such problems in case such fixed points exist. We argue that, similarly to performing and analyzing round elimination, also finding new fixed points will in many cases require some automated support for searching for fixed points. Note that in general, even for a relatively simple problem $\Pi$ with a small description, the smallest fixed point relaxation $\tilde{\Pi}$ of $\Pi$ might be much more complex and have a much larger description than the original problem $\Pi$. Consider for example the $\Delta$-coloring problem in $\Delta$-regular graphs. While the problem itself can be described\footnote{For an introduction to the description of problems, see \Cref{sec:moreround} or \Cref{sec:problems}.} with $\Delta$ different labels and $\Delta$ different node configurations (one for each possible color), the round elimination fixed point for $\Delta$-coloring that has been described in \cite{hideandseek} consists of $2^\Delta$ different labels and $2^\Delta-1$ different node configurations (and no smaller fixed point for $\Delta$-coloring is known or suspected to exist). Finding fixed points for problems that are not as symmetric and not as well-behaved as $\Delta$-coloring might quickly become infeasible when it has to be done by hand, even when using the support of existing software for performing single round elimination steps.

\paragraph{Our Contribution: a Procedure for Finding Fixed Points Automatically.}
As our first main contribution, we provide a method to automatically
generate relaxations $\tilde{\Pi}$ of a given locally checkable
problem $\Pi$ that are fixed points under the round elimination
framework. As input, the method takes a problem $\Pi$ and an extended
label set $\tilde{\Sigma}$ that satisfies $\Sigma\subseteq \tilde{\Sigma}$,
where $\Sigma$ is the set of labels of $\Pi$. In addition, the method
uses a diagram $D$ that determines certain relations between the labels in
$\tilde{\Sigma}$. Formally, $D$ is a directed acyclic graph with node
set $\tilde{\Sigma}$ and with certain additional properties. Based on
the original problem $\Pi$ and the diagram $D$, the problem
$\tilde{\Pi}$ is obtained in a way that is very similar to a novel
way of performing round elimination that we outline in \Cref{sec:moreround} and formally introduce in
\Cref{sec:newre}. Whether the generated fixed point $\tilde{\Pi}$ is
non-trivial (i.e., whether $\tilde{\Pi}$ is not $0$-round-solvable)
can depend on the diagram $D$ that we use. We introduce and formally
analyze our fixed point generation method in \Cref{sec:procedure} and
we discuss ways to select a good diagram for the method in
\Cref{sec:diagram}. 

\paragraph{Our Contribution: a First Simple Application of Our Procedure.}
As a first direct application we get a simpler
proof of a result of \cite{hideandseek}: By applying our method to the
$\Delta$-coloring problem, together with a simple diagram (which is
basically the Hasse diagram of the power set of the $\Delta$ labels of
$\Delta$-coloring), we directly get the $\Delta$-coloring fixed point
that was presented in \cite{hideandseek}.

\subsubsection{Lower Bounds for Defective Coloring Problems}\label{sec:contribution:defective}

As explained in the introduction, understanding whether $(1+o(1))\Delta/c$-defective $c$-coloring is an easy or a hard problem is of fundamental importance, since the complexity of such a problem may have direct implications on the complexity of $(\Delta+1)$-coloring, which is a major open question in the field.
As a more involved application of our fixed point generation method, we develop lower bounds for \emph{defective coloring problems}.

\paragraph{Our Contribution: Defective $2$-Coloring.}
Not many bounds on the complexity of defective colorings are known (we discuss known bounds in \Cref{sec:related}). An exception is the case of defective colorings with $2$ colors, which is understood. By computing an MIS (which can be done in $O(\Delta+\log^* n)$ rounds~\cite{barenboim14distributed}) and assigning the MIS nodes one of the colors and the remaining nodes the other color, one obtains a $(\Delta-1)$-defective $2$-coloring of the graph. Interestingly, the problem becomes hard if we try to just go one step further: in \cite{BalliuHLOS19}, it was shown that computing a $(\Delta-2)$-defective $2$-coloring is a hard problem. This result has been shown via a reduction from the hardness of sinkless orientation. However, this reduction is based on the construction of virtual graphs on which the defective coloring algorithm is executed in order to obtain a sinkless orientation on the original graph, and in particular the lower bounds are not proved by providing a non-trivial fixed point. As a second application of our fixed point generation method, we show the following.

\vspace{0.2cm}
\begin{mdframed}[innertopmargin=8pt,innerbottommargin=5pt,innermargin=5pt]
	\center
	There exists a non-trivial fixed point relaxation for $(\Delta-2)$-defective $2$-coloring.
\end{mdframed}
\vspace{0.1cm}

This result is significant in light of the fundamental open question stated in \cite{trulytight,hideandseek} asking whether, for every locally checkable problem $\Pi$ that has a deterministic $\Omega(\log_\Delta n)$ and a randomized $\Omega(\log_\Delta\log n)$ lower bound, such a lower bound can be proven via a round elimination fixed point. The $(\Delta-2)$-defective $2$-coloring problem was one of an only very small number of such problems for which previously no fixed point lower bound proof was known.

\paragraph{Our Contribution: Defective $3$-Coloring.}
As a main application of our automatic fixed point procedure, we study the defective coloring problem with $3$ colors. From the arbdefective coloring lower bound of \cite{hideandseek}, it is known that $d$-defective $3$-coloring is hard if $3(d+1)\leq \Delta$ and thus if $d\leq \frac{\Delta}{3}-1$. In \cite{BalliuHLOS19}, it was further shown that if $d\geq \frac{2\Delta-4}{3}$, $d$-defective $3$-coloring can be solved in $O(\Delta+\log^* n)$ rounds. By using our fixed point method, we manage to partially close this gap by proving the following statement (cf.~\Cref{thm:3col}).

\vspace{0.2cm}
\begin{mdframed}[innertopmargin=8pt,innerbottommargin=5pt,innermargin=5pt]
	For $d\leq\frac{\Delta - 3}{2}$, the $d$-defective $3$-coloring problem is a hard problem, i.e., it requires $\Omega(\log_\Delta n)$ rounds deterministically and $\Omega(\log_\Delta\log n)$ rounds with randomization.
\end{mdframed}
\vspace{0.1cm}

This in particular his implies that there is no $(1 + o(1))\Delta/3$-defective $3$-coloring algorithm that violating those time lower bounds, thereby ruling out the possibility of using defective $3$-coloring as an approach for attacking $O(\Delta)$ and $(\Delta + 1)$-coloring in the manner outlined before \Cref{sec:ourcongen}.

We note that the fixed point that we automatically generate for this problem is highly non-trivial, and that manually proving that the fixed point that we provide is indeed a fixed point would require to perform a case analysis over hundreds of cases. For this reason,  we do not manually prove that the fixed point that we provide is indeed a fixed point. Instead, we provide a way to automate this process, by reducing the problem of determining whether a problem is a fixed point to the problem of proving that certain systems of inequalities have no solution.
The remaining task of showing that said systems have no solution can be performed automatically via computer tools. This automatization process provides a partial answer to Open Question 9 in \cite{hideandseek}.
The details appear in \Cref{sec:3col}.

\subsubsection{A More Efficient Method for Performing Round Elimination}
\label{sec:moreround}

The procedure for finding fixed points automatically mentioned in \Cref{sec:contribution:fpprocedure} is based on a novel way for applying the round elimination technique. More in detail, such a result is obtained as follows. We first provide a novel way for applying round elimination, that is, a novel way for computing a locally checkable problem $\Pi'$ that is exactly one round easier than $\Pi$. Then, we show that, by applying such a procedure in a slightly modified way, instead of obtaining the problem $\Pi'$, we obtain some problem $\tilde{\Pi}$ which is guaranteed to be a fixed point relaxation of $\Pi$. While in some cases the obtained problem $\tilde{\Pi}$ may be solvable in $0$ rounds (i.e., this \emph{must} be the case when applying the procedure on an \emph{easy} problem), the results presented in \Cref{sec:contribution:defective} are obtained by proving that the fixed points that we get by applying the procedure on defective colorings are non-trivial.

While our new procedure for applying the round elimination technique has applications for finding fixed points, this procedure is interesting on its own. In order to better explain the reason, we first highlight the main issue of the standard way of applying round elimination.
While for a given locally checkable problem $\Pi$, the framework of Brandt~\cite{Brandt2019} gives a fully automatic way for computing a locally checkable problem $\Pi'$ that is exactly one round easier than $\Pi$, this computation is in general not computationally efficient. To illustrate why, we somewhat informally sketch how round elimination works (for a formal description we refer to \Cref{ssec:re}).

\paragraph{How Round Elimination Works.}For the automatic round elimination framework, a locally checkable problem on a $\Delta$-regular graph $G=(V,E)$ is formalized on the bipartite graph $H$ between the nodes $V$ and the edges $E$ of $G$.\footnote{More generally, round elimination can be defined on biregular bipartite graphs or hypergraphs (see \Cref{sec:preliminaries}).} That is, $H$ is obtained by adding an additional node in the middle of the edges in $E$. Each edge of $G$ is thus split into $2$ halfedges. A solution to a locally checkable problem is given by an assignment of labels from a finite alphabet $\Sigma$ to all edges of $H$ (i.e., to each halfedge of $G$). The validity of a solution is given by a set of allowed node and edge configurations, where a node configuration is a multiset of labels of size $\Delta$ and an edge configuration is a multiset of labels of size $2$. One step of round elimination on $G$ is done by performing two steps of round elimination on $H$ (note that one round on $G$ corresponds to two rounds on $H$). When starting from a node-centric problem $\Pi$ (i.e., a problem where the nodes in $H$ corresponding to \emph{nodes} in $G$ assign the labels to their incident half-edges), the first step transforms $\Pi$ into an edge-centric problem $\Pi'$ that is exactly one round easier on $H$ and the second step transforms the problem into a node-centric problem $\Pi''$ that is one round easier than $\Pi'$ on $H$ and thus one round easier than $\Pi$ on $G$. The label set $\Sigma'$ of $\Pi'$ is the power set $2^{\Sigma}$ of $\Sigma$ and the label set $\Sigma''$ of $\Pi''$ is the power set of $\Sigma'$.
The allowed edge configurations of $\Pi'$ are, roughly speaking, the multisets $\{\L_1,\L_2\}$ of labels $\L_1,\L_2\in \Sigma'=2^\Sigma$ such that \emph{for all} $\ell_1\in \L_1$ and $\ell_2\in \L_2$, $\{\ell_1,\ell_2\}$ is an allowed edge configuration of $\Pi$.\footnote{In the formally precise definition of the set of edge configurations of $\Pi'$ provided in \Cref{ssec:re}, we'll refine this definition slightly.}
The allowed node configurations of $\Pi'$ are all the multisets $\{\L_1,\dots,\L_{\Delta}\}$ of labels $\L_i\in \Sigma'$ (that appear in some allowed edge configuration of $\Pi'$) such that \emph{there exists} an allowed node configuration $\{\ell_1,\dots,\ell_\Delta\}$ with $\ell_i\in \L_i$ in problem $\Pi$. In the second step, $\Pi''$ is obtained in the same way from $\Pi'$, but by exchanging the roles of nodes and edges. That is, in the second step, the ``for all'' quantifier is applied to the allowed node configurations and the ``exists'' quantifier is applied to the allowed edge configurations (of $\Pi'$).

\paragraph{The Computationally Expensive Part.}
Note that from a computational point of view, it is mainly the
application of the ``for all'' quantifier on the edge side when going
from $\Pi$ to $\Pi'$ and even more importantly on the node side when
going from $\Pi'$ to $\Pi''$ that is challenging. When implemented
naively, one has to iterate over all possible size-$2$ multisets of
$\Sigma'$ in the first step and over all possible size-$\Delta$
multisets of $\Sigma''$ in the second step. While in general, the
problem $\Pi''$ that is one round easier than the original problem
$\Pi$ on $G$ can be doubly exponentially larger than $\Pi$, for
interesting problems this is often not the case. For such more
well-behaved problems, the ``for all'' case can potentially be
computed in a much more efficient way. 

\paragraph{Our Contribution.}As our final contribution, we
give a new elegant way to perform the application of the ``for all''
quantifier in round elimination. The method makes use of the fact that
often the node and edge configurations of a problem can be represented
by a relatively small number of \emph{condensed configurations}. A
condensed node or edge configuration is a multiset
$\{S_1,\dots,S_k\}$  (where $k=\Delta$ for nodes and $k=2$ for edges)
of sets $S_1, \dots, S_k\subseteq\Sigma$ of labels, representing the
set of all configurations $\{\ell_1,\dots,\ell_k\}$ for which
$\ell_i\in S_i$ for all $i\in\{1,\dots,k\}$. We prove that the ``for
all'' part of round elimination can be performed by a simple process
that consists of steps of the following kind. In each step, we take
two condensed configurations of the current problem and we combine
those condensed configurations in some way to generate new condensed
configurations. We then remove redundant configurations and continue
until such a step cannot generate any new condensed configurations. In
the end, each condensed configuration $\{S_1,\dots,S_k\}$ is
interpreted as a multiset of labels of the new problem. We formally
define the process and prove its correctness in \Cref{sec:newre}.

Informally, we prove that each configuration of the resulting problem can be described as a binary tree, where leaf nodes are condensed configurations of the original problem, and each internal node of the tree is the configuration obtained by combining its two children. 

Since our new procedure is mainly used as a tool for obtaining fixed points, we do not formally state the benefits of this new procedure. However, we informally highlight the following:
\begin{itemize}
	\item The new procedure avoids the cost of enumerating all possible size-$\Delta$ multisets of $\Sigma''$ in the second step, and its running time only depends on the number of input configurations, output configurations, and the height of the aforementioned trees. Such trees have height at most $\Delta \cdot |\Sigma'|$, and we observe that, for many natural problems, the height is much smaller. We thus obtain that, for many problems of interest, the running time of the new procedure is output-sensitive. 
	\item Thanks to the new procedure, we obtain that, in order to check whether a problem is a fixed point, it is sufficient to check whether the combination of pairs of condensed configurations does not create new configurations. While this drastically reduces the time complexity of checking whether a problem is a fixed point, this also makes it much easier to \emph{prove} that a problem is a fixed point. In fact, in the latter case, it is sufficient to consider two configurations at a time, instead of going through an exponential number of cases. We point out that, even if some friendly oracle gave us the fixed point for defective $3$-coloring presented in \Cref{sec:3col}, we believe that, without exploiting this new procedure, proving that such a problem is indeed a fixed point would not have been possible.
\end{itemize}

\subsection{Further Related Work.}\label{sec:related}
\paragraph{Existing Fixed Points.}
In \cite{Brandt2016,Brandt2019}, it is shown that if expressed in the right way, the problem of computing a sinkless orientation of the edges of a $\Delta$-regular graph (for $\Delta\geq 3$) is a non-trivial round elimination fixed point, which implies an $\Omega(\log_\Delta n)$ deterministic and an $\Omega(\log_\Delta\log n)$ randomized lower bound for the sinkless orientation problem.
An example for a problem that is not a fixed point itself but can be relaxed to a non-trivial fixed point is the $\Delta$-coloring problem. While successively applying round elimination to the $\Delta$-coloring problem results in a sequence of problems whose descriptions get exponentially larger in \emph{each} step, it is shown in \cite{hideandseek} that there exists a problem $\tilde{\Pi}$ that contains $\Delta$-coloring (i.e., solving $\Delta$-coloring solves $\tilde{\Pi}$, but not vice versa) such that $\tilde{\Pi}$ is a non-trivial round elimination fixed point. Non-trival round elimination fixed point relaxations for other locally checkable graph problems have been obtained in \cite{binary,Balliu0KO23}. 

\paragraph{Using Fixed Points For Proving Lower Bounds as a Function of $\Delta$.}
In \cite{hideandseek,Balliu0KO23}, it is shown that fixed points can also be used to determine lower bounds on the $\Delta$-dependency of problems that can be solved in time $f(\Delta) \cdot O(\log^* n)$. For example, when applying round elimination to the maximal independent set (MIS) problem, one essentially obtains problems that consist of an MIS on a part of the graph and a coloring with a certain number of colors on the remainder of the graph. However, as even the respective coloring problem alone grows exponentially in each round elimination step, the same is true for MIS. It is shown in \cite{hideandseek} that if one relaxes the problem sequence such that the problems in the sequence essentially consist of an MIS on one part of the graph and a fixed point relaxation of the coloring problem on the remainder of the graph, then one obtains a problem sequence that becomes manageable and that can be used to obtain tight (as a function of $\Delta$) lower bounds for MIS and also for many more general problems.

\paragraph{Defective Coloring.}
Not only we do not know the complexity of $d$-defective $c$-coloring for most of the values of $c$ and $d$, but also we do not even know in which cases it is an easy problem (i.e., it can be solved in $f(\Delta) \cdot O(\log^* n)$ rounds) and in which cases it is a hard problem (i.e., deterministic algorithms require $\Omega(\log_\Delta n)$ rounds and randomized algorithms require $\Omega(\log_\Delta\log n)$ rounds). It is known that a $d$-defective $O\big(\big(\frac{\Delta}{d+1}\big)^2\big)$-coloring can be computed in $O(\log^* n)$ rounds (with no additional dependency on $\Delta$)~\cite{Kuhn2009}. By using a simpler version of an algorithm described in \cite{BalliuHLOS19}, it is further possible to compute a $d$-defective $p^2$-coloring in $O(\Delta+\log^* n)$ rounds as long as $(d+1)p>\Delta$. For $d$-arbdefective $c$-coloring, which is a relaxation of $d$-defective $c$-coloring, it is further known that the problem is easy if and only if $c(d+1)>\Delta$~\cite{hideandseek}. This in particular implies that $d$-defective $c$-coloring is a hard problem if $c(d+1)\leq \Delta$. We therefore know that the problem is hard if the number of colors is at most $\frac{\Delta}{d+1}$ and that it is easy if the number of colors is more than $\big(\frac{\Delta}{d+1}\big)^2$.

\section{Road Map}\label{sec:roadmap}
\paragraph{Preliminaries.} In \Cref{sec:preliminaries}, we provide some preliminaries. We first define the model of computation and the language that we use to formally describe problems. Then, we describe the round elimination framework.

\paragraph{A new way of applying round elimination.}
On a high level, round elimination allows us to start from a problem $\Pi$ and to compute a problem $\Pi'$ that, under some assumptions, is exactly one round easier (in the distributed setting) than $\Pi$. As it will become clear in \Cref{sec:preliminaries}, computing $\Pi'$ as a function of $\Pi$ can be a tricky process. 
In \Cref{sec:newre} we provide a novel and simplified way to compute $\Pi'$ as a function of $\Pi$.

\paragraph{Fixed point generation.}
A problem $\Pi'$ is a non-trivial fixed point relaxation of $\Pi$ if it satisfies the following:
\begin{itemize}
	\item $\Pi'$ can be solved in $0$ rounds if we are given a solution for $\Pi$;
	\item $\Pi'$ cannot be solved in $0$ rounds in the so-called port numbering model (see \Cref{sec:preliminaries} for the definition of this model);
	\item By applying round elimination on $\Pi'$, we obtain $\Pi'$ itself.
\end{itemize}
It is known by prior work (see \Cref{thm:lifting}) that, if there exists a non-trivial fixed point relaxation for a problem $\Pi$, then $\Pi$ requires $\Omega(\log_\Delta n)$ rounds for deterministic algorithms and $\Omega(\log_\Delta \log n)$ rounds for randomized ones. Finding non-trivial fixed point relaxations is one of the very few ways that we have to prove such lower bounds. In \Cref{sec:procedure}, we provide an automatic way to obtain non-trivial fixed point relaxations. More in detail, we provide a procedure $\fpp$ that takes in input a problem $\Pi$ and an object $D$ (called diagram), and it produces a problem $\Pi'$ that is always guaranteed to be a fixed point. Whether such a fixed point is non-trivial depends on $\Pi$ and on the choice of $D$.

\paragraph{Selecting the right diagram.}
As mentioned before, the choice of the diagram may affect the triviality of the obtained fixed point. 
In \Cref{sec:diagram}, we first provide a generic way to construct a diagram as a function of $\Pi$, that we call default diagram. 
Then, we show possible ways to modify the default diagram in the case in which the fixed-point obtained with the default diagram is a trivial one.

\paragraph{An alternative proof for the hardness of $\Delta$-coloring.}
In \Cref{sec:deltacol}, we show a first application of our fixed point procedure, by providing a non-trivial fixed point relaxation for the $\Delta$-coloring problem.
Such a fixed point was already shown in \cite{hideandseek}, but here we show a much easier proof.
 While this section is not the main contribution of our work, its main purpose is to warm-up the reader for what comes later.

\paragraph{An alternative proof for the hardness of defective $2$-coloring.}
In \Cref{sec:2col}, we show another application of our fixed point procedure, by providing a non-trivial fixed point relaxation for the $(\Delta-2)$-defective $2$-coloring problem. This is one of the few problems for which an $\Omega(\log_\Delta n)$ lower bound is known by prior work \cite{BalliuHLOS19}, but a non-trivial fixed point relaxation for this problem was unknown. Whether a non-trivial fixed point relaxation exists for all problems that require $\Omega(\log_\Delta n)$ deterministic rounds is one of the major open questions about round elimination, and hence in this section we make progress in understanding it. Again, this section is not the main contribution of our work, and its main purpose is to prepare the reader for what comes next.

\paragraph{Defective $3$-coloring.}
In \Cref{sec:3col}, we use our fixed point procedure to show a lower bound for defective $3$-coloring. 
While the proofs in \Cref{sec:deltacol} and \Cref{sec:2col} require a relatively short case analysis, the proof in \Cref{sec:3col} requires to analyze hundreds of cases.
For this reason, in this section, we prove that such a case analysis can be performed automatically by using computer tools. In particular, we reduce the task of checking whether a given problem $\Pi$ is the result of applying our fixed point procedure, to proving that all systems of inequalities belonging to a certain finite set have no solution, which can be checked automatically via computer tools.

\paragraph{Open questions.}We conclude, in \Cref{sec:open}, with some open questions.

\section{Preliminaries}\label{sec:preliminaries}
\subsection{The \boldmath\LOCAL Model}
The computational model that we
consider is the standard \LOCAL model of distributed
computing~\cite{Linial1992,Peleg2000}, where the nodes $V$ of a graph
$G=(V,E)$ communicate over the edges $E$. 
More precisely, time is divided into synchronous rounds, and in each round each node can send an arbitrarily large message to each neighbor.
Moreover, between sending messages, nodes can perform any internal computation on the information they gathered so far.
In the beginning of the computation, each node $v$ is aware of its own degree $\deg(v)$, and has an internal ordering of its incident edges represented by the \emph{ports} $1, \dots, \deg(v)$ being assigned bijectively to $v$'s incident edges.
We also assume that each node is aware of the number $n$ of nodes and the maximum degree $\Delta$ of the input graph.
As we will prove lower bounds in this work, this assumption makes our results only stronger.
Moreover, each node is equipped with some \emph{symmetry-breaking information} to avoid trivial impossibilities: in the case of deterministic algorithms, each node is assigned some globally unique ID of length $O(\log n)$ bits; in the case of randomized algorithms, each node instead has access to an unlimited amount of private random bits.
Each node executes the same algorithm that governs which messages a node sends (depending on the accumulated knowledge of the node) and what the node outputs at the end of the computation.
Each node has to terminate at some point and then provide a local output; all local outputs together form the global solution to the problem.
The (\emph{round} or \emph{time}) \emph{complexity} of a distributed algorithm is the number of rounds until the last node terminates. 
In the randomized setting, as usual, the algorithms are required to be Monte-Carlo algorithms that produce a correct solution with high probability, i.e., with probability at least $1-1/n$.

While the lower bounds we prove hold in the \LOCAL model, for technical reasons we will also make use of the \emph{port numbering} model along the way.
The (deterministic) port numbering model is the same as the deterministic \LOCAL model apart from two differences:
\begin{enumerate}
	\item No symmetry-breaking information is provided, i.e., nodes are not equipped with IDs.
	\item For each hyperedge $e$, a total order on the set of incident nodes is provided (which can be formalized via a bijection between this node set and the set $\{ 1, \dots, k \}$, where $k$ denotes the number of nodes contained in $e$). 
\end{enumerate}
The second difference can be seen as an analog (on the hyperedge side) of the port numbers via which the nodes can distinguish between incident hyperedges.

\subsection{Problems}
\label{sec:problems}
The problems we study in this work fall into the class of \emph{locally checkable problems}.
Locally checkable problems are problems that can be defined via local constraints and encompass the vast majority of problems studied in the \LOCAL model.
A modern formalism to define these problems is given by the so-called black-white formalism that we will also use in this paper.
In fact, as we will see, this formalism captures locally checkable problems not only on graphs, but more generally on hypergraphs (where we will denote the maximum number of nodes in a hyperedge by $\delta$).
Note that \Cref{ssec:exso} provides an example illustrating (some of) the definitions provided in this section.

\paragraph{The black-white formalism.}
In the black-white formalism, a locally checkable problem is given as a triple $\Pi = (\Sigma_\Pi, \nodeconst_\Pi, \edgeconst_\Pi)$.
Here, $\Sigma_\Pi$ is a finite set of elements, called \emph{labels}, $\nodeconst_\Pi = (\nodeconst_1, \dots, \nodeconst_\Delta)$ and $\edgeconst_\Pi = (\edgeconst_1, \dots, \edgeconst_\delta)$, where each $\nodeconst_i$ and $\edgeconst_i$ is a collection of multisets of cardinality $i$ with labels from $\Sigma_\Pi$.
We call $\nodeconst_\Pi$ the \emph{node constraint} of $\Pi$ and $\edgeconst_\Pi$ the \emph{edge constraint} of $\Pi$.
On a hypergraph, a correct solution for $\Pi$ is an assignment of labels from $\Sigma_\Pi$ to the incident node-hyperedge pairs such that for each node $v$, the multiset of labels corresponding to $v$ is contained in $\nodeconst_{\deg(v)}$, and analogously for hyperedges w.r.t.\ the respective $\edgeconst_i$.
More formally, let $\fF$ denote the set of pairs $(v,e)$ where $e$ is a hyperedge incident to $v$.
A correct solution for $\Pi$ on a hypergraph $G = (V,E)$ is a mapping $f \colon \fF \rightarrow \Sigma_\Pi$ such that, for each $v \in V$, we have $\{ f(v,e') \mid e' \ni v \} \in \nodeconst_{\deg(v)}$, and, for each $e \in E$, we have $\{ f(v',e) \mid v' \in e \} \in \edgeconst_{\rank(e)}$.
Here, the \emph{rank} $\rank(e)$ of a hyperedge $e$ is the number of nodes contained in $e$, and the displayed sets are to be understood as multisets.

When solving a locally checkable problem in the distributed setting, each node $v$ has to output one label for each ``incident'' node-hyperedge pair in $\fF$ such that the induced global solution is correct.
While the improvements for the general round elimination technique (discussed below) that we will obtain in this work apply to the general hypergraph setting, for the results about concrete problems that we provide we can restrict attention to the special case of graphs.
In this special case, each hyperedge is of rank $2$, and consequently we will replace the edge constraint $(\edgeconst_1, \dots, \edgeconst_\delta)$ by $\edgeconst_2$.
Moreover, to simplify notation, in this case, we will set $\edgeconst := \edgeconst_2$.

We remark that besides providing a formalism for graphs by considering them as a special case of hypergraphs, the black-white formalism provides a (different) way to encode and study problems on \emph{bipartite} graphs, by identifying the ``black'' nodes in the bipartition with the \emph{nodes} in the above formalism, and the ``white'' nodes with the \emph{hyperedges}.
This relation to bipartite graphs is also where the name ``black-white formalism'' comes from.

As can be observed, the definition of the problems in this formalism depends on $\Delta$ (and $\delta$), which provides the power to also describe important problems like $(\Delta + 1)$-coloring in this formalism.
If we are to be very precise, in this formalism each problem is a collection of problems indexed by $\Delta$ (and, if considered on hypergraphs, $\delta$).
Throughout the paper, we implicitly assume that some (arbitrary) $\Delta$ (and, if required, some $\delta$) is fixed.
Note that this does not impact the generality of our results.

Finally, we remark that, for simplicity, we consider two locally checkable problems given in the black-white formalism as identical if one can be obtained from the other by renaming the labels used to describe the latter.

\paragraph{Configurations.}
We will use the term \emph{configuration} to refer to a multiset of labels, and write it in either of the two equivalent forms $\{ \ell_1, \dots, \ell_i\}$ and $\ell_1 \dots \ell_i$.
Note that the order of the $\ell_j$ does not matter (also in the second form): all configurations that can be obtained from a configuration by reordering are considered to be the same configuration.
When referring to the multiset of labels assigned to the pairs $(v,e')$ incident to a fixed node $v$, we will use the term \emph{node configuration}; when referring to the multiset of labels assigned to the pairs $(v',e)$ corresponding to a fixed (hyper)edge $e$, we will use the term \emph{edge configuration}.
Moreover, for simplicity we may slightly abuse notation by writing $\{ \ell_1, \dots, \ell_i\} \in \L_1 \times \dots \times \L_i$ if $\L_1, \dots, \L_i$ are sets containing the labels $\ell_1, \dots, \ell_i$, respectively.

It will be convenient to refer to certain collections of configurations in a condensed manner.
A \emph{condensed configuration} $\fC$ is a configuration $\{ \L_1, \dots, \L_i \}$ of sets of labels.
Configuration $\fC$ is to be understood as the set of all configurations $\{\ell_1, \dots, \ell_i\} \in \L_1 \times \dots \times \L_i$ (though we will also consider the condensed configuration $\fC$ as a configuration of sets when convenient).
To indicate that a configuration of sets represents a condensed configuration, we will often write each set in the configuration in the form $[\ell_1 \s \dots \s \ell_j]$ (unless the set only contains one element $\ell$, in which case we will simply write the set as $\ell$).

\paragraph{Diagrams.}
A useful way of capturing certain aspects of problems is via so-called diagrams.
A \emph{diagram} $D = (\Sigma_D, E_D)$ is nothing else than a directed acyclic graph with node set $\Sigma_D$ and edge set $E_D$.
The \emph{edge diagram} of a problem $\Pi = (\Sigma_\Pi, \nodeconst_\Pi, \edgeconst_\Pi)$ is the diagram $D$ obtained by setting $\Sigma_D := \Sigma_\Pi$ and defining $E_D$ as the set of those directed edges $(\ell, \ell')$ that satisfy that $\ell' \neq \ell$ and, for every configuration $\{\ell_1, \dots, \ell_{\delta}\} \in \edgeconst_\Pi$ with $\ell_i = \ell$ for some $1 \leq i \leq \delta$, also $\{\ell_1, \dots, \ell_{i-1}, \ell', \ell_{i+1}, \dots, \ell_{\delta}\} \in \edgeconst_\Pi$.
When displaying a diagram, we often omit arrows that can be obtained as the composition of displayed arrows.
We call a subset $S \subseteq \Sigma_D$ \emph{right-closed} (w.r.t.\ $D$) if, for any edge $(\ell, \ell') \in E_D$, $\ell \in S$ implies $\ell' \in S$.

\subsection{The Round Elimination Technique}\label{ssec:re}
In this section, we give a formal introduction to round elimination.
As some of the definitions provided in this section are fairly technical, the reader is encouraged to consult the illustrating example provided in \Cref{ssec:exso} alongside reading the definitions.

For technical reasons, round elimination requires the considered input (hyper)graphs to be regular (and uniform).
As such, we will assume throughout the paper that every node of the input (hyper)graph has the same degree $\Delta$ and every (hyper)edge has the same rank $\delta$ (which, in the case of graphs, is simply $2$).
This also simplifies the representation of locally checkable problems $\Pi = (\Sigma_\Pi, \nodeconst_\Pi, \edgeconst_\Pi)$: now we can assume that $\nodeconst_\Pi$ and $\edgeconst_\Pi$ are collections of multisets of cardinalities $\Delta$ and $\delta$, respectively, instead of sequences of similar collections.
Note that, as we will prove lower bounds in this work, the inherent restriction to regular graphs makes our results only stronger.

\paragraph{$\re(\cdot)$ and $\rere(\cdot)$.}
At the heart of the round elimination technique lie the round elimination operators $\re$ and $\rere$, which are functions that take a locally checkable problem in the black-white formalism as input and return such a problem.
More precisely, for a locally checkable problem $\Pi = (\Sigma_\Pi, \nodeconst_\Pi, \edgeconst_\Pi)$, the locally checkable problem $\re(\Pi) = (\Sigma_{\re(\Pi)}, \nodeconst_{\re(\Pi)}, \edgeconst_{\re(\Pi)})$ is defined as follows.

The label set $\Sigma_{\re(\Pi)}$ of $\re(\Pi)$ is simply the set of non-empty subsets of $\Sigma_\Pi$, i.e., $\Sigma_{\re(\Pi)} := 2^{\Sigma_\Pi} \setminus \{\{\}\}$.
For the definition of the edge constraint $\edgeconst_{\re(\Pi)}$ of $\re(\Pi)$, we need the notion of a \emph{maximal configuration}.
Let $\fZ$ be a collection of configurations of sets of labels.
Then, a configuration $\L_1 \dots \L_i \in \fZ$ is maximal (in $\fZ$) if there is no configuration $\L'_1 \dots \L'_i \in \fZ$ (of the same length) such that there exists a bijection $\phi \colon \{ 1, \dots, i \} \rightarrow \{ 1, \dots, i \}$ satisfying $\L_j \subseteq \L'_{\phi(j)}$ for all $1 \leq j \leq i$ and $\L_j \subsetneq \L'_{\phi(j)}$ for at least one $1 \leq j \leq i$.
In other words, a configuration of sets is maximal if no other configuration in the considered configuration space can be reached by enlarging (some of) the sets (and reordering the sets).

Now we can define $\edgeconst_{\re(\Pi)}$ as follows.
Let $\edgeconst$ denote the collection of all configurations $\L_1 \dots \L_\delta$ such that $\L_1, \dots, \L_\delta \in \Sigma_{\re(\Pi)}$ and for all choices $(\ell_1, \dots, \ell_\delta) \in \L_1 \times \dots \times \L_\delta$ of labels we have $\{ \ell_1, \dots, \ell_\delta \} \in \edgeconst_\Pi$.
Then, $\edgeconst_{\re(\Pi)}$ is obtained from $\edgeconst$ by removing all configurations that are not \emph{maximal} in $\edgeconst$.
Finally, the node constraint $\nodeconst_{\re(\Pi)}$ of $\re(\Pi)$ is defined as the collection of all configurations $\L_1 \dots \L_\Delta$ such that each $\L_i$ appears in at least one configuration from $\edgeconst_{\re(\Pi)}$ and there exists a choice $(\ell_1, \dots, \ell_\Delta) \in \L_1 \times \dots \times \L_\Delta$ of labels satisfying $\{ \ell_1, \dots, \ell_\Delta \} \in \nodeconst_\Pi$.

The problem $\rere(\Pi) =  (\Sigma_{\rere(\Pi)}, \nodeconst_{\rere(\Pi)}, \edgeconst_{\rere(\Pi)})$ is defined dually to $\re(\Pi)$, where the role of nodes and hyperedges are reversed.
More precisely, we have the following.
As before, $\Sigma_{\rere(\Pi)} = \Sigma_{\re(\Pi)} = 2^{\Sigma_{\Pi}}  \setminus \{\{\}\}$.
The node constraint $\nodeconst_{\rere(\Pi)}$ of $\rere(\Pi)$ is the collection of \emph{maximal} configurations $\L_1 \dots \L_\Delta$ such that $\L_1, \dots, \L_\Delta \in \Sigma_{\rere(\Pi)}$ and for all choices $(\ell_1, \dots, \ell_\Delta) \in \L_1 \times \dots \times \L_\Delta$ of labels we have $\{ \ell_1, \dots, \ell_\Delta \} \in \nodeconst_\Pi$.
The edge constraint $\edgeconst_{\rere(\Pi)}$ of $\rere(\Pi)$ is the collection of all configurations $\L_1 \dots \L_\delta$ such that each $\L_i$ appears in at least one configuration from $\nodeconst_{\rere(\Pi)}$ and there exists a choice $(\ell_1, \dots, \ell_\delta) \in \L_1 \times \dots \times \L_\delta$ of labels satisfying $\{ \ell_1, \dots, \ell_\delta \} \in \edgeconst_\Pi$.

We will refer to the operation of deriving $\edgeconst_{\re(\Pi)}$ from $\edgeconst_{\Pi}$ (and $\nodeconst_{\rere(\Pi)}$ from $\nodeconst_{\Pi}$) as \emph{applying the universal quantifier (to $\edgeconst_{\Pi}$ and $\nodeconst_{\Pi}$, respectively)} and say that a problem \emph{satisfies the universal quantifier} if it is the result of such an operation.

The hard part in computing $\re(\Pi)$ and $\rere(\Pi)$ is applying the universal quantifier. In fact, consider the problem $\re(\Pi)$. There is an easy way to compute $\nodeconst_{\re(\Pi)}$, that is the following. Start from all the configurations in $\nodeconst_{\Pi}$, and for each configuration add to $\nodeconst_{\re(\Pi)}$ the condensed configuration obtained by replacing each label $\ell$ by the set that contains all label sets in $\Sigma_{\re(\Pi)}$ containing $\ell$.

\paragraph{The round elimination sequence.}
In the round elimination framework, the two operators $\re$ and $\rere$ are used to define a sequence of problems that is essential for obtaining complexity lower bounds via round elimination.
This sequence $\Pi_0, \Pi_1, \Pi_2, \dots$ is defined via $\Pi_{i+1} := \rere(\re(\Pi_i))$ for all $i \geq 0$, where $\Pi_0$ is the given problem of interest.
The following theorem provides a way to obtain lower bounds for the complexity of $\Pi_0$ via analyzing the $0$-round-solvability of the problems in the sequence.
It is a simplified version of Theorem 7.1 from~\cite{hideandseek}.

\begin{theorem}\label{thm:lifting}
	Let $\Pi_0, \Pi_1, \ldots, \Pi_t$ be a sequence of problems satisfying $\Pi_{i+1} = \rere(\re(\Pi_i))$ for all $0 \leq i \leq t - 1$.
	Moreover, let $B$ be an integer (that may depend on $n$ and/or $\Delta$) such that $|\Sigma_{\Pi_i}| \leq B$ for all $0 \leq i \leq t$, and $|\Sigma_{\re(\Pi_i)}| \leq B$ for all $0 \leq i \leq t - 1$.
	Then, if $\Pi_t$ is not $0$-round-solvable in the port numbering model, $\Pi_0$ has lower bounds of $\Omega(\min\{t, \log_\Delta n - \log_\Delta \log B\})$ rounds in the deterministic \LOCAL model and $\Omega(\min\{t, \log_\Delta \log n - \log_\Delta \log B\})$ rounds in the randomized \LOCAL model.
\end{theorem}

\paragraph{Fixed points.}
As implied by~\Cref{thm:lifting}, it is crucial for proving lower bounds via round elimination to be able to determine the $0$-round solvability of problems in the round elimination sequence produced by the studied problem $\Pi_0$.
A class of problems that produces very simple sequences are so-called fixed points.
A locally checkable problem $\Pi$ is called a \emph{fixed point} if $\rere(\re(\Pi)) = \Pi$.
Moreover, for a fixed point $\Pi$, the problem $\Pi' := \re(\Pi)$ is called the \emph{intermediate problem}.
Note that such an intermediate problem $\Pi'$ satisfies $\re(\rere(\Pi')) = \Pi'$. We get the following corollary from \Cref{thm:lifting}.

\begin{corollary}\label{cor:lifting}
	Let $\Pi$ be a fixed point in the round elimination framework.
	Then, if $\Pi$ is not $0$-round-solvable in the port numbering model, $\Pi$ has lower bounds of $\Omega(\log_\Delta n)$ rounds in the deterministic \LOCAL model and $\Omega(\log_\Delta\log n)$ rounds in the randomized \LOCAL model.

\end{corollary}

\paragraph{0-round-solvability.}
Due to \Cref{thm:lifting}, we are interested in determining whether a problem can be solved in $0$ rounds or not.
For technical reasons, throughout the paper, whenever we consider the $0$-round-solvability of a problem, we will consider it in the \emph{port numbering model}.
In the port numbering model, $0$-round-solvability admits a simple characterization: a problem $\Pi$ is $0$-round-solvable if and only if there is a configuration $\ell_1 \s \dots \s \ell_{\Delta} \in \nodeconst_\Pi$ such that, for any $\delta$ (not necessarily distinct) labels $\ell'_1, \dots, \ell'_{\delta} \in \{ \ell_1, \dots, \ell_{\Delta} \}$, it holds that $\ell'_1 \s \dots \s \ell'_{\delta}\in \edgeconst_\Pi$.
We will use the terms \emph{trivial} and \emph{non-trivial} to refer to $0$-round-solvable and non-$0$-round-solvable problems, respectively.
In particular, we will be interested in trivial and non-trivial fixed points.

\subsection{Example: Sinkless Orientation}
\label{ssec:exso}
To illustrate the definitions provided above, we will consider the problem of \emph{sinkless orientation}, introduced in~\cite{Brandt2016}, on $3$-regular graphs.
In this problem, the task is to orient the edges of the input graph such that no node is a sink, i.e., each node has at least one outgoing incident edge.
Sinkless orientation can be encoded as a problem $\Pi$ in the black-white formalism by setting
\begin{itemize}
	\item $\Sigma_{\Pi} := \{ \I, \O \}$,
	\item $\nodeconst_{\Pi} := \{ \I \s \I \s \O, \quad \I \s \O \s \O, \quad \O \s \O \s \O \}$, and
	\item $\edgeconst_{\Pi} := \{ \I \s \O \}$.
\end{itemize}
Here, the label $\I$ assigned to a node-edge pair $(v,e) \in \fF$ indicates that edge $e$ is oriented towards $v$, whereas the label $\O$ assigned to $(v,e)$ would indicate that $e$ is oriented away from $v$.
The edge constraint $\edgeconst_\Pi$ simply represents the requirement of a proper orientation, i.e., that each edge has to be oriented away from exactly one endpoint and oriented towards the other endpoint.
The node constraint $\nodeconst_\Pi$ represents that each node has at least one outgoing edge, by requiring that at least one incident node-edge pair is assigned the label $\O$.

The node constraint $\nodeconst_\Pi$ can also be written as the condensed configuration $[\I \s \O] \s [\I \s \O] \s \O$, as the latter represents the set $\{ \I \s \I \s \O,\s \I \s \O \s \O,\s \O \s \I \s \O,\s \O \s \O \s \O \}$, which is exactly $\nodeconst_\Pi$ (since $\I \s \O \s \O = \O \s \I \s \O$).
The edge diagram of $\Pi$ is simply the directed graph with node set $\{ \I, \O \}$ and no edges, as replacing $\I$ with $\O$ (or $\O$ with $\I$) in the edge configuration $\I \s \O \in \edgeconst_\Pi$ does not result in an configuration contained in $\edgeconst_\Pi$.

In the following we illustrate the application of $\rere$ to $\Pi$ (which in the case of $\Pi$ being sinkless orientation is a bit more interesting than applying $\re$).
For the problem $\rere(\Pi)$ we obtain the following.
\begin{itemize}
	\item $\Sigma_{\rere(\Pi)} = \{ \emptyset, \{ \I \}, \{ \O \}, \{ \I, \O \} \}$. 
	\item For the node constraint $\nodeconst_{\rere(\Pi)}$, we first compute the set $\nodeconst$ of all configurations $\L_1 \s \L_2 \s \L_3$ (with labels from $\Sigma_{\rere(\Pi)}$) such that $\ell_1 \s \ell_2 \s \ell_3 \in \nodeconst_\Pi$ for all choices $(\ell_1, \ell_2, \ell_3) \in \L_1 \times \L_2 \times \L_3$.
		From the definition of $\nodeconst_\Pi$, we can infer that $\nodeconst$ is precisely the set of configurations $\L_1 \s \L_2 \s \L_3$ such that one of $\L_1, \L_2, \L_3$ is a subset of  $\{ \O \}$ and the other two subsets of $\{ \I, \O \}$.
		Now, we obtain $\nodeconst_{\rere(\Pi)}$ from $\nodeconst$ by removing all non-maximal configurations.
		As is straightforward to verify, the only configuration that is maximal is $\{ \O \} \s \{ \I, \O \} \s \{ \I, \O \}$ (and its permutations).
		As such, we obtain $\nodeconst_{\rere(\Pi)} = \{ \{ \O \} \s \{ \I, \O \} \s \{ \I, \O \} \}$.
	\item By the definition of the edge constraint $\edgeconst_{\rere(\Pi)}$, we obtain $\edgeconst_{\rere(\Pi)} = \{ \{ \O \} \s \{ \I, \O \}, \s \{ \I, \O \} \s \{ \I, \O \} \}$.
		Again, we can write this set of configurations as the condensed configuration $[ \{ \O \}, \{ \I, \O \} ] \s \{ \I, \O \}$.
\end{itemize}

We remark that since the two labels $\emptyset$ and $\{ \I \}$ contained in $\Sigma_{\rere(\Pi)}$ do not occur in $\nodeconst_{\rere(\Pi)}$ or $\edgeconst_{\rere(\Pi)}$, we can, for simplicity, remove them from $\Sigma_{\rere(\Pi)}$, resulting in $\Sigma_{\rere(\Pi)} = \{ \{ \O \}, \{ \I, \O \} \}$.
Moreover, for convenience, we may rename the labels $\{ \O \}$ and $\{ \I, \O \}$ to $\O$ and $\I$, respectively.
In this case, using condensed configurations, the problem $\rere(\Pi)$ would be given by $\Sigma_{\rere(\Pi)} = \{ \O, \I \}$, $\nodeconst_{\rere(\Pi)} = \{ \O \s \I \s \I \}$, and $\edgeconst_{\rere(\Pi)} = [ \O, \I ] \s \I$.

Using the characterization of $0$-round-solvability given in \Cref{ssec:re}, it is straightforward to verify that $\rere(\Pi)$ cannot be solved in $0$ rounds as $\O$ is a label in the only configuration contained in $\nodeconst_{\rere(\Pi)}$ but $\O \s \O \notin \edgeconst_{\rere(\Pi)}$.

\section{A New Way of Applying Round Elimination}\label{sec:newre}
\newcommand{\cond}[0]{\Gamma}

In this section, we describe a novel and simple way for applying the round elimination technique. As already discussed in \Cref{ssec:re}, the hard and error-prone part in applying the $\re(\cdot)$ and $\rere(\cdot)$ operators consists in applying the universal quantifier. Let $\Pi = (\Sigma_{\Pi},\nodeconst_{\Pi},\edgeconst_{\Pi})$ be the problem of interest, where $\nodeconst_{\Pi}$ contains multisets of size $\Delta$ and $\edgeconst_{\Pi}$ contains multisets of size $\delta$. Also, let $\Pi' = \re(\Pi) = (\Sigma_{\Pi'}, \nodeconst_{\Pi'},\edgeconst_{\Pi'})$. Recall that applying the universal quantifier means computing $\edgeconst_{\Pi'}$ as follows.
First, let $\ccs$ be the maximal set such that for all $\L_1 \s \ldots \s \L_\delta \in \ccs$ it holds that, for all $i$, $\L_i \in 2^{\Sigma_{\Pi}}  \setminus \{\{\}\}$, and all multisets $\{\ell_1, \ldots, \ell_\delta\} \in \L_1 \times \ldots \times \L_\delta$ are in $\edgeconst_\Pi$. Then,  $\edgeconst_{\Pi'}$ is obtained by removing all non-maximal configurations from $\ccs$.
This definition, if implemented in a naive manner, requires considering all possible configurations from labels in $2^{\Sigma_{\Pi}}$, and then, for each of them, checking if all possible configurations obtained by selecting one label from each set in the configuration are contained in $\edgeconst_\Pi$.

\subsection{A new way to compute $\edgeconst_{\Pi'}$.}
We show a drastically simplified way of applying the universal quantifier, that, at each point in time, requires to consider only two configurations and to perform elementary operations on those.

\paragraph{Input of the new procedure.}
While, formally, the given constraint $\edgeconst_{\Pi}$ is described as a set of multisets, in some cases the given constraint is described in a more compact form, that is, by providing \emph{condensed configurations}. The procedure that we describe does not need to unpack condensed configurations into a set of non-condensed ones, and this feature allows to apply this new procedure more easily. 
For this reason, we assume that $\edgeconst_{\Pi}$ is described as a set $\cond_{\Pi}$ of condensed configurations, that is, $\cond_{\Pi}$ contains multisets, where each multiset $\fL \in \cond_{\Pi}$ is of the form $\{\L_1, \ldots, \L_\delta\}$, and for all $1 \le i \le \delta$ it holds that $\L_i \subseteq \Sigma_{\Pi}$. Clearly, if we are given $\edgeconst_{\Pi}$ as a list of non-condensed configurations, we can convert it into this form by replacing each label with a singleton set.
While we assume that the input is described as a set of condensed configurations, the output of the procedure is going to be a set of non-condensed configurations. 
We call the condensed configurations in $\cond_{\Pi}$ \emph{input configurations}.

\paragraph{Combining configurations.}
At the heart of our procedure lies an operation that \emph{combines} two given configurations of sets.
We now formally define what it means to combine two such configurations. Let $\fL = \{\L_1, \ldots, \L_\delta\}$ and $\fL' = \{\L'_1, \ldots, \L'_\delta\}$ be two configurations, where $\L_i$ and $\L'_i$ are sets. Let $\phi \colon \set{1,\ldots,\delta} \rightarrow \set{1,\ldots,\delta}$ be a bijection, i.e., a permutation of $\set{1,\ldots,\delta}$. Let $u \in \set{1,\ldots,\delta}$. Combining $\L$ and $\L'$ w.r.t.\ $\phi$ and $u$ means constructing the configuration $\fC = \{\C_1, \ldots, \C_\delta\}$ where $\C_i = \L_i \cup \L'_{\phi(i)}$ if $i = u$ and $\C_i = \L_i \cap \L'_{\phi(i)}$ otherwise.
In other words, we consider an arbitrary perfect matching between the sets of the two configurations, and we take the union for one matched pair and the intersection for the remaining matched pairs. In \Cref{fig:combining}, we show an example of a combination of two configurations.
\begin{figure}[h]
	\begin{align*}
		\{\I, \O\} &\cup \{\O\} = \{\I, \O\} \\
		\{\I, \O\} &\cap \{\I, \O\} = \{\I, \O\} \\
		\{\O\} &\cap \{\I, \O\} = \{\O\}
	\end{align*}
	\caption{One possible way to combine $\set{\I, \O} \set{\I, \O} \set{\O}$ with itself. The resulting configuration is  $\set{\I, \O} \set{\I, \O} \set{\O}$.}
	\label{fig:combining}
\end{figure}

\paragraph{The New Procedure.}
\newcommand{\newre}{\mathrm{NewRE}}
In the following, we construct a sequence $(\Psi_i)$ of sets of configurations until certain desirable properties are obtained.
The first step of the procedure is setting $\Psi_0 = \cond_{\Pi}$. The next step is to apply a subroutine that creates $\Psi_{i+1}$ as a function of $\Psi_i$, and this subroutine is repeatedly applied until we get that $\Psi_{i+1} = \Psi_i$. Let the final result be $\edgeconst_{\Pi'}^*$.

The subroutine computes all possible combinations of pairs of configurations (including a configuration with itself) that are in $\Psi_i$, for all possible permutations $\phi$ and for all possible choices of $u$. If a resulting configuration contains an empty set, the configuration is discarded.
Let $\Psi_{i+1}$ be the set of configurations obtained by starting from the configurations in $\Psi_i$, adding the newly computed configurations, and then removing the non-maximal ones.
We call the defined procedure $\newre$, which is described more formally in \Cref{algo:newre}.

\begin{algorithm}
	\caption{The new procedure.}	\label{algo:newre}
	\begin{algorithmic}
		\LComment{Applies the procedure to the input configurations $\cond$}
		\Procedure{NewRe}{$\cond$, $\delta$}
			\State $\Psi_0 \gets \cond$
			\For{$i \gets 0,1,2,\ldots$}
				\State $\Psi \gets \Psi_{i}$
				\ForAll{$\fL \in \Psi_i$}
					\ForAll{$\fL' \in \Psi_i$}
						\ForAll{permutations $\phi$ over the integers $\{1,\ldots,\delta\}$}
							\ForAll{$1 \le u \le \delta$}
								\State $\fC \gets \textsc{Combine}$($\fL$, $\fL'$, $\delta$, $\phi$, $u$)
								\If{ $\{\} \notin \fC$}
									\State $\Psi \gets \Psi \cup \{\fC\}$
								\EndIf
							\EndFor
						\EndFor
					\EndFor
				\EndFor
				\State $\Psi_{i+1} \gets \textsc{DiscardNonMaximal}$($\Psi$)
				\If{$\Psi_{i+1} = \Psi_i$}
					\State \textbf{break}
				\EndIf
			\EndFor
			\State \Return $\Psi_{i}$
		\EndProcedure
		
		\LComment{Combines two configurations w.r.t.\ a given permutation $\phi$ and position $u$}
		\Procedure{Combine}{$\fL = \{\L_1,\ldots,\L_\delta\}$, $\fL' = \{\L'_1,\ldots,\L'_\delta\}$, $\delta$, $\phi$, $u$}
			\For{$i \gets 1,\ldots,\delta$}
				\If{$i = u}$
					\State $\C_i = \L_i \cup \L'_{\phi(i)}$
				\Else
					\State $\C_i = \L_i \cap \L'_{\phi(i)}$ 
				\EndIf
			\EndFor
			\State $\fC \gets \{\C_1,\ldots,\C_\delta\}$
			\State \Return $\fC$
		\EndProcedure
		
		\LComment{Returns the set of maximal configurations of $\Psi$}
		\Procedure{DiscardNonMaximal}{$\Psi$}
			\State $S \gets \{\}$
			\ForAll{$\fL \in \Psi$}
				\If{$\lnot$($\exists \fL' \in \Psi$ s.t. $\fL' \neq \fL$ and $\textsc{Dominates}$($\fL'$, $\fL$))}
					\State $S \gets S \cup \{\fL\}$
				\EndIf
			\EndFor
			\Return $S$
		\EndProcedure
		
		\LComment{Determines whether $\fL'$ dominates $\fL$}
		\Procedure{Dominates}{$\fL' = \{\L'_1,\ldots,\L'_\delta\}$, $\fL = \{\L_1,\ldots,\L_\delta\}$}
			\State \Return $\exists$ permutation $\phi$ such that, for all $1 \le i \le \delta$, $\L_i \subseteq \L'_{\phi(i)}$
		\EndProcedure
	\end{algorithmic}
\end{algorithm}

In the rest of the section, we will prove that the constraint $\edgeconst_{\Pi'}^*$ returned by $\newre$ is equal to the constraint $\edgeconst_{\Pi'}$ as defined according to the definition of round elimination given in \Cref{sec:preliminaries}, that is, we prove the following theorem.
\begin{theorem}\label{thm:newreiscorrect}
	 $\edgeconst_{\Pi'}^* = \edgeconst_{\Pi'}$.
\end{theorem}

\paragraph{Example of \boldmath $\newre$.}
Before proving \Cref{thm:newreiscorrect}, we provide an example of the application of the procedure $\newre$. Consider the problem of $3$-coloring in $3$-regular graphs. This problem can be defined, in the black-white formalism, as follows (we call this problem $\Pi$).
\begin{equation*}
	\begin{aligned}
		\begin{split}
			\nodeconst_{\Pi}\text{:}\\
			&\A &\s& \A & \s &\A \\
			&\B &\s& \B & \s &\B \\
			&\C &\s& \C & \s &\C \\
		\end{split}
		\qquad
		\begin{split}
			\edgeconst_{\Pi}\text{:}\\
			&\A &\s& [\B \s \C]\\
			&\B &\s& \C\\
		\end{split}
	\end{aligned}
\end{equation*}
The constraints can be interpreted as follows. A node is of color $\A$, $\B$, or $\C$, and the edge constraint forbids nodes of the same color to be neighbors. For the purpose of this example, we will first provide the problem $\re(\Pi)$ without showing how it is obtained, and then, we will show how to apply the new procedure on $\re(\Pi)$, in order to obtain the node constraint of  $\rere(\re(\Pi))$. The problem $\re(\Pi)$, after renaming, can be defined as follows.
\begin{equation*}
	\begin{aligned}
		\begin{split}
			\nodeconst_{\re(\Pi)}\text{:}\\
			&[\A \s \C \s \E] &\s& [\A \s \C \s \E ] & \s &[\A \s \C \s \E] \\
			&[\B \s \C \s \F] &\s& [\B \s \C \s \F ] & \s &[\B \s \C \s \F] \\
			&[\D \s \E \s \F] &\s& [\D \s \E \s \F ] & \s &[\D \s \E \s \F] \\
		\end{split}
		\qquad
		\begin{split}
			\edgeconst_{\re(\Pi)}\text{:}\\
			&\A &\s& \F\\
			&\B &\s& \E\\
			&\D &\s& \C\\
		\end{split}
	\end{aligned}
\end{equation*}
We now show how to obtain the node constraint of $\rere(\re(\Pi))$ by applying the procedure $\newre$ on the node constraint of $\re(\Pi)$. The first step is computing $\cond$, which is obtained by replacing each condensed configuration of $\nodeconst_{\re(\Pi)}$ with a single configuration. Hence, 
\begin{align*}
	\cond &= \{  \\
	 	&\{\{\A,\C,\E\},\{\A,\C,\E\},\{\A,\C,\E\}\},   \\ 
	 	&\{\{\B,\C,\F\},\{\B,\C,\F\},\{\B,\C,\F\}\},   \\
	 	&\{\{\D,\E,\F\},\{\D,\E,\F\},\{\D,\E,\F\}\}    \\
	 \}.&
\end{align*}
Then, the procedure $\newre$ initializes $\Psi_0$ to $\cond$, and in order to compute $\Psi_i$ it considers all possible pairs of lines, all possible permutations $\phi$, and all possible positions $1 \le u \le 3$. Consider the following choice of parameters:
\begin{align*}
	\fL  &= \{\{\A,\C,\E\},\{\A,\C,\E\},\{\A,\C,\E\}\}\\
	\fL' &= \{\{\B,\C,\F\},\{\B,\C,\F\},\{\B,\C,\F\}\}\\
	&\phi(1) = 1, \phi(2) = 2, \phi(3) = 3, u = 1.
\end{align*}
By combining these configurations w.r.t.\ these parameters, we obtain the following configuration:
\[
	\fC = \{\{\A,\B,\C,\E,\F\},\{\C\},\{\C\}\}.
\]
Observe that this configuration is not dominated by any configuration that is already present, and that any configuration that is already present is not dominated by this configuration. One can check that $\Psi_1$ contains exactly the following configurations.
\begin{align*}
	&\{\{\A, \C, \E\},  \{\A, \C, \E\}, \{\A, \C, \E\}\} \\
	&\{\{\B, \C, \F\},  \{\B, \C, \F\}, \{\B, \C, \F\}\} \\
	&\{\{\D, \E, \F\},  \{\D, \E, \F\}, \{\D, \E, \F\}\} \\
	&\{\{\A, \B, \C,\E,\F\}, \{\C\}, \{\C\}\} \\
	&\{\{\A, \C, \D,\E,\F\}, \{\E\}, \{\E\}\} \\
	&\{\{\B, \C, \D,\E,\F\},  \{\F\}, \{\F\}\} \\
\end{align*}
Moreover, it is possible to check that, by computing $\Psi_2$ starting from $\Psi_1$, no new configurations are obtained, and hence $\Psi_1 = \Psi_2$ and the procedure terminates. For example, consider the following parameters:
\begin{align*}
	\fL  &= \{\{\A,\C,\E\},\{\A,\C,\E\},\{\A,\C,\E\}\}\\
	\fL' &= \{\{\A,\B,\C,\E,\F\},\{\C\},\{\C\}\}\\
	&\phi(1) = 1, \phi(2) = 2, \phi(3) = 3, u = 2.
\end{align*}
By combining these configurations w.r.t.\ these parameters, we obtain the following configuration:
\[
\fC = \{\{\A,\C,\E\},\{\A,\C,\E\},\{\C\}\}.
\]
This configuration is dominated by $\fL$. Another interesting example is given by the following parameters:
\begin{align*}
	\fL  &= \{\{\A,\C,\E\},\{\A,\C,\E\},\{\A,\C,\E\}\}\\
	\fL' &= \{\{\B,\C,\D,\E,\F\},\{\F\},\{\F\}\}\\
	&\phi(1) = 1, \phi(2) = 2, \phi(3) = 3, u = 1.
\end{align*}
By combining these configurations w.r.t.\ these parameters, we obtain the following configuration:
\[
\fC = \{\{\A,\B,\C,\D,\E,\F\},\{\},\{\}\}.
\]
This configuration contains an empty set, and hence it is discarded.

\subsection{Soundness and Completeness of $\newre$}
We now prove that $\newre$ generates all and only the maximal configurations that satisfy the universal quantifier, that is, \Cref{thm:newreiscorrect}.

\subsubsection{Procedure Soundness}
By the definition of condensed configurations, $\Psi_0$ is initialized with configurations that satisfy the universal quantifier.
We now show that any combination of valid configurations (i.e., that satisfy the universal quantifier) generates configurations that are also valid, implying that we never obtain invalid configurations.

\begin{lemma}[Combination is sound]
	\label{lem:sound}	
	Given two configurations $\fL = \{\L_1, \ldots, \L_\delta\}$ and $\fL' = \{\L'_1, \ldots, \L'_\delta\}$ that satisfy the universal quantifier, any combination $\fC$ of $\fL$ and $\fL'$ also satisfies the universal quantifier.
\end{lemma}
\begin{proof}
	Let $\fC = \{\C_1, \ldots, \C_\delta \}$ be a configuration obtained by combining $\fL$ and $\fL'$ w.r.t.\ some $\phi$ and $u$, such that $\{\} \notin \fC$. Consider an arbitrary choice $\{\c_1, \ldots, \c_\delta\} \in \C_1 \times \ldots \times \C_\delta$.
Observe that $\C_u = \L_u \cup \L'_{\phi(u)}$. Hence, $\c_u$ is contained in $\L_u$ or in $\L'_{\phi(u)}$. W.l.o.g., let $\c_u$ be in $\L_u$.
Observe that, for each $i \neq u$, $\c_i \in \L_i \cap \L'_{\phi(i)}$, and hence $\c_i \in \L_i$. Hence, the configuration $\{\c_1, \ldots, \c_\delta\}$ is in $\L_1 \times \ldots \times \L_\delta$.
\end{proof}

\subsubsection{Procedure Completeness}
In the rest of the section we show that $\newre$ is also \emph{complete}, that is, the resulting $\edgeconst^*_{\Pi'}$ contains all the configurations required by the definition. Combined with \Cref{lem:sound}, we obtain that $\newre$ generates exactly the configurations required by the definition of $\edgeconst_{\Pi'}$.

\paragraph{Domination relation.}
The notion of maximality implicitly defines a notion of domination between configurations: a configuration $\{\L_1, \ldots, \L_\delta\}$ is dominated by a configuration $\{\L'_1, \ldots, \L'_\delta\}$ if there exists a permutation $\phi$ such that $\L_i \subseteq \L'_{\phi(i)}$ for all $i$.
\begin{lemma}[Transitivity]
	\label{lem:domtrans}
	The domination relation is transitive.
\end{lemma}
\begin{proof}
	Assume that we are given the configurations $\fL_1 = \{\L_{1,1}, \ldots, \L_{1,\delta}\}$, $\fL_2 = \{\L_{2,1}, \ldots, \L_{2,\delta}\}$, and $\fL_3 = \{\L_{3,1}, \ldots, \L_{3,\delta}\}$ such that $\fL_1$ is dominated by $\fL_2$ and $\fL_2$ is dominated by $\fL_3$. Let $\phi$ (resp.\ $\phi'$) be the permutation satisfying that $\L_{1,i} \subseteq \L_{2,\phi(i)}$ (resp.\ $\L_{2,i} \subseteq \L_{3,\phi'(i)}$) for all $i$. We obtain that $\L_{1,i} \subseteq \L_{3,\phi'(\phi(i))}$ for all $i$, and hence that $\fL_1$ is dominated by $\fL_3$.
\end{proof}
We say that a configuration $\fL'$ is strictly dominated by $\fL$ if $\fL'$ is dominated by $\fL$ and $\fL$ is not dominated by $\fL'$. We prove that the strict domination relation is \emph{well-founded}. This property will be used later to prove that our procedure is complete. 
Recall that a relation is well-founded if, in any non-empty set, there is a minimal element. In the case of the strict domination relation this means that, given any set of configurations, there exists at least one configuration that is not strictly dominated by any other configuration. 

\begin{lemma}[Well-foundedness]\label{lem:wellfounded}
	The strict domination relation is well-founded.
\end{lemma}
\begin{proof}
	Define the weight of a configuration as the sum of the cardinalities of its sets. Obviously no configuration can have negative weight. We prove that, if a configuration $\fL = \{\L_1, \ldots, \L_\delta\}$ strictly dominates a configuration $\fL' = \{\L'_1, \ldots, \L'_\delta\}$, then the weight of $\fL'$ is strictly less than the weight of $\fL$. In fact, let $\phi$ be the permutation witnessing the strict domination relation between $\fL$ and $\fL'$, that is, $\L'_i \subseteq L_{\phi(i)}$ for all $1 \le i \le \delta$, and the inclusion is strict for at least one value of $i$. Observe that $|\L'_i| \le |L_{\phi(i)}|$ for all $1 \le i \le \delta$, and there is at least one value of $i$ such that $|\L'_i| < |L_{\phi(i)}|$. Hence, the weight of $\fL$ is strictly larger than the weight of $\fL'$.
	
	A suitable minimal configuration for any set is a configuration with the lowest weight. If there was a configuration strictly dominated by a lowest-weight configuration, that configuration would have even lower weight, which is a contradiction.
\end{proof}

\paragraph{Configuration construction from the input  configurations.}
We call a configuration $\fL = \{\L_1, \ldots, \L_\delta\}$ a \emph{singleton configuration} if $|\L_i| = 1$ for all $i$.
\begin{lemma}[Configuration splitting]
	\label{lem:linesplit}
	For any non-singleton configuration $\fC$ there exist two configurations strictly dominated by $\fC$ that can be combined into $\fC$.
\end{lemma}

\begin{proof}
	Since $\fC$ is not a singleton configuration, it must contain some set $\S = \{\a, \b, \dots\}$ that contains at least two elements. We create two configurations from $\fC$: one where $\S$ is replaced with $\S \setminus \set{\a}$ and another where $\S$ is replaced with $\S \setminus \set{\b}$. Observe that $(\S \setminus \set{\a}) \cup (\S \setminus \set{\b}) = \S$, so the two created configurations can be combined into $\fC$ and $\fC$ strictly dominates them, as they have strictly fewer labels in them.
\end{proof}

\begin{lemma}[Configuration construction]
	\label{lem:explode}
	Any configuration can be built by combining singleton configurations that it dominates.
\end{lemma}

\begin{proof}
	By \Cref{lem:wellfounded}, the strict domination relation is well-founded, and hence we can perform well-founded induction on configurations. Thus, it suffices to show that if all configurations strictly dominated by $\fC$ can be built from dominated singleton configurations, $\fC$ can be, too.
	
	The lemma clearly holds if $\fC$ is a singleton configuration. If $\fC$ is not a singleton configuration, we can use the configuration splitting lemma (\Cref{lem:linesplit}) to show that it can be built out of two configurations that it strictly dominates. The induction hypothesis tells us that those configurations in turn can be built from dominated singleton configurations.
\end{proof}

\begin{lemma}[Property of dominating configurations]\label{lem:biggood}
	If configurations $\fL_1 = \{\L_{1,1}, \ldots, \L_{1,\delta}\}$ and $\fL_2 = \{\L_{2,1}, \ldots, \L_{2,\delta}\}$ can be combined into some configuration $\fC = \{\C_1, \ldots, \C_\delta\}$, then any two configurations $\fL'_1 = \{\L'_{1,1}, \ldots, \L'_{1,\delta}\}$  and $\fL'_2 = \{\L'_{2,1}, \ldots, \L'_{2,\delta}\}$ such that $\fL'_1$ dominates $\fL_1$ and $\fL'_2$ dominates $\fL_2$ can be combined into some configuration $\fC' =  \{\C'_1, \ldots, \C'_\delta\}$ that dominates $\fC$.
\end{lemma}
\begin{proof}
	Assume that $\fC$ is obtained by combining $\fL_1$ and $\fL_2$ w.r.t.\ $\phi$ and $u$. Let $\phi_1$ be the permutation satisfying $\L_{1,i} \subseteq \L'_{1,\phi_1(i)}$ for all $i$, and let $\phi_2$ be the permutation satisfying $\L_{2,i} \subseteq \L'_{2,\phi_2(i)}$ for all $i$. W.l.o.g., assume that $\phi_1$ and $\phi_2$ are the identity function.
	Let $\fC'$ be the combination of $\fL'_1$ and $\fL'_2$ w.r.t.\ $\phi$ and $u$. Observe that $\C_i \subseteq \C'_i$ for all $i$, and hence $\fC'$ dominates $\fC$.
\end{proof}

\begin{lemma}[Combination is complete]
	\label{lem:complete}
	Let $\fC$ be an arbitrary configuration that satisfies the universal quantifier. A configuration dominating $\fC$ can be obtained by combining input configurations.
\end{lemma}
\begin{proof}
	According to \Cref{lem:explode}, the configuration $\fC$ can be built from singleton configurations that it dominates. By the definition of the domination relation, since $\fC$ satisfies the universal quantifier, those singleton configurations also satisfy the universal quantifier. Moreover, all singleton configurations are dominated by at least one condensed configuration that is part of the input. By \Cref{lem:biggood}, we can replace the singleton configurations required by \Cref{lem:explode} with the ones that dominate them and that are part of the input.
\end{proof}

\begin{corollary}
	All the maximal configurations can be built by combining input configurations.
\end{corollary}

\begin{proof}
	A configuration is maximal if it is not strictly dominated by any other configuration satisfying the universal quantifier. Thus, a configuration satisfying the universal quantifier and dominating a maximal configuration must be the configuration itself, and hence by \Cref{lem:complete} it is possible to obtain it by combining input configurations.
\end{proof}

\paragraph{Procedure completeness.}
We have shown that it is possible to start from the input configurations and to repeatedly combine them in order to obtain any maximal configuration from $\edgeconst_{\Pi'}$. However, $\newre$ works slightly differently: at each step, non-maximal configurations are discarded (this makes the procedure easier to apply and more efficient in practice). We now prove that $\newre$ is anyways complete. We denote with \emph{missing configuration} a configuration satisfying the universal quantifier that is not dominated by any of the already computed configurations.

\begin{lemma}
	\label{lem:biggermissing}
	A configuration that dominates a missing configuration is also missing.
\end{lemma}
\begin{proof}
	Let $\fL_1$ and $\fL_2$ be any configuration such that $\fL_2$ dominates $\fL_1$ and $\fL_1$ is missing. Suppose that $\fL_2$ is not missing. Then there is a configuration $\fL_3$ that dominates $\fL_2$. Because the domination relation is transitive (\Cref{lem:domtrans}), $\fL_3$ dominates $\fL_1$ as well, so $\fL_1$ is not missing, which is a contradiction.
\end{proof}

\begin{lemma}
	\label{thm:justtwo}
	Suppose that at least one configuration is missing. Then, some missing configuration can be obtained by combining two already computed configurations.
\end{lemma}
\begin{proof}
	In the following, by \emph{valid configuration} we denote a configuration that satisfies the universal quantifier.
	According to \Cref{lem:complete}, there is some way of combining the input configurations that produces a configuration $\fL$ dominating the missing configuration.
	By \Cref{lem:sound}, any way of (recursively) combining the input configurations can only produce valid configurations, and hence all the configurations leading up to $\fL$ are also valid.
	By definition, a configuration that is valid but not missing is dominated by some computed configuration. Thus, all the configurations leading up to $\fL$ are either missing or dominated by a computed configuration. 
	
	We consider the two cases separately. Let $\fL_1$ and $\fL_2$ be two configurations that can be combined to produce $\fL$. Suppose there exist two \emph{already computed} configurations $\fL'_1$ and $\fL'_2$ that dominate $\fL'_1$ and $\fL'_2$. According to \Cref{lem:biggood}, combining $\fL'_1$ and $\fL'_2$ in the right way yields a configuration that dominates $\fL$. \Cref{lem:biggermissing} tells us that if the obtained configuration strictly dominates the missing one, then the obtained one is missing as well.
	
	Now, suppose that one (or both) of the configurations are missing. In this case, recurse into the missing configuration. Eventually, we will reach a pair of non-missing configurations, as the combination starts with input configurations. At that point, the previous case yields the lemma statement.
\end{proof}

\Cref{thm:justtwo}, combined with \Cref{lem:sound}, implies that, when $\newre$ terminates, it indeed produces all and only the configurations that satisfy the universal quantifier, and hence that it is correct. We now prove that $\newre$ terminates in finite time.
\begin{lemma}[Termination]
	Procedure $\newre$ terminates in finite time.
\end{lemma}
\begin{proof}
	Let $\Psi$ be a set, where each element of the set is a multiset of size $\delta$, and each element of the multiset is a subset of $\Sigma$. That is, $\Psi$ is a constraint with configurations of size $\delta$ and of labels in $2^\Sigma$. 
	We denote with $f(\Psi)$ the number of all possible configurations of size $\delta$ and of labels in $2^\Sigma$ that are dominated by the configurations present in $\Psi$.
	
	Procedure $\newre$ starts from $\cond_{\Pi}$ (the given condensed configurations) and then it repeatedly combines configurations until nothing new is obtained. Recall that $\Psi_0, \Psi_1, \ldots$ is the sequence of constraints computed in $\newre$. Let $n_i = f(\Psi_i)$. Observe that $n_0 \ge 0$, and for all $i$, $n_{i+1} \ge n_i + 1$, since if no missing configuration is obtained, then $\newre$ terminates. The termination of $\newre$ is guaranteed by the fact that $f(\edgeconst_{\Pi'})$ is a finite number (since $|\Sigma|$ and $\delta$ are finite).
\end{proof}

\section{Fixed Point Generation}\label{sec:procedure}
In the \LOCAL model, one of the few known ways that we have for showing that a problem $\Pi$ cannot be solved locally (i.e., in constant time if a suitable form of symmetry breaking is provided) is to prove that the problem can be relaxed into a non-trivial fixed point $\Pi'$.  A non-trivial fixed point relaxation $\Pi'$ for a problem $\Pi$ is a problem satisfying the following: $\Pi'$ can be solved in $0$ rounds given a solution for $\Pi$, $\rere(\re(\Pi')) = \Pi'$, and $\Pi'$ cannot be solved in $0$ rounds in the port numbering model. It is known, by prior work (see \Cref{thm:lifting}), that a non-trivial fixed point relaxation $\Pi'$ for a problem $\Pi$ implies that $\Pi$, in the \LOCAL model, requires $\Omega(\log n)$ rounds for deterministic algorithms and $\Omega(\log \log n)$ rounds for randomized ones. 

In this section, we present a procedure, called $\fpp$, that, given a problem $\Pi$, is able to automatically find a fixed point relaxation for $\Pi$. Sometimes, this fixed point relaxation is a trivial (i.e., $0$-round-solvable) problem (even if the problem we start from has complexity $\Omega(\log n)$), and some other times the fixed point relaxation is non-trivial. 
In the next sections (see \Cref{sec:2col,sec:deltacol,sec:3col}) we will show that, for various problems of interest, procedure $\fpp$ actually provides a non-trivial fixed point relaxation. Hence, while this procedure may not be universal, it is broad enough to be applicable to a variety of interesting problems.

Procedure $\fpp$ takes as input a problem $\Pi$ and a diagram $D$, and the choice of the diagram may affect the triviality of the resulting fixed point. In \Cref{sec:diagram} we will provide a default choice for $D$ (as a function of $\Pi$), and in the case where $\fpp$ fails for the default choice (i.e., it produces a trivial fixed point), we show ways for tweaking it that allow, in some cases, to obtain non-trivial fixed points.  Procedure $\fpp$ is very similar to procedure $\newre$; in fact it only differs in how unions and intersections are computed.

\paragraph{Procedure input.}
The procedure takes as input a problem $\Pi = (\Sigma_\Pi, \nodeconst_\Pi, \edgeconst_\Pi)$, and a \emph{target diagram} $D = (\Sigma_D,E_D)$, which is a directed acyclic graph satisfying the following:
\begin{itemize}
	\item $\Sigma_\Pi \subseteq \Sigma_D$, that is, the label set of $D$ is a superset of the label set of $\Pi$.
	\item If we consider $D$ as a partially ordered set, every pair of elements must have a unique infimum and supremum. More formally, for $\ell \in \Sigma_D$, let $\mathrm{Pred}(\ell)$ (resp.\ $\mathrm{Succ}(\ell)$) be the set of labels in $\Sigma_D$ that can reach (resp.\ are reachable by) $\ell$ according to the edges $E_D$, including $\ell$.
	For $\ell_1,\ell_2 \in \Sigma_D$, let $\mathrm{Pred}(\ell_1,\ell_2) = \mathrm{Pred}(\ell_1) \cap \mathrm{Pred}(\ell_2)$ (resp.\ $\mathrm{Succ}(\ell_1,\ell_2) = \mathrm{Succ}(\ell_1) \cap \mathrm{Succ}(\ell_2)$) be the set of common predecessors (resp.\ successors) of $\ell_1$ and $\ell_2$.
	For $\ell_1,\ell_2 \in \Sigma_D$, let $\mathrm{MaxPred}(\ell_1,\ell_2)$ be the set of elements $\ell \in \mathrm{Pred}(\ell_1,\ell_2)$ satisfying that $\mathrm{Succ}(\ell) \cap \mathrm{Pred}(\ell_1,\ell_2) = \set{\ell}$. Similarly, let $\mathrm{MinSucc}(\ell_1,\ell_2)$ be the set of elements $\ell \in \mathrm{Succ}(\ell_1,\ell_2)$ satisfying that $\mathrm{Pred}(\ell) \cap \mathrm{Succ}(\ell_1,\ell_2) = \set{\ell}$. We require that, for all $\ell_1,\ell_2 \in \Sigma_D$, $|\mathrm{MaxPred}(\ell_1,\ell_2)| = |\mathrm{MinSucc}(\ell_1,\ell_2)| = 1$, and we call $\diaginf{\ell_1}{\ell_2}$ the element in $\mathrm{MaxPred}(\ell_1,\ell_2)$, and $\diagsup{\ell_1}{\ell_2}$ the element in $\mathrm{MinSucc}(\ell_1,\ell_2)$. 
\end{itemize}
An example of a valid target diagram is shown in \Cref{fig:example-valid-diagram}.

\begin{figure}
	\centering
	\includegraphics[width=0.5\textwidth]{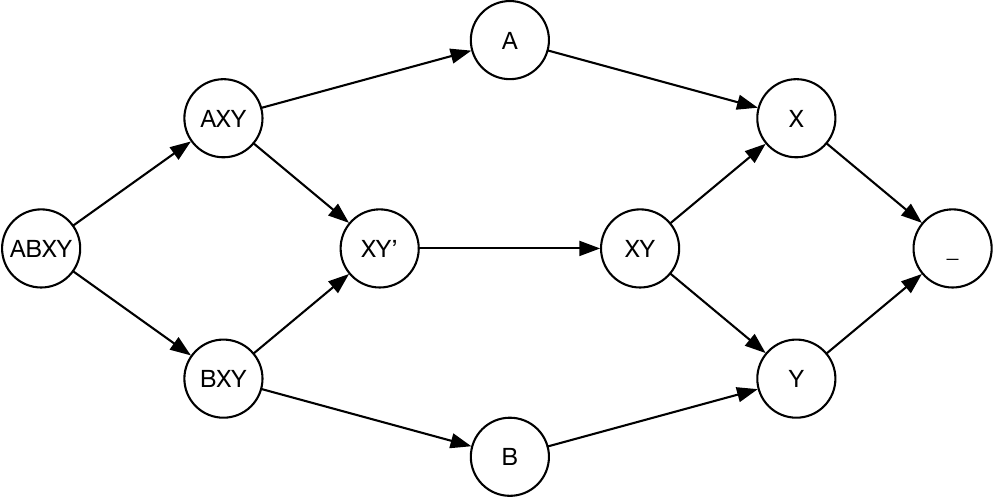}
	\caption{An example of a valid target diagram, for some problem $\Pi$ with labels $\set{A,B,X,Y}$. For example, for labels $A$ and $XY$ we have that $\diaginf{A}{XY} = AXY$ and $\diagsup{A}{XY} = X$.}
	\label{fig:example-valid-diagram}
\end{figure}

In the rest of the section we will define procedure $\fpp$ and we will prove that it satisfies the following theorem.
\begin{theorem}\label{thm:itisafixedpointrelaxation}
 The problem $\Pi' := \fpp(\Pi,D)$ is a fixed point relaxation of $\Pi$.
\end{theorem}

We start by extending the notions of maximality, domination, and combinations, defined in \Cref{sec:newre}, to maximality, domination, and combinations w.r.t.\ a given diagram $D$.

\paragraph{Domination, maximality, and combinations.}
A configuration $\fL = \{\L_1, \ldots, \L_\delta\}$ dominates a configuration $\fL' = \{\L'_1, \ldots, \L'_\delta\}$ w.r.t.\ diagram $D$ if and only if there exists a permutation $\phi$ such that, for all $i$, $\L_i \in \mathrm{Succ}(\L'_{\phi(i)})$, or, equivalently, $\diagsup{\L_i}{\L'_{\phi(i)}} = \L_i$ (in $D$). For example, according to the diagram of \Cref{fig:example-valid-diagram}, the configuration $A \s X \s X$ strictly dominates the configuration $A \s A \s X$.
A configuration is maximal w.r.t.\ $D$ if no other configuration strictly dominates it w.r.t.\ $D$. 

The combination of two configurations is defined similarly as in \Cref{sec:newre}, with the only difference that instead of performing a union for one matched pair and intersections for the remaining ones, we take the supremum for one matched pair and infima for the remaining ones. More formally, let $\fL = \{\L_1, \ldots, \L_\delta\}$ and $\fL' = \{\L'_1 \ldots \L'_\delta\}$ be two configurations. Let $\phi$ be a permutation of $\set{1,\ldots,\delta}$, and let $u \in \set{1,\ldots,\delta}$. Combining $\fL$ and $\fL'$ w.r.t.\ $\phi$ and $u$ means constructing the configuration $\fC = \{\C_1, \ldots, \C_\delta\}$ where $\C_i = \diagsup{\L_i}{\L'_{\phi(i)}}$ if $i = u$ and $\C_i = \diaginf{\L_i}{\L'_{\phi(i)}}$ otherwise. In other words, we consider an arbitrary perfect matching between the sets of the two configurations, and we take the supremum of one matched pair and the infima of the remaining matched pairs.

In order to define procedure $\fpp$ we will make use of a subprocedure, which we will call $\fpprocname$, that we describe in the following.

\paragraph{The subprocedure $\fpprocname$.}
Procedure $\fpprocname$ takes as input two parameters, a constraint $S$ and a target diagram $D$. We now define $S' = \fpproc{S}{D}$. Informally, this procedure is very similar to $\newre$; the only difference is that instead of computing unions and intersections, we take suprema and infima in the diagram. Examples of applications of this procedure can be found in \Cref{sec:deltacol,sec:2col,sec:3col}.

More formally, first, we initialize $S'$ by taking all configurations present in $S$ that are maximal w.r.t.\ $D$.
Then, we compute all possible combinations (w.r.t.\ $D$) of pairs of configurations from $S'$ (including a configuration with itself), for all possible permutations $\phi$ and for all possible choices of $u$. We add the newly computed configurations to $S'$, and then we remove the ones that are non-maximal (w.r.t.\ $D$). We repeat this operation until the set $S'$ does not change anymore.

\paragraph{The procedure $\fpp$.}
Assume that we are given a problem $\Pi = (\Sigma_{\Pi}, \nodeconst_{\Pi}, \edgeconst_{\Pi})$ and a diagram $D = (\Sigma_D,E_D)$. Let $\nodeconst = \fpproc{\nodeconst_{\Pi}}{D}$, and let $\edgeconst = \fpproc{\edgeconst_{\Pi}}{D'}$, where $D'$ is obtained from $D$ by reversing its edges, that is, $D' =  (\Sigma_D,E'_D)$, and $E'_D = \set{(u,v) ~|~ (v,u) \in E_D}$.
Let $\Pi' = (\Sigma_D, \nodeconst, \gen{\edgeconst})$, where $\gen{\edgeconst}$ is defined to be the set of configurations obtained by replacing each configuration $\{\ell_1, \ldots, \ell_\delta\} \in \edgeconst$ with the condensed configuration $\{\gen{\ell_1}, \ldots, \gen{\ell_\delta}\}$ according to the diagram $D$, and $\gen{\ell}$ w.r.t.\ $D$ is the set $\mathrm{Succ}(L)$ w.r.t.\ $D$.
Then, procedure $\fpp(\Pi,D)$ returns $\Pi'$.

\subsection{Proof of \Cref{thm:itisafixedpointrelaxation}}
Let $\Pi'' = (\Sigma_D, \gen{\nodeconst}, \edgeconst)$, where $\gen{\nodeconst}$ is taken according to $D'$.
In order to prove \Cref{thm:itisafixedpointrelaxation}, we prove that $\Pi'$ can be solved in $0$ rounds given a solution for $\Pi$ (\Cref{lem:relaxation}), and that $\Pi'$ is a fixed point with intermediate problem $\Pi''$, that is, $\Pi'' = \re(\Pi')$ and $\Pi' = \rere(\Pi'')$ (\Cref{lem:fp}).
\begin{lemma}\label{lem:relaxation}
	Given a solution for $\Pi$, it is possible to solve $\Pi'$ in $0$ rounds.
\end{lemma}
\begin{proof}
	Let $\fC = \{\ell_1,\ldots,\ell_\delta\}$ be an arbitrary configuration in $\edgeconst_{\Pi}$. We start by proving that there exists a condensed configuration $\fL = \{\L_1,\ldots,\L_\delta\} \in \gen{\edgeconst}$ satisfying  $\{\ell_1,\ldots,\ell_\delta\} \in \L_1 \times \ldots \times \L_\delta$. In other words, we prove that $\gen{\edgeconst}$ allows at least the configurations allowed by $\edgeconst_{\Pi}$. Let us keep track of the specific configuration $\fC$ while running the procedure constructing $\edgeconst$: either it is removed during the initialization of $\edgeconst$ (because it is non-maximal), or it is later replaced with some newly constructed configuration that dominates it (w.r.t.\ $D'$), or it stays in $\edgeconst$ until the end. In all cases, after the procedure ends, $\edgeconst$ must contain a configuration $\fC' = \{\ell'_1,\ldots,\ell'_\delta\}$ that dominates $\fC$ (w.r.t.\ $D'$). By the definition of $\gen{\cdot}$ it must hold that some permutation of $\{\ell_1,\ldots,\ell_\delta\}$ is in $\gen{\ell'_1} \times \ldots \times \gen{\ell'_\delta}$.
	
	Now, consider $\nodeconst$. By construction, for each configuration in $\nodeconst_{\Pi}$ there is at least one configuration in $\nodeconst$ dominating it (w.r.t.\ $D$). Assume we are given a solution for $\Pi$, and consider some node $v$ that is outputting some configuration $\fC \in \nodeconst_{\Pi}$ in this solution. Either $\fC$ is contained in $\nodeconst$ as well, or it has been replaced by a configuration $\fC'$ dominating it (w.r.t.\ $D$). In the first case, node $v$ does nothing, while in the second case node $v$ changes its output to $\fC'$, in a way that replaces each label $\ell$ of $\fC$ with a label of $\fC'$ that is a successor of $\ell$ in $D$. We now show that this results in a valid output for $\Pi'$.
	If $\fC$ is dominated by $\fC'$, then $\fC'$ can be obtained from $\fC$ by replacing each label by one of its successors according to $D$. Consider an edge (or hyperedge) incident to some node that was outputting $\fC$ and now is outputting $\fC'$, and assume that the (hyper)edge had the configuration $\{\ell_1, \ldots, \ell_\delta\}$, which, by the above argument, is in $\gen{\edgeconst}$. The new configuration of the (hyper)edge must be $\{\ell'_1, \ldots, \ell'_\delta\}$, where $\ell'_i$  is either $\ell_i$ or one of its successors. Observe that $\{\ell'_1, \ldots, \ell'_\delta\}$ is in $\gen{\edgeconst}$ by the definition of $\gen{\cdot}$.
\end{proof}
\begin{lemma}\label{lem:fp}
	$\re(\Pi') = \Pi''$ and $\rere(\Pi'') = \Pi'$.
\end{lemma}
\begin{proof}
	\newcommand{\inre}[1]{#1_{\mathrm{re}}}
	\newcommand{\infp}[1]{#1_{\mathrm{fp}}}
	Recall that we consider two problems to be equal if one can be obtained from the other by renaming labels.
	In the following, we will use the terms \emph{RE-maximal}, \emph{RE-dominated}, and \emph{RE-combine} to refer to the original maximality, domination, and combination definitions of \Cref{sec:newre}.
	
	We prove that $\rere(\Pi'') = \Pi'$; the other case is symmetric. Recall that $\Pi'' = (\Sigma_D, \gen{\nodeconst}, \edgeconst)$ and $\Pi' = (\Sigma_D, \nodeconst, \gen{\edgeconst})$, and that $\newre$, applied on $\gen{\nodeconst}$, works as follows. The constraint $\gen{\nodeconst}$ is already provided as a set of condensed configurations of the form $\{\gen{\L_1}, \ldots, \gen{\L_\Delta}\}$, and hence the constraint $\nodeconst_{\rere(\Pi'')}$ is initialized by putting in it all configurations of $\gen{\nodeconst}$, and then discarding non-RE-maximal ones. Let the obtained set be $\nodeconst'$.
	Then, $\newre$ repeatedly RE-combines two configurations and adds the result to $\nodeconst_{\rere(\Pi'')}$ (discarding non-RE-maximal configurations). 
	
	We start by showing that, under the renaming $r := \gen{\ell} \rightarrow \ell$, we get that $\nodeconst' = \nodeconst$. Observe that $\nodeconst'$ contains exactly the RE-maximal configurations of the form $\{\gen{\ell_1}, \ldots, \gen{\ell_\Delta}\}$, where $\{\ell_1,\ldots,\ell_\Delta\} \in \nodeconst$, and $\gen{\cdot}$ is taken according to $D'$ (as it is also in the following). Note also that, by construction, all configurations of $\nodeconst$ are maximal w.r.t.\ $D$. We show that all the condensed configurations of $\gen{\nodeconst}$ (considered as configurations of label sets) are RE-maximal. Assume for a contradiction that there is a configuration $\fL = \{\gen{\ell_1}, \ldots, \gen{\ell_\Delta}\}$ in $\gen{\nodeconst}$ that is strictly RE-dominated by another configuration $\fL' = \{\gen{\ell'_1}, \ldots, \gen{\ell'_\Delta}\}$ in $\gen{\nodeconst}$. This implies that there exists a permutation $\phi$ such that, for all $i$, $\gen{\ell_i} \subseteq \gen{\ell'_{\phi(i)}}$, and for at least one value of $i$ the inclusion is strict. By the definition of $\gen{\cdot}$, this implies that, for all $i$, $\ell_i$ is a successor of $\ell'_{\phi(i)}$ in $D'$, and, for at least one value of $i$, it is a successor different from $\ell'_{\phi(i)}$ itself. This implies that the configuration $\fL$ is non-maximal w.r.t.\ $D$, a contradiction. Hence, we obtain that under the renaming $r$, $\nodeconst' = \nodeconst$.
		
	We now prove that if we take two arbitrary configurations from $\nodeconst'$, and we RE-combine them in an arbitrary way, we either obtain a configuration already present in $\nodeconst'$, or a configuration RE-dominated by a configuration in $\nodeconst'$, implying that, after $\newre$ terminates, $\nodeconst_{\rere(\Pi'')}$ is equal to $\nodeconst'$. 
	Let $\inre{\fL} = \{\gen{\ell_1}, \ldots, \gen{\ell_\Delta}\}$ and $\inre{\fL'} = \{\gen{\ell'_1}, \ldots, \gen{\ell'_\Delta}\}$ be two arbitrary configurations in $\nodeconst'$, let $\phi$ be an arbitrary permutation of $\{1,\ldots,\Delta\}$, and let $u$ be an arbitrary value in $\{1,\ldots,\Delta\}$. Let $\infp{\fL} = \{\ell_1, \ldots, \ell_\Delta\}$ and $\infp{\fL'} = \{\ell'_1, \ldots, \ell'_\Delta\}$.
	Let $\inre{\fC}$ be the RE-combination of $\inre{\fL}$ and $\inre{\fL'}$ obtained w.r.t.\  $\phi$ and $u$. Let $\infp{\fC} = \{\ell''_1, \ldots, \ell''_\Delta\}$ be the combination of $\infp{\fL}$ and $\infp{\fL'}$ obtained w.r.t.\ $\phi$ and $u$ in procedure $\fpp$ by using diagram $D$ when constructing $\nodeconst$. Observe that $\infp{\fC}$, or another configuration dominating it w.r.t.\ diagram $D$, must be present in $\nodeconst$ (by the definition of $\nodeconst$), and hence that $\fC = \{\gen{\ell''_1}, \ldots, \gen{\ell''_\Delta}\}$, or another configuration RE-dominating it, must be present in $\gen{\nodeconst}$ and hence in $\nodeconst'$.  We show that $\fC$ RE-dominates $\inre{\fC}$, and hence that $\newre$ does not generate new configurations. W.l.o.g.\ let $\phi$ be the identity function, and $u = 1$.
	By definition, $\inre{\fC} = \{\gen{\ell_1} \cup \gen{\ell'_1}, \gen{\ell_2} \cap \gen{\ell'_2}, \ldots, \gen{\ell_\Delta} \cap \gen{\ell'_\Delta}\}$ and $\infp{\fC} = \{\diagsup{\ell_1}{\ell'_{1}}, \diaginf{\ell_2}{\ell'_{2}}, \ldots, \diaginf{\ell_\Delta}{\ell'_{\Delta}}\}$. Observe that $\gen{\ell_i} \cap \gen{\ell'_i} = \gen{\diaginf{\ell_i}{\ell'_{i}}}$, and that $\gen{\ell_1} \cup \gen{\ell'_1} \subseteq \gen{\diagsup{\ell_1}{\ell'_{1}}}$, where $\diaginf{\cdot}{\cdot}$ and $\diagsup{\cdot}{\cdot}$ are taken according to $D$, and $\gen{\cdot}$ is taken according to $D'$. Hence,  $\fC$ RE-dominates $\inre{\fC}$. Thus, $\nodeconst_{\rere(\Pi'')} = \nodeconst'$, which, under the renaming $r$ is equal to $\nodeconst$.
	
	The constraint $\edgeconst_{\rere(\Pi'')}$ is obtained from $\edgeconst$ by replacing each configuration $\{\ell_1, \ldots, \ell_\delta\}$ with the condensed configuration $\{S(\ell_1), \ldots, S(\ell_\delta)\}$, where $S(\ell)$ is the set of sets appearing in $\nodeconst_{\rere(\Pi'')}$ and containing $\ell$. Observe that $S(\ell) = \set{\gen{\ell'} ~|~ \ell' \text{ is a successor of } \ell \text{ according to } D}$, where $\gen{\cdot}$ is taken according to $D'$.
	 Under the renaming $r$, observe that $S(\ell)$ contains all labels $\ell'$ that are successors of $\ell$ in $D$, and hence $S(\ell) = \gen{\ell}$, where this time $\gen{\cdot}$ is taken according to $D$. Hence, under the renaming $r$, each configuration  $\{\ell_1, \ldots, \ell_\delta\}$ of $\edgeconst$ produces the condensed configuration $\{\gen{\ell_1},\ldots,\gen{\ell_\delta}\}$, where $\gen{\cdot}$ is taken according to $D$. Therefore, $\edgeconst_{\rere(\Pi'')} = \gen{\edgeconst}$ (where $\gen{\cdot}$ is taken according to $D$), as required.	
\end{proof}

\subsection{Applying Procedure $\fpp$ Faster}
We now present some shortcuts that we can take when applying the procedure.
\begin{observation}\label{obs:twocomb}
	Given two configurations $\fL = \{\L_1, \ldots, \L_k\}$ and $\fL' = \{\L'_1, \ldots, \L'_k\}$, if $\diagsup{\L_u}{\L'_{\phi(u)}} \in \set{\L_u,\L'_{\phi(u)}}$, combining $\fL$ and $\fL'$ w.r.t.\ $\phi$ and $u$ can only produce configurations dominated by $\fL$ or by $\fL'$.
\end{observation}
\begin{proof}
	Since for all indices $i \neq u$ we apply the $\diaginf{\cdot}{\cdot}$ operator, any such obtained configuration is dominated by $\fL$ or $\fL'$.
\end{proof}
As a special case of \Cref{obs:twocomb} where $\fL = \fL'$, we obtain the following.
\begin{corollary}\label{obs:selfcomb}
	Given a configuration $\fL = \{\L_1, \ldots, \L_k\}$, if for all pairs $1 \le i,j \le k$ it holds that $\L_i \in \mathrm{Succ}(\L_j)$ or $\L_{j} \in \mathrm{Succ}(\L_{i})$, then by combining $\fL$ with itself (in any way) we only obtain configurations dominated by $\fL$. 
\end{corollary}
\begin{observation}\label{obs:betterunions}
	Given two configurations $\fL = \{\L_1, \ldots, \L_k\}$ and $\fL' = \{\L'_1, \ldots, \L'_k\}$, a permutation $\phi$ and an index $1 \le u  \le k$, assume that there exist two indices $j$ and $j'$ satisfying that $\diagsup{\L_u}{\L'_{\phi(u)}} = \diagsup{\L_{j}}{\L'_{j'}}$,  $L_u \in \mathrm{Succ}(\L_{j})$,  $\L'_{\phi(u)} \in\mathrm{Succ}(\L'_{j'})$, and either $\L_u \neq \L_j$ or $\L_{\phi(u)} \neq \L_{j'}$. Then, the combination $\fC$ of $\fL$ and $\fL'$ w.r.t.\ $\phi$ and $u$ is dominated by a combination of $\fL$ and $\fL'$ that does not satisfy this assumption.
\end{observation}
\begin{proof}
	Consider the permutation $\rho$ defined via $\rho(j) := j'$, $\rho(u) := \phi(j)$, $\rho(\phi^{-1}(j')) := \phi(u)$, and $\rho(i) := \phi(i)$ if $i \not\in \{ j,u, \phi^{-1}(j')\}$.
	Denote by $\fC'$ the combination of $\fL$ and $\fL'$ w.r.t.\ $\rho$ and $j$.
	Observe that
	\begin{itemize}
		\item $\diagsup{\L_j}{\L'_{\rho(j)}} = \diagsup{\L_u}{\L'_{\phi(u)}}$,
		\item $\diaginf{\L_u}{\L'_{\rho(u)}} \in \mathrm{Succ}(\diaginf{\L_j}{\L'_{\rho(u)}}) = \mathrm{Succ}(\diaginf{\L_j}{\L'_{\phi(j)}})$ since $\L_u \in \mathrm{Succ}(\L_j)$,
		\item $\diaginf{\L_{\phi^{-1}(j')}}{\L'_{\rho(\phi^{-1}(j'))}} = \diaginf{\L_{\phi^{-1}(j')}}{\L'_{\phi(u)}} \in \mathrm{Succ}(\diaginf{\L_{\phi^{-1}(j')}}{\L'_{j'}})$ since $\L'_{\phi(u)} \in\mathrm{Succ}(\L'_{j'})$,
		\item $\diaginf{\L_i}{\L'_{\rho(i)}} = \diaginf{\L_i}{\L'_{\phi(i)}}$ for all $i \not\in \{ j,u, \phi^{-1}(j')\}$.
	\end{itemize}
	This implies that $\fC'$ dominates $\fC$. If $\rho$ and $j$ (together with $\fL$ and $\fL'$) still satisfy the assumptions given in the lemma, then we recurse. This recursion eventually stops since in each recursion step the two arguments on which the $\diagsup{\cdot}{\cdot}$ operator is applied are replaced by predecessors, and for at least one of the two the predecessor is a strict one.
\end{proof}
When executing procedure $\fpp$, we can ignore the combinations satisfying the premises of at least one of \Cref{obs:twocomb,obs:selfcomb,obs:betterunions}, since they create configurations that are dominated by configurations that are already present or computed in cases not satisfying the premises.

We also observe that, in order to prove that a non-trivial fixed point relaxation $\Pi'$ for a problem $\Pi$ exists, there are two possible strategies:
\begin{itemize}
	\item Prove that, by applying the fixed point procedure on $\Pi$ with diagram $D$, the result is $\Pi'$. Also, prove that $\Pi'$ is not $0$-round-solvable.
	\item Prove that, by applying the fixed point procedure on $\Pi'$ with diagram $D$, the result is $\Pi'$ itself. Also, prove that $\Pi'$ can be solved in $0$ rounds given a solution for $\Pi$, and that $\Pi'$ is not $0$-round-solvable.
\end{itemize}
While the second strategy gives a slightly weaker result, that is, it does not show that one would indeed get $\Pi'$ by starting from $\Pi$ and applying the procedure, if the goal is showing that a non-trivial fixed point relaxation for $\Pi$ exists, the second strategy is sufficient. We summarize this observation as follows.
\begin{observation}\label{obs:closedunderfp}
	Assume that $\Pi'$ can be solved in $0$ rounds given a solution for $\Pi$, that $\Pi'$ cannot be solved in $0$ rounds, and that, by applying the fixed point procedure on $\Pi'$ with diagram $D$, the result is $\Pi'$ itself. Then, $\Pi'$ is a non-trivial fixed point relaxation of $\Pi$.
\end{observation}

\section{Selecting the Right Diagram}\label{sec:diagram}
In this section, we give some intuition on how we can choose a good diagram for applying the procedure $\fpp$, and we examine, as an example, the problem of computing a $1$-defective $2$-coloring in $3$-regular graphs. In this problem, we consider $3$-regular graphs and we require nodes to output either \emph{red} or \emph{blue} such that each node has at most one neighbor with the same color. This problem is known to require $\Omega(\log n)$ rounds for deterministic algorithms and $\Omega(\log \log n)$ rounds for randomized ones by prior work \cite{BalliuHLOS19}. In fact, we do not make any formal claim in this section, and we only use this problem as a running example (more details about this problem are provided in \Cref{sec:2col}).

\paragraph{The default diagram.}
Recall that whether $\fpp$ produces a trivial or a non-trivial fixed point may depend on the choice of the diagram that we give to the procedure. Let $\Pi = (\Sigma_{\Pi},\nodeconst_{\Pi},\edgeconst_{\Pi})$ be the problem on which we want to apply $\fpp$.  We now describe a diagram $D = (\Sigma_D, E_D)$ that can be used as a default choice. We set $\Sigma_D$ to be the set of right-closed subsets w.r.t.\ the edge diagram of $\Pi$ (where we also identify $\gen{\ell} \in \Sigma_D$ with $\ell \in \Sigma_\Pi)$, and we take as edges the ones given by the strict superset relation, i.e., $E_D = \set{(\L_1,\L_2) ~|~ \L_2 \subsetneq \L_1}$. The high level idea here is that $\Sigma_D$ is a superset of the labels that would be generated if we apply the universal quantifier to $\edgeconst_{\Pi}$. Diagram $D$ clearly satisfies the requirements on the diagram specified by procedure $\fpp$, because $\diagsup{\L_1}{\L_2} = \L_1 \cap \L_2$, and the intersection of two right-closed subsets is again a right-closed subset (a similar statement holds for $\diaginf{\cdot}{\cdot}$ as well).

\paragraph{Improving our choice.}
If the default diagram leads to a trivial fixed point, we can tweak the diagram and hope that, by applying procedure $\fpp$ with the new diagram, we obtain a better result, i.e., a non-trivial fixed point.
One possible way to improve the diagram is the following. We consider the trivial fixed point obtained via the default diagram and extract the configurations that allow us to solve it in $0$ rounds (i.e., the configurations $\fC$ satisfying that, if every node outputs $\fC$ without any coordination with its neighbors, then the obtained solution is correct). Consider one such configuration, $\fL = \{\L_1, \ldots, \L_\delta\}$. For each $\L_i$, we can build a tree that represents the expression that we computed in order to obtain $\L_i$, where leaves are labels of some initial configurations, and each internal node is either a $\diagsup{\cdot}{\cdot}$ operation or an $\diaginf{\cdot}{\cdot}$ operation. The idea is then to \emph{add} nodes (and edges) to the diagram such that, by computing the same expression, we obtain a different result, and in particular a configuration which does not contribute anymore to making the problem $0$-round-solvable. While this description is very vague, we now provide a concrete example.

\paragraph{The example problem.} One way to encode the problem of computing a $1$-defective $2$-coloring in $3$-regular graphs as a node-edge checkable problem is the following (we will provide more details in \Cref{sec:2col}):
\begin{equation*}
\begin{aligned}
\begin{split}
\nodeconst_{\Pi}\text{:}\\
&\A &\s& \A &\s& [\A \s \X] \\
&\B &\s& \B &\s& [\B \s \Y] \\
\end{split}
\qquad
\begin{split}
\edgeconst_{\Pi}\text{:}\\
&[\A \s \X] &\s& [\B \s \Y]\\
&[\X \s \Y] &\s& [\X \s \Y]\\
\end{split}
\end{aligned}
\end{equation*}

\paragraph{The default diagram for the example problem.}
\newcommand{\emptybox}{\bbempty}
If we compute edge the diagram of $\Pi$, we obtain that its nodes are $\set{\A,\B,\X,\Y}$ and its edges are $\set{(\A,\X),(\B,\Y)}$. Hence, the set of right-closed subsets is $\set{\emptyset, \set{\X}, \set{\Y},\set{\X,\Y}, \set{\A,\X},\set{\B,\Y},\set{\A,\X,\Y},\set{\B,\X,\Y},\set{\A,\B,\X,\Y}}$. In order to build the diagram $D$ for applying $\fpp$, we rename these sets as follows. The empty set becomes $\mybox{\phantom{E}}$, each set $\gen{\ell}$ becomes $\mybox{\ell}$ (for example the set $\set{\A,\X}$ is actually equal to $\gen{\A}$, so it becomes $\bA$, which we identify with the original label $\A$), and all the other sets $\set{\ell_1,\ldots,\ell_k}$ become $\mybox{\ell_1\ldots\ell_k}$. We put an edge from label $\ell_1$ to label $\ell_2$ if $\ell_2 \subsetneq \ell_1$. The obtained diagram is shown in \Cref{fig:basediag} (edges that can be obtained by transitivity are omitted).
\begin{figure}
	\centering
	\includegraphics[width=0.5\textwidth]{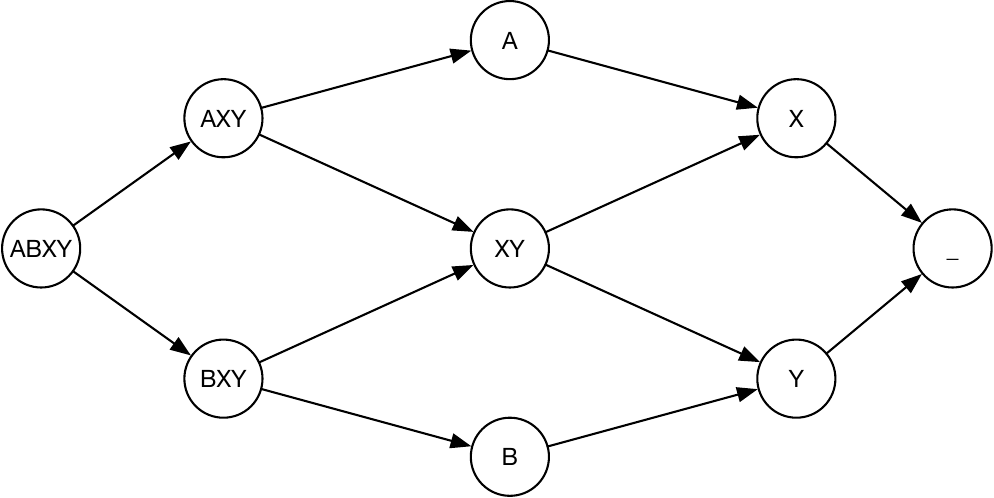}
	\caption{The default diagram for the problem of computing a $1$-defective $2$-coloring in $3$-regular graphs. The symbols $\mybox{\cdot}$ are omitted for clarity, and \_ represents the empty set. Edges that can be obtained via transitivity are omitted.}
	\label{fig:basediag}
\end{figure}

\paragraph{Computing $\fpp(\Pi,D)$.}
Let us start by computing $\fpproc{\nodeconst_{\Pi}}{D}$. By keeping only maximal configurations, we obtain that the starting configurations are $\bA \s \bA \s \bX$ and $\bB \s \bB \s \bY$.
By \Cref{obs:selfcomb}, the only combinations that we need to perform are those that combine $\bA \s \bA \s \bX$ with $\bB \s \bB \s \bY$. There are five ways to combine them:
\begin{itemize}
	\item $\diagsup{\bA}{\bB} \s \diaginf{\bA}{\bB} \s \diaginf{\bX}{\bY} = \emptybox \s \bABXY \s \bXY$.  
	\item $\diagsup{\bA}{\bB} \s \diaginf{\bA}{\bY} \s \diaginf{\bX}{\bB} = \emptybox \s \bAXY \s \bBXY$.
	\item $\diaginf{\bA}{\bB} \s \diaginf{\bA}{\bB} \s \diagsup{\bX}{\bY} = \bABXY \s \bABXY \s \emptybox$. Observe that this configuration is dominated by the first one.
	\item $\diaginf{\bA}{\bB} \s \diagsup{\bA}{\bY} \s \diaginf{\bX}{\bB} = \bABXY \s \emptybox \s \bBXY$. Observe that this configuration is dominated by the first one.
	\item $\diaginf{\bA}{\bB} \s \diaginf{\bA}{\bY} \s \diagsup{\bX}{\bB} = \bABXY \s \bAXY \s \emptybox$. Observe that this configuration is dominated by the first one.
\end{itemize}
Hence, we obtain the configurations $\bA \s \bA \s \bX$, $\bB \s \bB \s \bY$, $\bABXY \s \bXY \s \emptybox$, and $\bAXY \s \bBXY \s \emptybox$. The next step is to combine the new configurations with each other and themselves, and the new configurations with the old ones. By repeating this process until nothing new is obtained (discarding, at each step, non-maximal configurations), and by then also computing  $\fpproc{\edgeconst_{\Pi}}{D'}$, where $D'$ is obtained from $D$ by reversing the direction of the edges, we obtain the problem $\Pi' = (\Sigma_D,\nodeconst,\gen{\edgeconst})$ described as follows.
\begin{equation*}
\begin{aligned}
\begin{split}
\nodeconst_{\Pi}\text{:}\\
&\bA &\s& \bA &\s& \bX \\
&\bB &\s& \bB &\s& \bY \\
&\bABXY &\s& \bXY &\s& \emptybox \\
&\bAXY &\s& \bBXY &\s& \emptybox \\
&\bAXY &\s& \bX &\s& \bX \\
&\bBXY &\s& \bY &\s& \bY \\
&\bXY &\s& \bXY &\s& \bXY \\
\end{split}
\qquad
\begin{split}
\edgeconst_{\Pi}\text{:}\\
&[\bA \bX \emptybox] &\s& [\bB \bY \emptybox]\\
&\emptybox &\s& [\bABXY \bAXY \bBXY \bXY \bA \bB \bX \bY \emptybox]\\
&[\bY \emptybox] &\s& [\bAXY \bA \bXY \bX \bY \emptybox]\\
&[\bX \emptybox] &\s& [\bBXY \bB \bXY \bX \bY \emptybox]\\
&[\bXY \bX \bY \emptybox] &\s& [\bXY \bX \bY \emptybox]\\\\\\
\end{split}
\end{aligned}
\end{equation*}
We can now observe that the configuration $\bXY \s \bXY \s \bXY$ makes problem $\Pi'$ $0$-rounds-solvable; hence if we want to obtain a non-trivial fixed point relaxation for $\Pi$ we need to tweak the diagram and apply $\fpp$ again. 

By keeping track of the combinations performed to obtain such a configuration, we obtain the combinations depicted in \Cref{fig:xy} (note that there may be different ways to obtain such a configuration, and this is just an example).
\begin{figure}
	\centering
	\includegraphics[width=0.7\textwidth]{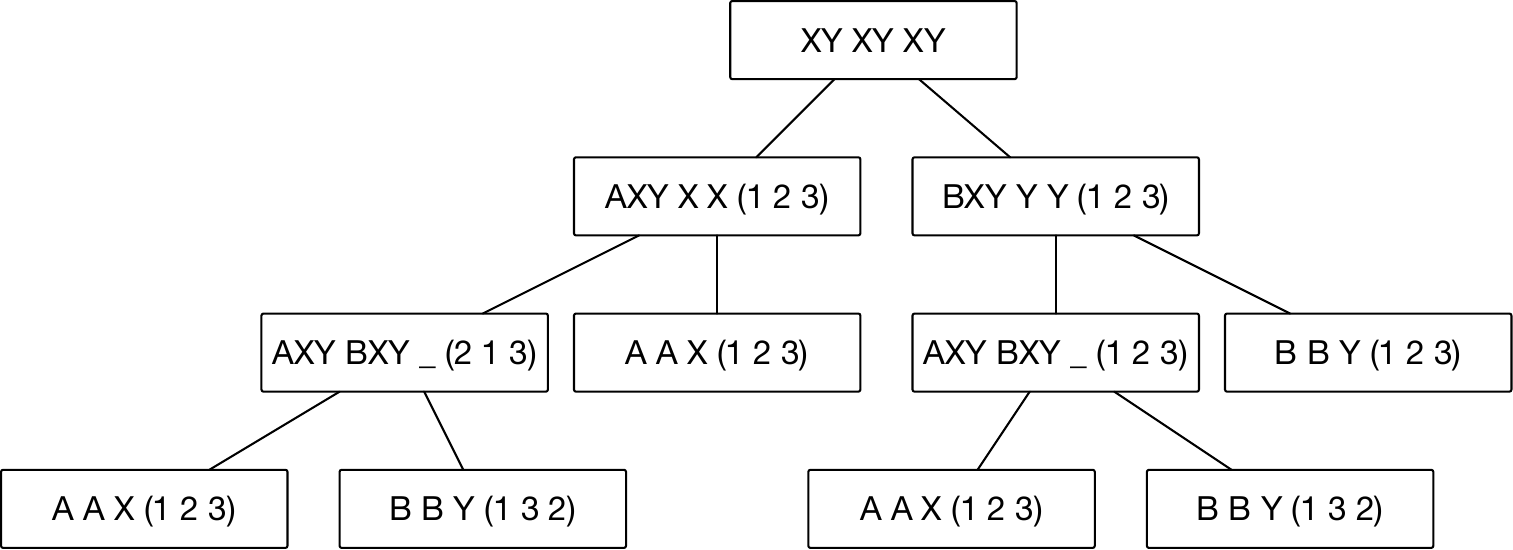}
	\caption{A collection of combinations that can be used to obtain the configuration $\bXY \s \bXY \s \bXY$. The numbers represent how to match the labels of the two configurations to obtain the parent one. Labels in position $1$ are those where we apply the $\diagsup{\cdot}{\cdot}$ operator. For example, combining $\bAXY \s \bBXY \s \emptybox ~ (2 1 3)$ and $\bA \s \bA \s \bX ~ (1 2 3)$ means taking configuration $\diagsup{\bBXY}{\bA} \s \diaginf{\bAXY}{\bA} \s \diaginf{\emptybox}{\bX} = \bX \s \bAXY \s \bX$.}
	\label{fig:xy}
\end{figure}
If we now write a separate expression for each resulting $\bXY$, we obtain that the three resulting $\bXY$ are obtained as follows.
\begin{enumerate}
	\item $\diaginf{\diaginf{\bX}{\diagsup{\bA}{\bB}}}{\diaginf{\bY}{\diagsup{\bA}{\bB}}}$
	\item $\diagsup{\diaginf{\bA}{\diaginf{\bA}{\bY}}}{\diaginf{\bB}{\diaginf{\bB}{\bX}}} = \diagsup{\diaginf{\bA}{\bY}}{\diaginf{\bB}{\bX}}$
	\item $\diaginf{\diagsup{\bA}{\diaginf{\bB}{\bX}}}{\diagsup{\bB}{\diaginf{\bA}{\bY}}}$
\end{enumerate}
We now replace label $\bXY$ with two copies of it, $\bXY$ and $\mybox{\X\Y'}$, and connect these two labels with the others as depicted in \Cref{fig:example-valid-diagram}. The idea here is that if an expression produces $\bXY$ by performing a $\diagsup{\cdot}{\cdot}$ operation on nodes that are predecessors of $\bXY$, then with the new diagram we are going to obtain $\mybox{\X\Y'}$ as a result, instead of $\bXY$. We then have to hope that there are no alternative ways to combine configurations that produce $\bXY \s \bXY \s \bXY$ anyways, and that $\bXY$ and $\mybox{\X\Y'}$ do not become self compatible and pairwise compatible on the edge constraint (i.e., the configurations $\bXY \s \bXY$, $\bXY \s \mybox{\X\Y'}$, and $\mybox{\X\Y'} \s \mybox{\X\Y'}$ are not all contained in the edge constraint), as it would defeat the purpose of distinguishing them.
By applying again $\fpp$, this time with the diagram of \Cref{fig:example-valid-diagram}, we obtain the following:
\begin{equation*}
\begin{aligned}
\begin{split}
\nodeconst_{\Pi}\text{:}\\
&\bA &\s& \bA &\s& \bX \\
&\bB &\s& \bB &\s& \bY \\
&\bABXY &\s& \bXY &\s& \emptybox \\
&\bAXY &\s& \bBXY &\s& \emptybox \\
&\bAXY &\s& \bX &\s& \bX \\
&\bBXY &\s& \bY &\s& \bY \\
&\mybox{\X\Y'} &\s& \bXY &\s& \bXY \\
\end{split}
\qquad
\begin{split}
\edgeconst_{\Pi}\text{:}\\
&[\bA \bX \emptybox] &\s& [\bB \bY \emptybox]\\
&\emptybox &\s& [\bABXY \bAXY \bBXY \bXY \mybox{\X\Y'} \bA \bB \bX \bY \emptybox]\\
&[\bY \emptybox] &\s& [\bAXY \bA \bXY \mybox{\X\Y'} \bX \bY \emptybox]\\
&[\bX \emptybox] &\s& [\bBXY \bB \bXY \mybox{\X\Y'} \bX \bY \emptybox]\\
&[\bXY \bX \bY \emptybox] &\s& [\mybox{\X\Y'} \bXY \bX \bY \emptybox]\\\\\\
\end{split}
\end{aligned}
\end{equation*}
It is straightforward to check that this problem is not $0$-round-solvable, and hence the tweaking of the diagram succeeded.

\section{An Alternative Proof for the Hardness of $\Delta$-Coloring}\label{sec:deltacol}
In this section, we provide a first application of procedure $\fpp$. We revisit a result proved in \cite{hideandseek}, that is, a non-trivial fixed point relaxation for the $\Delta$-coloring problem, and we show that, by using $\fpp$, it is much easier to prove this result. To this end, we apply \Cref{obs:closedunderfp}, i.e., we provide a problem that can be solved in $0$ rounds given a $\Delta$-coloring, we show that by applying $\fpp$ to this problem we obtain the problem itself, and we show that this problem is not $0$-round-solvable. This problem is exactly the one provided in \cite{hideandseek}.

\paragraph{The problem $\Pi_\Delta$.}
We now define this problem $\Pi_\Delta$ that we will prove to be a non-trivial fixed point relaxation of the $\Delta$-coloring problem. Assume that the colors are $\ccs = \{1, \ldots, \Delta\}$.
The set $\Sigma_{\Pi_{\Delta}}$ of labels of this problem is the set of all the possible subsets of $\ccs$, that is, for each $\ccc \in 2^\ccs$, we have the label $\ccc$. For stylistic reasons, we may represent the label  $\ccc = \{c_1,\ldots,\c_k\}$ as $\mybox{c_1 \ldots c_k}$.
The node constraint $\nodeconst_{\Pi_{\Delta}}$ of $\Pi_\Delta$ contains the following configurations:
\[
\ccc^{\Delta - k + 1} \s \bbempty^{k - 1}, \text{ for all $\ccc \in 2^\ccs \setminus \{\emptyset\}$, where $k = |\ccc|$. }
\]
The edge constraint $\edgeconst_{\Pi_{\Delta}}$ of $\Pi_\Delta$ contains the following configurations:
\[
\ccc_1 \s \ccc_2, \text{for all $\ccc_1,\ccc_2 \in 2^\ccs$ such that $\ccc_1 \cap \ccc_2 = \emptyset$.}
\]

\paragraph{An example for $\Delta=3$.}
We now provide an explicit example of $\Pi_\Delta$ for $\Delta=3$.
\begin{equation*}
	\begin{aligned}
		\begin{split}
			\nodeconst_{\Pi_{3}}\text{:}\\
			&\mybox{1} &\s& \mybox{1} &\s& \mybox{1} \\
			&\mybox{2} &\s& \mybox{2} &\s& \mybox{2} \\
			&\mybox{3} &\s& \mybox{3} &\s& \mybox{3} \\
			&\mybox{12} &\s& \mybox{12} &\s& \bbempty \\
			&\mybox{13} &\s& \mybox{13} &\s& \bbempty \\
			&\mybox{23} &\s& \mybox{23} &\s& \bbempty \\
			&\mybox{123} &\s& \bbempty &\s& \bbempty \\\\
		\end{split}
		\qquad
		\begin{split}
			\edgeconst_{\Pi_{3}}\text{:}\\
			&\bbempty &\s& [\bbempty\mybox{1}\mybox{2}\mybox{3}\mybox{12}\mybox{23}\mybox{13}\mybox{123}] \\
			&\mybox{1} &\s& [\bbempty\mybox{2}\mybox{3}\mybox{23}]\\
			&\mybox{2} &\s& [\bbempty\mybox{1}\mybox{3}\mybox{13}]\\
			&\mybox{3} &\s& [\bbempty\mybox{1}\mybox{2}\mybox{12}]\\
			&\mybox{12} &\s& [\bbempty\mybox{3}]\\
			&\mybox{13} &\s& [\bbempty\mybox{2}]\\
			&\mybox{23} &\s& [\bbempty\mybox{1}]\\
			&\mybox{123} &\s& \bbempty\\
		\end{split}
	\end{aligned}
\end{equation*}

\paragraph{The claim.}
We devote the rest of the section to proving the following statement.
\begin{theorem}\label{thm:deltacol}
	The problem $\Pi_\Delta$ is a non-trivial fixed point relaxation of the $\Delta$-coloring problem.
\end{theorem}
Note that $\Pi_\Delta$ is not $0$-round-solvable, since all configurations allowed by the node constraint contain at least one label $\ccc\neq \bbempty$, and $\ccc \s \ccc$ is not contained in $\edgeconst_{\Pi_{\Delta}}$. Also, note that the problem is clearly solvable in $0$ rounds given a $\Delta$-coloring, since it is enough for a node of color $c$ to output the configuration $\mybox{c}^\Delta$. We now apply subprocedure $\fpprocname$ on the node constraint $\nodeconst_{\Pi_{\Delta}}$, by using the diagram $D = (\Sigma_D,E_D)$ defined as $\Sigma_D = \Sigma_{\Pi_{\Delta}}$ and $E_D = \{ (\ell_1,\ell_2) ~|~ \ell_2 \subsetneq \ell_1 \}$, and we show that the result is the node constraint itself. Later, we will apply subprocedure $\fpprocname$ on the edge constraint $\edgeconst_{\Pi_{\Delta}}$, using the diagram $D'$ obtained by flipping the edges of $D$, and we will show that the resulting constraint $\edgeconst$ satisfies $\gen{\edgeconst} = \edgeconst_{\Pi_{\Delta}}$.

\paragraph{The node constraint.}
Consider two arbitrary allowed configurations in $\nodeconst_{\Pi_{\Delta}}$, $\fL_1 = \ccc_1^{\Delta - k_1 + 1} \s \bbempty^{k_1 - 1}$ and $\fL_2 = \ccc_2^{\Delta - k_2 + 1} \s \bbempty^{k_2 - 1}$, where $k_1 = |\ccc_1|$ and $k_2 = |\ccc_2|$. By \Cref{obs:selfcomb}, we can restrict our attention to the case $\ccc_1 \neq \ccc_2$, and by \Cref{obs:twocomb}, we can restrict our attention to the case in which the $\diagsup{\cdot}{\cdot}$ operator is applied to $\ccc_1$ and $\ccc_2$. Hence, by combining $\fL_1$ and $\fL_2$ we obtain a configuration of the following form: $\fC = \diagsup{\ccc_1}{\ccc_2} \s \diaginf{\ccc_1}{\ccc_2}^a \s \diaginf{\ccc_1}{\bbempty}^b \s \diaginf{\bbempty}{\ccc_2}^c \s \diaginf{\bbempty}{\bbempty}^d = (\ccc_1 \cap \ccc_2) \s (\ccc_1 \cup \ccc_2)^a \s \ccc_1^b \s \ccc_2^c \s \bbempty^d$, where $a+b+c+d+1 = \Delta$, $a+b+1 = \Delta - k_1 + 1$, and $a +c+1 = \Delta - k_2 + 1$.

Let $\ccc_{\cup} = \ccc_1 \cup \ccc_2$, and let $k_{\cup} = |\ccc_{\cup}|$. First, consider the case $a \ge \Delta - k_{\cup} + 1$. In this case, $\fC$ is dominated by the configuration $\ccc_{\cup}^{\Delta - k_{\cup} + 1} \s \bbempty^{k_{\cup} - 1}$. Hence, we are left with the case  $a \le \Delta - k_{\cup}$. We prove that, in this case, $\fC$ is dominated by the configuration $\ccc_{\cap}^{\Delta - k_{\cap} + 1} \s \bbempty^{k_{\cap} - 1}$, where $\ccc_{\cap} = \ccc_1 \cap \ccc_2$ and $k_{\cap} = |\ccc_{\cap}|$. This domination holds if $a+b+c+1 \ge \Delta - k_{\cap} + 1$, and hence we now prove that this inequality holds.
Observe that $k_{\cup} + k_{\cap}  = k_1 + k_2 = 2 \Delta - 2a - b - c \ge 2 \Delta - (\Delta - k_{\cup}) - a - b - c = \Delta  + k_{\cup} - a -b -c$. This implies that $k_{\cap} \ge  \Delta - a -b -c$, which implies the claim.

\paragraph{The edge constraint.}
The first step in computing $\edgeconst$ is taking all maximal configurations of $\edgeconst_{\Pi_{\Delta}}$. By the definition of $D'$, it is easy to see that such configurations are exactly those of the form $\ccc_1 \s \ccc_2$ where $\ccc_1 \cap \ccc_2 = \emptyset$ and $\ccc_1 \cup \ccc_2 = \ccs$. Now, consider two arbitrary configurations in $\edgeconst$, $\fL_1 = \ccc_{1,1} \s \ccc_{1,2}$ and $\fL_2 = \ccc_{2,1} \s \ccc_{2,2}$.  W.l.o.g., assume that the combination is $\fC = \diaginf{\ccc_{1,1}}{\ccc_{2,1}} \s \diagsup{\ccc_{1,2}}{\ccc_{2,2}} =  (\ccc_{1,1} \cap \ccc_{2,1}) \s (\ccc_{1,2} \cup \ccc_{2,2})$. Since $\ccc_{1,1} \cap \ccc_{1,2} = \emptyset$ and $\ccc_{2,1} \cap \ccc_{2,2} = \emptyset$, it holds that $(\ccc_{1,1} \cap \ccc_{2,1}) \cap (\ccc_{1,2} \cup \ccc_{2,2}) = \emptyset$, and hence $\fC$ is dominated by $(\ccc_{1,1} \cap \ccc_{2,1}) \s (\ccs \setminus (\ccc_{1,1} \cap \ccc_{2,1})) \in \edgeconst$. This implies that no new configurations are added to $\edgeconst$. Finally, observe that $\gen{\edgeconst}$ gives exactly the configurations in $\edgeconst_{\Pi_{\Delta}}$.

\section{An Alternative Proof for the Hardness of Defective $2$-Coloring}\label{sec:2col}
One of the major open questions about round elimination is whether there exists a non-trivial fixed point relaxation for each problem that requires $\Omega(\log_\Delta n)$ rounds \cite{trulytight,hideandseek}. There are very few locally checkable problems that are known to require $\Omega(\log_\Delta n)$ rounds and for which a non-trivial fixed point relaxation is not known, and one of them is the problem of coloring the nodes of a graph with $2$ colors, such that each node has at least $2$ neighbors of a different color than itself. This problem is known to require $\Omega(\log_\Delta n)$ rounds \cite{BalliuHLOS19}. We now provide an alternative proof by providing a non-trivial fixed point relaxation for the problem. In the following, we refer to this problem as \emph{defective $2$-coloring}.
To obtain the new proof, we apply \Cref{obs:closedunderfp}, i.e., we provide a problem $\Pi_\Delta$ that can be solved in $0$ rounds given a defective $2$-coloring, we show that $\Pi_\Delta$ is not $0$-round-solvable, and that by applying $\fpp$ to $\Pi_\Delta$  we obtain $\Pi_\Delta$ itself.
We remark that parts of some proofs presented in \Cref{sec:3col} are based on proofs provided in this section.

\paragraph{The problem $\Pi_\Delta$.}
The set of labels is $\Sigma_{\Pi_{\Delta}} = \{\bbempty, \bbX, \bbY, \bbXY, \bbXYp, \bbAX, \bbBY, \bbAXYp,\bbBXYp,\allowbreak \bbABXYp\}$. The constraints are defined as follows.
\begin{equation*}
	\begin{aligned}
		\begin{split}
			&\nodeconst_{\Pi_\Delta}\text{:} &&&&&&& &\edgeconst_{\Pi_\Delta}\text{:}\\
			&\bbX^{\Delta - 2} & & \bbAX^{2}   &&&&& &[\bbBY\bbY\bbempty]    \s\s [\bbAX\bbX\bbempty]\\
			&\bbY^{\Delta - 2} & & \bbBY^{2}   &&&&& &[\bbABXYp \bbAXYp \bbBXYp \bbAX \bbBY \bbXYp \bbXY \bbX \bbY \bbempty]  \s\s \bbempty \\
			&  \bbX^{\Delta - 1}  & &  \bbAXYp &&&&& 	&[\bbAXYp \bbAX \bbXYp \bbXY \bbX \bbY \bbempty]   \s\s [\bbY \bbempty] \\
			&  \bbY^{\Delta-1}  & &  \bbBXYp   &&&&& 	&[\bbBXYp \bbBY \bbXYp \bbXY \bbX \bbY \bbempty]   \s\s [\bbX \bbempty] \\
			& \bbempty   & &  \bbXY^{\Delta-3} & & \bbAXYp & & \bbBXYp  &~~~~~~~~~~&[\bbXYp \bbXY \bbX \bbY \bbempty]   \s\s [\bbXY \bbX \bbY \bbempty] \\
			&  \bbempty  & &  \bbXY^{\Delta-2} & &   \bbABXYp  \\
			& \bbXY^{\Delta-1}   & &   \bbXYp  \\
			\\
		\end{split}
	\end{aligned}
\end{equation*}

\paragraph{The claim.}
We devote the rest of the section to proving the following statement.
\begin{theorem}
	The problem $\Pi_\Delta$ is a non-trivial fixed point relaxation of the defective $2$-coloring problem.
\end{theorem}
In the following, we will consider each label $\mybox{\ell_1\ldots\ell_k} \in \Sigma_{\Pi_{\Delta}}$ also as the set $\set{\ell_1,\ldots,\ell_k}$.
Note that all configurations contained in the node constraint contain at least one label $\L$ satisfying $\L \cap \{\A,\B,+\} \neq \emptyset$. From the edge constraint we can observe that all such labels satisfy that $\L \s \L$ is not in $\edgeconst_{\Pi_\Delta}$. Hence, the problem is not solvable in $0$ rounds.

Assume we are given a solution for the defective $2$-coloring problem, where each node already knows which neighbors have the same color as them (which can be inferred in just $1$ round of communication). We show that this solution can be converted in $0$ rounds into a solution for $\Pi_\Delta$. Assume that each node is either \emph{red} or \emph{blue}. Each red node outputs the configuration $\bbAX^{2} \s \bbX^{\Delta-2}$, by putting the label $\bbAX$ on two arbitrary edges connecting it to blue neighbors (which are guaranteed to exist) and $\bbX$ on all other incident edges. Blue nodes act similarly by using the configuration  $\bbBY^{2} \s \bbY^{\Delta-2}$. It is easy to see that this labeling satisfies the edge constraint.

We define the diagram $D = (\Sigma_D,E_D)$ via $\Sigma_D = \Sigma_{\Pi_{\Delta}}$ and $E_D = \{ (\ell_1,\ell_2) ~|~ \ell_2 \subsetneq \ell_1 \}$ (see \Cref{fig:diagbasecase}). The diagram $D'$ is obtained by flipping the edges of $D$.
\begin{figure}
	\centering
	\includegraphics[width=0.5\textwidth]{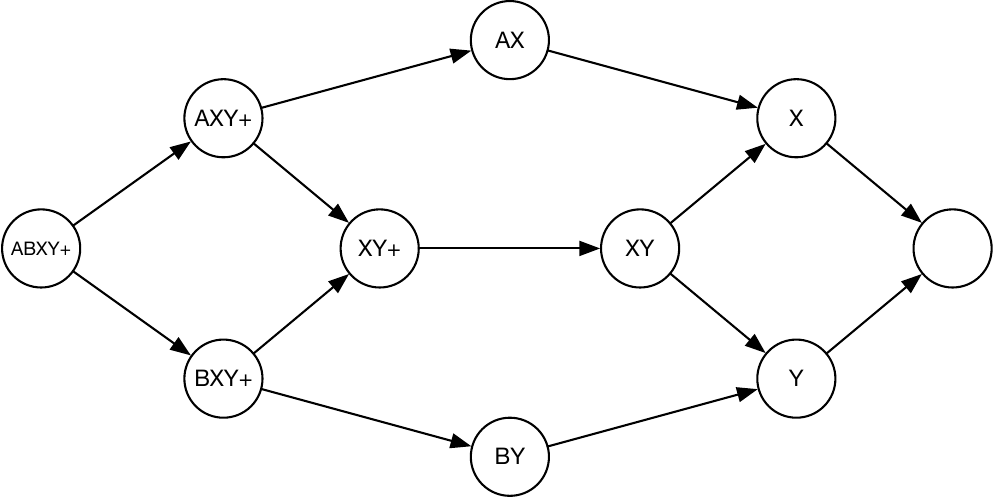}
	\caption{Diagram $D$ for the case of defective $2$-coloring. Edges that can be obtained via transitivity are omitted.}
	\label{fig:diagbasecase}
\end{figure}

\paragraph{The node constraint.}
We now apply subprocedure $\fpprocname$ on the node constraint, by using the diagram $D$. All configurations of $\nodeconst_{\Pi_\Delta}$ are maximal, and we now show that for any pair of configurations it holds that all their combinations are in $\nodeconst_{\Pi_\Delta}$.
We provide the list of allowed configurations here again for reference.
\begin{enumerate}
	\item $\bbX^{\Delta - 2} \s \bbAX^{2}$
	\item $\bbY^{\Delta - 2} \s \bbBY^{2}$
	\item $\bbX^{\Delta - 1} \s  \bbAXYp$
	\item $\bbY^{\Delta-1} \s \bbBXYp$
	\item $\bbempty  \s  \bbXY^{\Delta-3} \s \bbAXYp \s \bbBXYp$
	\item $\bbempty  \s  \bbXY^{\Delta-2} \s   \bbABXYp$
	\item $\bbXY^{\Delta-1}  \s  \bbXYp$
\end{enumerate}
\begin{itemize}
	\item \textbf{1 with 1.} By \Cref{obs:selfcomb}, this case can be discarded.
	\item \textbf{1 with 2, by taking the supremum between $\bbAX$ and $\bbBY$}. $\diagsup{\bbAX}{\bbBY} = \bbempty$. We obtain all configurations of the form $\bbempty \s \bbXY^a \s \bbAXYp^b \s \bbBXYp^c \s \bbABXYp^d$, where $1+a+b+c+d = \Delta$, $a+c = \Delta-2$, $a+b = \Delta-2$, $c+d=1$, and $b+d=1$. If $d \ge 1$, then the configuration is dominated by configuration 6, otherwise we get that $b = c = 1$, and hence the configuration is dominated by configuration 5.
	\item \textbf{1 with 2, all other cases.} By \Cref{obs:betterunions}, these cases can be discarded.
	\item \textbf{1 with 3.} By \Cref{obs:twocomb}, this case can be discarded.
	\item \textbf{1 with 4, by taking the supremum between $\bbAX$ and $\bbY$}. $\diagsup{\bbAX}{\bbY} = \bbempty$. We obtain all configurations of the form $\bbempty \s \bbXY^a \s \bbAXYp^b \s \bbBXYp^c \s \bbABXYp^d$, where $1+a+b+c+d = \Delta$, $a+c = \Delta-2$, $a+b = \Delta-2$, $c+d=1$, and $b+d=1$. If $d \ge 1$, then the configuration is dominated by configuration 6, otherwise we get that $b = c = 1$, and hence the configuration is dominated by configuration 5.
	\item \textbf{1 with 4, by taking the supremum between $\bbX$ and $\bbY$}. By \Cref{obs:betterunions}, this case can be discarded.
	\item \textbf{1 with 4, by taking the supremum between $\bbX$ and $\bbBXYp$}. By \Cref{obs:twocomb}, this case can be discarded.
	\item \textbf{1 with 4, by taking the supremum between $\bbAX$ and $\bbBXYp$}. $\diagsup{\bbAX}{\bbBXYp} = \bbX$. The obtained configuration is $\bbX \s  \bbXY^{\Delta-2} \s \bbAXYp$, which is dominated by configuration 3.
	\item \textbf{1 with 5, by taking the supremum between $\bbAX$ and $\bbBXYp$}. $\diagsup{\bbAX}{\bbBXYp} = \bbX$. Since all sets of configuration $1$ contain $\X$, the result must be dominated by $\bbX^2  \s  \bbXY^{\Delta-3} \s \bbAXYp$, which is dominated by configuration 3.
	\item \textbf{1 with 5, by taking the supremum between $\bbAX$ and $\bbXY$}. By \Cref{obs:betterunions}, this case can be discarded.
	\item \textbf{1 with 5, all other cases.} By \Cref{obs:twocomb}, these cases can be discarded.
	\item \textbf{1 with 6, by taking the supremum between $\bbAX$ and $\bbXY$}. $\diagsup{\bbAX}{\bbXY} = \bbX$. Since all sets of configuration $1$ contain $\X$, the result must be dominated by $\bbX^2  \s  \bbXY^{\Delta-3} \s \bbABXYp$, which is dominated by configuration 3.
	\item \textbf{1 with 6, all other cases.} By \Cref{obs:twocomb}, these cases can be discarded.
	\item \textbf{1 with 7, by taking the supremum between $\bbAX$ and $\bbXYp$}. $\diagsup{\bbAX}{\bbXYp} = \bbX$. The obtained configuration is $\bbX \s \bbXY^{\Delta-2} \s \bbAXYp$, which is dominated by configuration 3.
	\item \textbf{1 with 7, all other cases.} By \Cref{obs:twocomb} and \Cref{obs:betterunions}, these cases can be discarded.
	\item \textbf{2 with anything.} Configuration 2 is symmetric to configuration 1 (by exchanging $\A$ and $\B$, and $\X$ and $\Y$).
	\item \textbf{3 with 3.} By \Cref{obs:selfcomb}, this case can be discarded.
	\item \textbf{3 with 4, by taking the supremum between $\bbX$ and $\bbY$}. $\diagsup{\bbX}{\bbY} = \bbempty$. The obtained configurations are either 5 or 6.
	\item \textbf{3 with 4, by taking the supremum between $\bbAXYp$ and $\bbBXYp$}. $\diagsup{\bbAXYp}{\bbBXYp} = \bbXYp$. The obtained configuration is $\bbXYp \s \bbXY^{\Delta-1}$, which is configuration 7.
	\item \textbf{3 with 4, all other cases.} By \Cref{obs:twocomb}, these cases can be discarded.
	\item \textbf{3 with 5, by taking the supremum between $\bbAXYp$ and $\bbBXYp$}. $\diagsup{\bbAXYp}{\bbBXYp} = \bbXYp$. The result is $\bbXYp \s \bbX \s  \bbXY^{\Delta-3} \s \bbAXYp$, which is dominated by configuration 3.
	\item \textbf{3 with 5, all other cases.} By \Cref{obs:twocomb}, these cases can be discarded.
	\item \textbf{3 with 6.} By \Cref{obs:twocomb}, this case can be discarded.
	\item \textbf{3 with 7.} By \Cref{obs:twocomb}, this case can be discarded.
	\item \textbf{4 with anything.} Configuration 4 is symmetric to configuration 3.
	\item \textbf{5 with 5, by taking the supremum between $\bbAXYp$ and $\bbBXYp$}. $\diagsup{\bbAXYp}{\bbBXYp} = \bbXYp$. The result is dominated by either  $\bbXYp \s \bbempty \s  \bbXY^{\Delta-4} \s \bbAXYp \s \bbBXYp$ or $\bbXYp \s \bbempty \s  \bbXY^{\Delta-3} \s \bbABXYp$, which are dominated by configuration 5 and configuration 6, respectively.
	\item \textbf{All other cases.} By \Cref{obs:selfcomb} and \Cref{obs:twocomb}, these cases can be discarded.
\end{itemize}

\paragraph{The edge constraint.}
We now apply subprocedure $\fpprocname$ on the edge constraint, by using diagram $D'$. The first step in computing $\edgeconst$ is taking all maximal configurations of $\edgeconst_{\Pi_{\Delta}}$. By the definition of $D'$, we obtain the following configurations:
\begin{enumerate}
	\item $\bbAX \s \bbBY$
	\item $\bbABXYp \s \bbempty$
	\item $\bbAXYp \s \bbY$
	\item $\bbBXYp \s \bbX$
	\item $\bbXYp \s \bbXY$
\end{enumerate}
We now consider all possible combinations, ignoring symmetric cases, and ignoring cases where \Cref{obs:selfcomb} and \Cref{obs:twocomb} apply.
\begin{itemize}
	\item 1 with 1: $\diagsup{\bbAX}{\bbBY} \s \diaginf{\bbBY}{\bbAX}$ gives $\bbABXYp \s \bbempty$, which is configuration 2.
	\item 1 with 3: $\diagsup{\bbAX}{\bbY} \s \diaginf{\bbBY}{\bbAXYp}$ gives $\bbAXYp \s \bbY$, which is configuration 3.
	\item 1 with 3: $\diagsup{\bbBY}{\bbAXYp} \s \diaginf{\bbAX}{\bbY}$ gives $\bbABXYp \s \bbempty$, which is configuration 2.
	\item 1 with 5: $\diagsup{\bbAX}{\bbXYp} \s \diaginf{\bbBY}{\bbXY}$ gives $\bbAXYp \s \bbY$, which is configuration 3.
	\item 1 with 5: $\diagsup{\bbAX}{\bbXY} \s \diaginf{\bbBY}{\bbXYp}$ gives $\bbAXYp \s \bbY$, which is configuration 3.
	\item 3 with 4: $\diagsup{\bbAXYp}{\bbBXYp} \s \diaginf{\bbY}{\bbX}$ gives $\bbABXYp \s \bbempty$, which is configuration 2.
	\item 3 with 4: $\diagsup{\bbY}{\bbX} \s \diaginf{\bbAXYp}{\bbBXYp} $ gives $\bbXY \s \bbXYp$, which is configuration 5.
\end{itemize}
This implies that no new configurations are added to $\edgeconst$. Finally, observe that $\gen{\edgeconst}$ gives exactly the configurations in $\edgeconst_{\Pi_{\Delta}}$.

\section{Defective $3$-coloring}\label{sec:3col}
In this section, we show that $\lfloor (\Delta - 3)/2 \rfloor$-defective $3$-coloring requires $\Omega(\log_\Delta n)$ for deterministic algorithms and $\Omega(\log_\Delta \log n)$ for randomized ones. We actually show a stronger result: the lower bound that we prove holds also in the case in which one color is allowed to be \emph{arbdefective}. More in detail, we consider the following problem, which for simplicity we just call \emph{defective $3$-coloring}. The task is to assign to each node a color in $\{\A,\B,\C\}$ and to each edge between nodes of color $\C$ an orientation, such that the following is satisfied.
\begin{itemize}
	\item Each node of color $\A$ (resp.\ $\B$) has at most $d = \lfloor (\Delta - 3)/2 \rfloor$ neighbors of color $\A$ (resp.\ $\B$).
	\item The orientation satisfies that each node of color $\C$ has at most $d$ outgoing edges.
\end{itemize}
In the rest of \Cref{sec:3col}, we prove the following statement.
\begin{theorem}\label{thm:3col}
	The $\lfloor (\Delta - 3)/2 \rfloor$-defective $3$-coloring problem requires $\Omega(\log_\Delta n)$ rounds for deterministic algorithms and $\Omega(\log_\Delta \log n)$ rounds for randomized ones.
\end{theorem}
In order to prove this statement, we apply \Cref{cor:lifting}, which implies that it suffices to prove that there exists a non-trivial fixed point relaxation for the defective $3$-coloring problem.
In order to prove that there exists a non-trivial fixed point relaxation, we apply \Cref{obs:closedunderfp}, i.e., we provide a problem $\Pi_\Delta$ that can be solved in $0$ rounds given a defective $3$-coloring, we show that applying $\fpp$ to $\Pi_\Delta$ gives $\Pi_\Delta$ itself, and we show that this problem is not solvable in $0$ rounds.

We start in \Cref{ssec:3coldef} by defining the problem $\Pi_\Delta$. Then, in \Cref{ssec:3colsolvability} we show that $\Pi_\Delta$ can be solved in $0$ rounds given a solution for defective $3$-coloring, and that $\Pi_\Delta$ cannot be solved in $0$ rounds. In \Cref{ssec:applyfpp3col}, we show that applying $\fpp$ to $\Pi_\Delta$ gives $\Pi_\Delta$ itself.

We remark that, already in the case of defective $2$-coloring (which has a much simpler description than defective $3$-coloring), in order to prove that applying $\fpprocname$ on the node constraint gives the node constraint itself (which is part of applying $\fpp$), a long case analysis was needed. For the node constraint of $\Pi_\Delta$ that we are going to consider in this section, the number of cases is actually much larger (hundreds of cases). For this reason, we do not provide a handcrafted case analysis for the node constraint of $\Pi_\Delta$. Instead, we reduce the task of checking whether a given node constraint is the result of applying $\fpprocname$, to proving that all systems of inequalities belonging to a certain finite set have no solution, which can be checked automatically via computer tools.

\subsection{The Fixed Point}\label{ssec:3coldef}
We now define the problem $\Pi_\Delta$ that we will later show to be a non-trivial fixed point relaxation of defective $3$-coloring.
\paragraph{The set of labels.}
In the following, by $\Pi^2$ we denote the fixed point relaxation of defective $2$-coloring that we provided in \Cref{sec:2col}.
Recall that the labels of $\Pi^2$ defined in \Cref{sec:2col} are  $\Sigma = \{\bbempty, \bbX, \bbY, \bbXY, \bbXYp, \bbAX,\allowbreak \bbBY, \bbAXYp,\bbBXYp,\bbABXYp\}$.
In the following, we will consider each label $\mybox{\ell_1\ldots\ell_k}$ also as the set $\set{\ell_1,\ldots,\ell_k}$.
The labels $\Sigma_{\Pi_\Delta}$  of $\Pi_\Delta$ are defined as follows. For each label $\L$ in $\Sigma$, we add to $\Sigma_{\Pi_\Delta}$ the labels $\L$ and $\L \cup \{\C\}$.

\paragraph{The node constraint.}
We start by defining the node constraint $\nodeconst_{\Pi_\Delta}$ of $\Pi_\Delta$.
	Let $d = \lfloor (\Delta - 3)/2 \rfloor$.
	For $\Delta \le 4$ we get that  $d = 0$, and hence a problem that is at least as hard as $\Delta$-coloring, which, by \cite{Brandt2016}, implies the lower bounds stated in \Cref{thm:3col}. Hence, w.l.o.g., we restrict  to the case $\Delta \ge 5$, which in particular implies $d \ge 1$. Note that $2d + 3 \le \Delta \le 2d + 4$. 
	\begin{enumerate}
		\item\label{item:type1}
		\begin{align*}
			\bbempty^{d +1} \s \bbCX^{d} \s \bbACX^{\Delta - 2d - 1} \\
			\bbX^{2 d + 1} \s \bbACX^{\Delta - 2d - 1} \\
			\bbX^{d} \s \bbAX^{\Delta - d}
		\end{align*}
		
		\item\label{item:type2}
		\begin{align*}
			\bbempty^{d +1} \s \bbCY^{d} \s \bbBCY^{\Delta - 2d - 1} \\
			\bbY^{2d + 1} \s \bbBCY^{\Delta - 2d - 1} \\
			\bbY^{d} \s \bbBY^{\Delta - d}
		\end{align*}
		
		\item\label{item:type3}
		\begin{align*}
			\bbempty^{d +1} \s \bbCX^{d + 1} \s \bbACXYp^{d} \s \bbABCXYp^{\Delta - 3d - 2}  ~~~&(\text{if } \Delta > 3d + 2)\\
			\bbempty^{d +1} \s \bbCX^{d + 1} \s \bbACXYp^{\Delta - 2d - 2} ~~~&(\text{if } \Delta \le 3d + 2)\\
			\bbX^{2d + 2} \s \bbACXYp^{d} \s \bbABCXYp^{\Delta - 3d - 2}  ~~~&(\text{if } \Delta > 3d + 2)\\
			\bbX^{2d + 2} \s \bbACXYp^{\Delta - 2d - 2} ~~~&(\text{if } \Delta \le 3d + 2)\\
			\bbX^{d + 1} \s \bbAXYp^{2d+1} \s \bbABCXYp^{\Delta - 3d - 2} ~~~&(\text{if } \Delta > 3d + 2)\\
			\bbX^{d + 1} \s \bbAXYp^{\Delta - d - 1} ~~~&(\text{if } \Delta \le 3d + 2)\\
			\bbX^{d + 1} \s \bbAXYp^{d} \s \bbABXYp^{\Delta - 2d - 1} ~~~&
		\end{align*}
		
		\item\label{item:type4}
		\begin{align*}
			\bbempty^{d +1} \s \bbCY^{d + 1} \s \bbBCXYp^{d} \s \bbABCXYp^{\Delta - 3d - 2}  ~~~&(\text{if } \Delta > 3d + 2)\\
			\bbempty^{d +1} \s \bbCY^{d + 1} \s \bbBCXYp^{\Delta - 2d - 2} ~~~&(\text{if } \Delta \le 3d + 2)\\
			\bbY^{2d + 2} \s \bbBCXYp^{d} \s \bbABCXYp^{\Delta - 3d - 2}  ~~~&(\text{if } \Delta > 3d + 2)\\
			\bbY^{2d + 2} \s \bbBCXYp^{\Delta - 2d - 2} ~~~&(\text{if } \Delta \le 3d + 2)\\
			\bbY^{d + 1} \s \bbBXYp^{2d+1} \s \bbABCXYp^{\Delta - 3d - 2} ~~~&(\text{if } \Delta > 3d + 2)\\
			\bbY^{d + 1} \s \bbBXYp^{\Delta - d - 1} ~~~&(\text{if } \Delta \le 3d + 2)\\
			\bbY^{d + 1} \s \bbBXYp^{d} \s \bbABXYp^{\Delta - 2d - 1} ~~~&
		\end{align*}

		\item\label{item:type5} For all integers $j$ such that no exponent is negative.
		\begin{align*}
			\bbempty^{d + 2} \s \bbCXY^{j} \s \bbACXYp^{d - j} \s \bbBCXYp^{d - j} \s \bbABCXYp^{\Delta - 3d - 2 + j}\\
			\bbempty \s \bbXY^{d+1+j} \s \bbACXYp^{d -j } \s \bbBCXYp^{d-j} \s \bbABCXYp^{\Delta - 3d - 2 +j} \\
			\bbempty \s \bbXY^{j} \s \bbAXYp^{2d+1 -j} \s \bbBCXYp^{d - j} \s \bbABCXYp^{\Delta - 3d - 2 + j}\\
			\bbempty \s \bbXY^{j} \s \bbACXYp^{d-j} \s \bbBXYp^{2d+1-j} \s \bbABCXYp^{\Delta - 3d - 2 + j}\\
			\bbempty \s \bbXY^{j} \s \bbAXYp^{d-j} \s \bbBXYp^{d-j} \s \bbABXYp^{\Delta - 2d - 1 + j}\\
		\end{align*}
		
		\item\label{item:type6}
		\begin{align*}
			\bbempty^{d + 1} \s \bbCXY^{d+1} \s \bbCXYp \s \bbABCXYp  ~~~&(\text{if } \Delta = 2d + 4 )\\
			\bbempty^{d + 1} \s \bbCXY^{3d+4 - \Delta} \s \bbCXYp^{2 \Delta - 4d - 5} ~~~&(\text{if } \Delta \le 3d + 2 )\\
			\bbXY^{2d+2} \s \bbCXYp \s \bbABCXYp ~~~&(\text{if } \Delta = 2d+4)\\
			\bbXY^{4d + 5 - \Delta} \s \bbCXYp^{2\Delta - 4d - 5}~~~&(\text{if } \Delta \le 3d + 2 )\\
			\bbXY^{d+1} \s \bbXYp^{d + 2} \s \bbABCXYp~~~&(\text{if } \Delta = 2d+4)\\
			\bbXY^{3d + 4 - \Delta} \s \bbXYp^{2\Delta - 3d - 4}~~~&(\text{if } \Delta \le 3d + 2 )\\
			\bbXY^{j} \s \bbXYp^{2d + 3 - 2 j} \s \bbABXYp^{\Delta + j -2d -3} ~~~&(\text{for } 2 \le j \le d+1)\\
		\end{align*}
		
		\item\label{item:type7} 
		\begin{align*}
			\bbempty^{d} \s \bbC^{\Delta - d}
		\end{align*}
	\end{enumerate}

\paragraph{The edge constraint.}
The edge constraint $\edgeconst_{\Pi_\Delta}$ of $\Pi_\Delta$ is defined as follows. We consider the edge constraint $\edgeconst_{\Pi^2}$ of the fixed point relaxation of the defective $2$-coloring problem defined in \Cref{sec:2col}, and for each configuration $\L_1 \s \L_2 \in \edgeconst_{\Pi^2}$ we add to $\edgeconst_{\Pi_\Delta}$ the configurations $\L_1 \s \L_2$, $\L_1 \s (\L_2 \cup \set{\C})$, and  $(\L_1 \cup \set{\C}) \s \L_2$.

\subsection{Solvability of $\Pi_{\Delta}$}\label{ssec:3colsolvability}
In this section, we first show that $\Pi_\Delta$ can be solved in $0$ rounds given a defective $3$-coloring as input, and then we show that $\Pi_\Delta$ is not $0$-round-solvable.

\paragraph{Solving $\Pi_\Delta$ with defective $3$-coloring.}
Assume we are given a solution for the defective $3$-coloring problem, where each node already knows which neighbors have the same color as them (which can be inferred in just $1$ round of communication). We show that this solution can be converted in $0$ rounds to a solution for $\Pi_\Delta$. Each node of color $\A$ outputs the configuration $\bbAX^{\Delta-d} \s \bbX^{d}$, by putting the label $\bbAX$ on $\Delta-d$ arbitrary edges connecting it to neighbors of different color (which are guaranteed to exist) and $\bbX$ on all other incident edges.
Nodes of color $\B$ act similarly by using the configuration  $\bbBY^{\Delta-d} \s \bbY^{d}$. Now consider a node of color $\C$: it is guaranteed to have at most $d$ outgoing edges. On such edges the node outputs $\bbempty$, and then it outputs additional $\bbempty$ on arbitrary edges in order to obtain exactly $d$ edges labeled $\bbempty$. On all other incident edges, it outputs $\bbC$. Observe that the configuration used by nodes of color $\C$ is $\bbC^{\Delta-d} \s \bbempty^{d}$. It is easy to see that this labeling satisfies the edge constraint $\edgeconst_{\Pi_\Delta}$.

\paragraph{Non-$0$-round-solvability.}
Note that all configurations allowed by the node constraint contain at least one label $\L$ satisfying $\L \cap \{\A,\B,\C,+\} \neq \emptyset$. From the edge constraint we can observe that all such labels satisfy that $\L \s \L$ is not in $\edgeconst_{\Pi_\Delta}$. Hence, the problem is not solvable in $0$ rounds.

\subsection{Applying $\fpp$}\label{ssec:applyfpp3col}
In order to prove that by applying $\fpp$ on $\Pi_\Delta$ we obtain $\Pi_\Delta$ itself, we first define, in \Cref{sssec:diagram}, the diagram $D$ that we will use as input for procedure $\fpp$. Then, we apply procedure $\fpp$: in \Cref{sssec:3coledge}, we apply $\fpprocname$ on the edge constraint of $\Pi_\Delta$ (and we additionally apply the $\gen{\cdot}$ operator), and in \Cref{sssec:3colnode} we apply $\fpprocname$ on the node constraint.

\subsubsection{The Diagram}\label{sssec:diagram}
We define the diagram $D = (\Sigma_D,E_D)$ as $\Sigma_D = \Sigma_{\Pi_{\Delta}}$ and $E_D = \{ (\ell_1,\ell_2) ~|~ \ell_2 \subsetneq \ell_1 \}$ (see \Cref{fig:diag3col}). The diagram $D'$ is obtained by flipping the edges of $D$.
\begin{figure}
	\centering
	\includegraphics[width=0.7\textwidth]{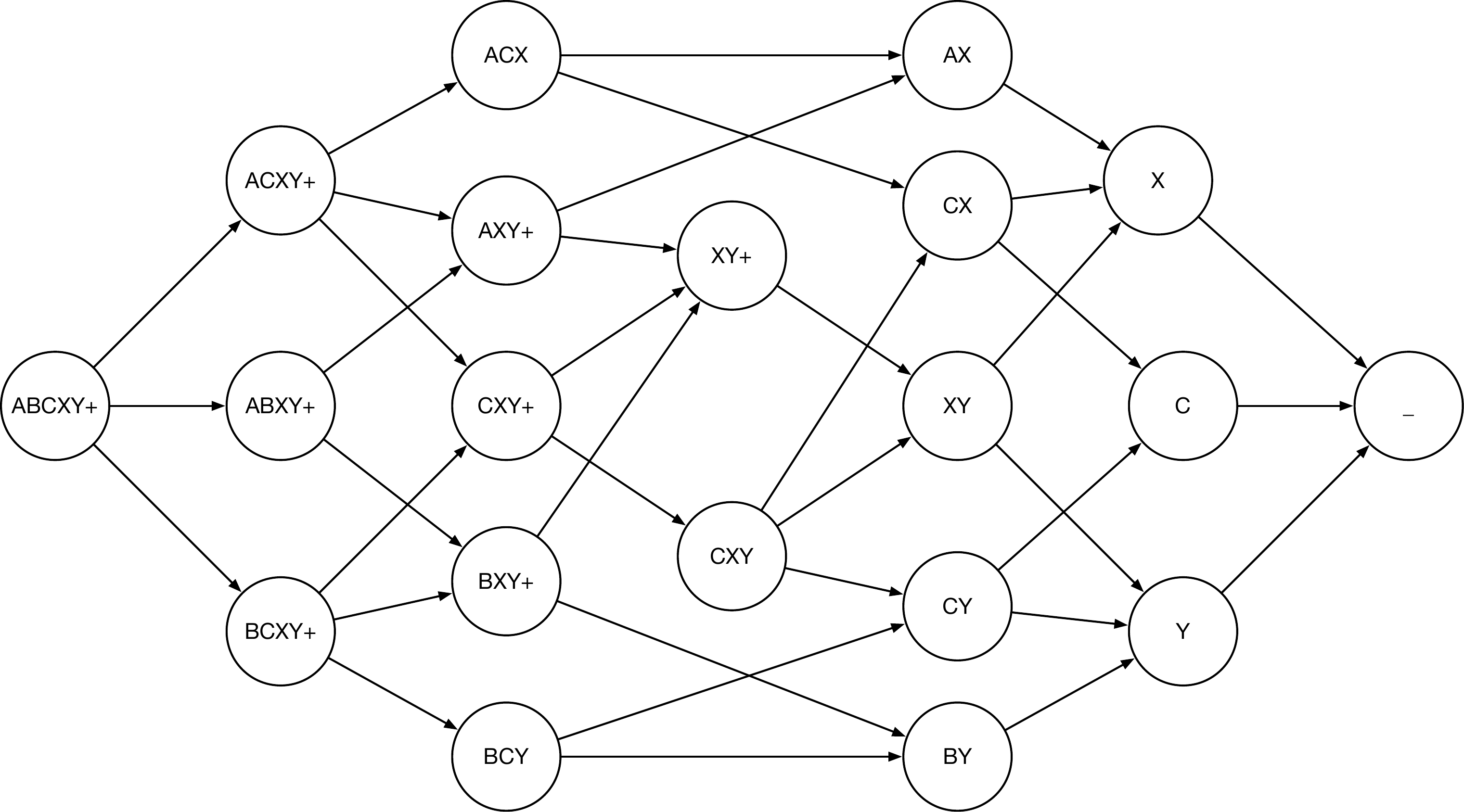}
	\caption{The diagram $D$.}
	\label{fig:diag3col}
\end{figure}

\subsubsection{Applying Subprocedure $\fpprocname$ on the Edge Constraint}\label{sssec:3coledge}
We now apply subprocedure $\fpprocname$ on the edge constraint, by using diagram $D'$. The first step in computing $\edgeconst$ is taking all maximal configurations of $\edgeconst_{\Pi_{\Delta}}$. Observe that, by the definition of $\edgeconst_{\Pi_{\Delta}}$, the maximal configurations can be obtained by starting from the maximal configurations $\fL$ in the edge constraint $\edgeconst_{\Pi^2}$ of the fixed point relaxation of defective $2$-coloring and adding, in all possible ways, the label $\C$ in exactly one set of $\fL$  (that is, for each maximal configuration in $\edgeconst_{\Pi^2}$, we obtain two maximal configurations for $\edgeconst_{\Pi_{\Delta}}$). Hence, we obtain the following maximal configurations.
\begin{enumerate}
	\item $\bbACX \s \bbBY$
	\item $\bbAX \s \bbBCY$
	
	\item $\bbABCXYp \s \bbempty$
	\item $\bbABXYp \s \bbC$
	
	\item $\bbACXYp \s \bbY$
	\item $\bbAXYp \s \bbCY$
	
	\item $\bbBCXYp \s \bbX$
	\item $\bbBXYp \s \bbCX$
	
	\item $\bbCXYp \s \bbXY$
	\item $\bbXYp \s \bbCXY$
\end{enumerate}
Consider an arbitrary pair of configurations $\fL_1 = \L_{1,1} \s \L_{1,2}$ and $\fL_2 = \L_{2,1} \s \L_{2,2}$. Consider now the configurations $\fL'_1 = (\L'_{1,1} \setminus \{\C\}) \s (\L_{1,2} \setminus \{\C\})$ and $\fL'_2 = (\L_{2,1} \setminus \{\C\}) \s (\L_{2,2} \setminus \{\C\})$. In \Cref{sec:2col}, we already showed that, by combining $\fL'_1$ with $\fL'_2$, we get a configuration $\fL = \L_1 \s \L_2$ in $\edgeconst_{\Pi^2}$. Now observe that, when combining $\fL_1$ with $\fL_2$, the result is either $(\L_1 \cup \{\C\}) \s \L_2$ or $\L_1 \s (\L_2 \cup \{\C\})$, and hence a configuration in $\edgeconst$.
This implies that no new configurations are added to $\edgeconst$. Finally, observe that $\gen{\edgeconst}$ gives exactly the configurations in $\edgeconst_{\Pi_{\Delta}}$.

\subsubsection{Applying Subprocedure $\fpprocname$ on the Node Constraint}\label{sssec:3colnode}
\allowdisplaybreaks
We devote this section to proving that, by applying  $\fpprocname$ on the node constraint $\nodeconst_{\Pi_\Delta}$, we obtain $\nodeconst_{\Pi_\Delta}$ itself. To this end, we need to prove that, for any pair of configurations in $\nodeconst_{\Pi_\Delta}$, and any possible combination of them, we obtain a configuration that is dominated w.r.t.\ $D$ by a configuration already present in $\nodeconst_{\Pi_\Delta}$.
Observe that $\nodeconst_{\Pi_\Delta}$ is described in $\Cref{ssec:3coldef}$ in a compact form: each row of the description (which we call a \emph{line}) represents (potentially) multiple configurations.  For example, the last line of case 5 is $\bbempty \s \bbXY^{j} \s \bbAXYp^{d-j} \s \bbBXYp^{d-j} \s \bbABXYp^{\Delta - 2d - 1 + j}$, which has a free variable $j$, and for different values of $j$ we get different configurations.
Since the number of lines in $\nodeconst_{\Pi_\Delta}$ is a fixed constant, while the number of configurations grows with $\Delta$, it will be convenient to consider combinations of lines instead of combinations of single configurations.
The notion of combination of configurations, and the notion of domination of configurations extends in the natural way to lines: when we combine two lines we get all possible combinations for all possible choices of the parameters, and a configuration is dominated by a line if there exists a value of the parameter of the line that gives a configuration that dominates.

Let $\fL_1$ and $\fL_2$ be two arbitrary lines in $\nodeconst_{\Pi_{\Delta}}$, and let $\L_1 \in \fL_1$ and $\L_2 \in \fL_2$ be two arbitrary labels.
We need to prove that there exist some target lines $\fT_1, \ldots, \fT_h$ in $\nodeconst_{\Pi_{\Delta}}$ satisfying that all the configurations produced by combining $\fL_1$ with $\fL_2$, where $\diagsup{\cdot}{\cdot}$ is applied on $\L_1$ and $\L_2$, are dominated (w.r.t.\ $D$) by at least one line in $\fT_1, \ldots, \fT_h$.
Let $S$ be the set of configurations obtained from the combinations, i.e., the set of configurations that the target lines are supposed to dominate. 
Note that $S$ contains the combination of $\fL_1$ with $\fL_2$ (where $\diagsup{\cdot}{\cdot}$ is applied on $\L_1$ and $\L_2$) for all possible values of the free variables in $\fL_1$ and $\fL_2$, and all possible matchings (i.e., all possible choices of pairs on which to apply $\diaginf{\cdot}{\cdot}$).
We remark that it is crucial for $h$ to be as small as possible since the approach that we will take relies on a procedure that has a running time that is exponential in $h$. In particular, we cannot simply use the whole set of lines as the set of target lines.

\paragraph{Notation.}
We will use the expression $\sqcup_{1 \le j \le x} \L_j^{e_i}$ to denote $\L_1^{e_1} \s \ldots \s L_x^{e_x}$.
Assume we want to combine the following two lines.
\begin{align*}
	\fL_1 = \bigsqcup_{1 \le j \le s} \L_{1,j}^{e_{1,j}} \\
	\fL_2 = \bigsqcup_{1 \le j \le t} \L_{2,j}^{e_{2,j}} \\
\end{align*}
The variable $e_{i,j}$ is the exponent of $\fL_i$ in position $j$.
Some given lines may have free variables in the exponents (i.e., variables that are not $\Delta$ nor $d$).
Let $f_i$ be the free variable of $\fL_i$, if it exists. Let $F \subseteq \{f_1, f_2\}$ be the set of existing free variables.
W.l.o.g., assume that we apply $\diagsup{\cdot}{\cdot}$ on $\L_{1,1}$ and $\L_{2,1}$. 

Let $\fC$ be a configuration obtained by combining $\fL_1$ and $\fL_2$ (where $\diagsup{\cdot}{\cdot}$ is applied on $\L_1$ and $\L_2$) for a fixed choice of the free variables of $\fL_1$ and $\fL_2$ and a fixed matching. More specifically, let $x_{i,j}$ denote how many copies of $\L_{1,i}$ are matched with copies of $\L_{2,j}$, excluding the matched pair on which $\diagsup{\cdot}{\cdot}$ is applied. We obtain the following.
\[
\fC = \diagsup{\L_{1,1}}{\L_{2,1}} \s \bigsqcup_{\substack{1 \le i \le s, \\ 1 \le j \le t}} \diaginf{\L_{1,i}}{\L_{2,j}}^{x_{i,j}}
\]
Assume that the target lines, for $1 \le i \le h$, are the following.
\begin{align*}
	\fT_i = \bigsqcup_{1 \le j \le h_i} \T_{i,j}^{t_{i,j}}
\end{align*}
The variable $t_{i,j}$ is the exponent of $\fT_i$ in position $j$. Let $k_i$ be the free variable of $\fT_i$, if it exists.
Throughout the remainder of \Cref{sssec:3colnode}, we will follow the above notation; in particular we will assume that the lines that we combine are $\fL_1$ and $\fL_2$, and that $\diagsup{\cdot}{\cdot}$ is applied on  $\L_{1,1}$ and $\L_{2,1}$. 

\paragraph{Our approach.}
Our approach consists of two parts. The first part is determining a set of target lines, while the second part consists in proving that, together, those target lines dominate all configurations obtained from the combination of $\fL_1$ and $\fL_2$ (where $\diagsup{\cdot}{\cdot}$ is applied on  $\L_{1,1}$ and $\L_{2,1}$). 
To show, in the second part, that the combinations of two lines are dominated by the target lines determined in the first part, it will be convenient to \emph{identify}, as a function of the free variables of the two combined lines and the variables $x_{i,j}$, which target line dominates the obtained combination. Therefore, in the first part, we will not only identify a set of target lines, but also, for each target line $\fT_i$ that has a free variable, an expression $\mathrm{expr}_i$ that specifies the value of the free variable of the target line as a function of the free variables of the two combined lines and the variables $x_{i,j}$.
More in detail, our approach consists of the following two parts.
\begin{enumerate}
	\item We create a list $\Psi$ of pairs $((\fL_1,\fL_2,\L_1,\L_2),\{(\fT_1,\mathrm{expr}_1), \ldots, (\fT_h,\mathrm{expr}_h)\})$ that contains one pair for each possible choice of $(\fL_1,\fL_2,\L_1,\L_2)$ satisfying that $\fL_1$ and $\fL_2$ are lines in $\nodeconst_{\Pi_{\Delta}}$, $\L_1 \in \fL_1$, and $\L_2 \in \fL_2$.  
	\item For each pair $((\fL_1,\fL_2,\L_1,\L_2),\{(\fT_1,\mathrm{expr}_1), \ldots, (\fT_h,\mathrm{expr}_h)\})$ in  $\Psi$, we run some automatic method to verify that it is \emph{valid}, i.e., that all the combinations given by $\fL_1$ and $\fL_2$ when taking $\diagsup{\cdot}{\cdot}$ on $\L_1$ and $\L_2$ are dominated by at least one line in $\{\fT_1, \ldots, \fT_h\}$.
\end{enumerate}
For the overall proof, it is not relevant how $\Psi$ is constructed (as long as step 2 succeeds), but for the interested reader we mention that we built this list by using computer tools for most of the cases, while for some hard cases we had to find the right target lines and expressions manually.  A computer program that encodes a list $\Psi$ for $\nodeconst_{\Pi_{\Delta}}$ and that verifies its correctness by applying the automatic method mentioned in step 2 can be found at \cite{checker}.

What remains to be done is to describe the automatic method mentioned in step 2 and to prove its correctness. Our automatic method is based on  \Cref{lem:auto}, which states that we can reduce the problem of verifying that a pair $((\fL_1,\fL_2,\L_1,\L_2),\{(\fT_1,\mathrm{expr}_1), \ldots, (\fT_h,\mathrm{expr}_h)\})$ is valid to the problem of proving that a finite set of systems of inequalities are all unsolvable over the integers. 
After the proof of \Cref{lem:auto} we describe our automatic method, and after that we provide an example.
\begin{lemma}\label{lem:auto}
	The problem of checking whether $((\fL_1,\fL_2,\L_1,\L_2),\{(\fT_1,\mathrm{expr}_1), \ldots, (\fT_h,\mathrm{expr}_h)\})$ is valid can be reduced to checking whether a finite set of systems of inequalities are all unsolvable over the integers.
\end{lemma}
\begin{proof}
We start by collecting some useful inequalities.

\paragraph{Restrictions on $\Delta$ and $d$.}
From the definition of $\nodeconst_{\Pi_{\Delta}}$, we obtain the following set of inequalities, which we denote by $\mathcal{A}_1$.
\begin{align*}
	d &\ge 1 \\
	\Delta &\le 2d + 4  \\
	\Delta &\ge 2d + 3 \\
	\Delta &\ge 5
\end{align*}

\paragraph{Restrictions on the input lines.}
The lines in the description of $\nodeconst_{\Pi_\Delta}$ come with some restrictions. For example, the first line of case $5$ requires $j \ge 0, d-j \ge 0, \Delta - 3d -2 +j \ge 0$, and the first line of case $6$ requires $\Delta = 2d+4$.
Let $\mathcal{A}_2$ be the set of inequalities expressing these restrictions on $\fL_1$ and $\fL_2$.

\paragraph{Restrictions on the obtained line.}
In the following, we infer some constraints on the variables $x_{i,j}$, as a function of the exponents of the lines $\fL_1$ and $\fL_2$. 
Let $z$ be $1$ if $i = 1$ and $0$ otherwise.
\begin{align}
	x_{i,j} &\ge 0 ~&(\text{for all }  1 \le i \le s, 1 \le j \le t) \label{constr:nonneg}\\
	e_{1,i} &= z + \sum_{1 \le j \le t} x_{i,j} ~&(\text{for all }  1 \le i \le s) \label{constr:exp1}\\
	e_{2,i} &= z + \sum_{1 \le j \le s} x_{j,i} ~&(\text{for all }  1 \le i \le t) \label{constr:exp2}
\end{align}
We call this set of inequalities $\mathcal{A}_3$.
The inequalities in (\ref{constr:nonneg}) represent the fact that the exponents of the obtained combination cannot be negative. The inequalities in (\ref{constr:exp1}) represent the fact that, for each label of $\fL_1$, we have a bound on how many copies are available to construct the combination, and this bound is $e_{1,i}$. The variable $z$ handles the special case of $i=1$, where the number of copies of $\L_{1,1}$ that we can use is $e_{1,1} - 1$, since $\L_{1,1}$ has been used once when computing $\diagsup{\L_{1,1}}{\L_{2,1}}$. The inequalities in (\ref{constr:exp2}) are analogous and concern the exponents of $\fL_2$.

\paragraph{Free variables of the target lines.}
Recall that $k_i$ is the free variable of $\fT_i$, if it exists. As discussed before, for each free variable $k_i$, we are given an expression $\mathrm{expr}_i$, as a function of $\Delta$, $d$, the variables $x_{i,j}$, and the free variables of $\fL_1$ and $\fL_2$. We call $\mathcal{A}_4$ the set of equations $k_i = \mathrm{expr}_i$.

\paragraph{Exponents of the target lines.}
Analogously to the restrictions on the input lines, some target lines may also have restrictions (e.g., the exponents of lines of case 5 are required to be non-negative, and lines of case 6 have some restriction on $\Delta$).
Let $\mathcal{P}_{i,1}$ be the set of inequalities expressing these constraints for the target line $\fT_i$.

\paragraph{A perfect matching between the obtained line and a target line.}
We now provide some necessary and sufficient conditions for a line $\fT = \T_1^{t_1} \s \ldots \s \T_h^{t_h}$ to dominate the obtained line $\fC$. 
Recall that a configuration $\mathcal{Y}$ is dominated by a configuration $\mathcal{Z}$ if there exists a perfect matching between the $\Delta$ labels of $\mathcal{Y}$ and the $\Delta$ labels of $\mathcal{Z}$ satisfying that, if label $\ell$ of $\mathcal{Y}$ is matched with label $\ell'$ of $\mathcal{Z}$, then $\ell'$ is reachable from $\ell$ in $D$. Hence, in the following, with \emph{$\ell$ can be matched with $\ell'$} we denote the fact that $\ell'$ is reachable from $\ell$ in $D$.
For a line $\fL = \L_1^{e_1} \s \ldots \L_x^{e_x}$, let $S_\fL = \{ \L_{j} ~|~ 1 \le j \le x\}$, that is, $S_\fL$ is the set of different labels appearing in $\fL$.
Call a set $R$ of labels  $\fL$-right-closed if, whenever a label $\ell$ is in $R$, the nodes that are in $\fL$ and are successors of $\ell$ in  diagram $D$ are also in $R$.  We define the set of \emph{right-closed cuts} $\mathcal{R}_{\fL}$ of a line $\fL$  by $\mathcal{R}_{\fL} := \{ R \subseteq S_{\fL} ~|~ R \text{ is $\fL$-right-closed }\}$. For example, the set of right-closed cuts of the line 
\[\bbempty^{d + 2} \s \bbCXY^{f} \s \bbACXYp^{d - f} \s \bbBCXYp^{d - f} \s \bbABCXYp^{\Delta - 3d - 2 + f}\]
is 
\begin{align*}
	\{&\{\}, \allowbreak \{\bbempty\},\allowbreak \{\bbempty,\bbCXY\},\allowbreak \{\bbempty,\bbCXY,\bbACXYp\},\allowbreak \{\bbempty,\bbCXY,\bbBCXYp\},\\&\{\bbempty,\bbCXY,\bbACXYp,\bbBCXYp\}, \{\bbempty,\bbCXY,\bbACXYp,\bbBCXYp,\bbABCXYp\}\}.
\end{align*}
For a $\fT$-right-closed set $R \subseteq S_\fT$ of labels of $\fT$, let
\[M(R) := \{ \ell \in S_{\fC} \mid \nexists \ell' \in S_{\fT} \setminus R \text{ s.t. $\ell'$  is reachable from $\ell$ in $D$}  \},\]
that is, $M(R) \subseteq S_\fC$ is the set of labels of $\fC$ that can only be matched with elements of $R$. In the following, for a predicate $P$, we set the expression $[P]$ equal to $1$ if the predicate $P$ is true, and to $0$ otherwise. Let $X(R)$ be the sum of the exponents of the elements of $M(R)$ in $\fC$, that is, 
\[
X(R) := [\diagsup{\L_{1,1}}{\L_{2,1}} \in M(R)] + \sum_{\substack{1 \le i \le s, \\ 1 \le j \le t}} [\diaginf{\L_{1,i}}{\L_{2,j}} \in M(R)] ~ x_{i,j}.
\]
Similarly, we define $T(R)$ as the sum of the exponents of the elements of $R$ in $\fT$, that is,
\[
T(R) := \sum_{1 \le j \le h} [\T_j \in R] ~ t_j.
\]
We claim that a perfect matching between $\fC$ and $\fT$ exists if and only if, for all right-closed cuts $R$ of~$\fT$,
\[
X(R) \le T(R)
\]
holds.

We prove this claim as follows. We construct a bipartite graph $G = (V \cup U, E)$, where $V = \{v_1, \ldots, v_\Delta\}$ and $U = \{u_1, \ldots, u_\Delta\}$. We define $\ell(v_i)$ as the $i$th label in the multiset $\fC$ taken in some arbitrary order, and $\ell(u_i)$ as the $i$th label in the multiset $\fT$ taken in some arbitrary order. We put an edge between $v_i$ and $u_j$ if and only if, in the diagram $D$, $\ell(u_j)$ can be reached from $\ell(v_i)$. Assume that the condition of our claim holds, that is, for all right-closed cuts $R$ of $\fT$ it holds that $X(R) \le T(R)$. 
In the following, we use Hall's marriage theorem \cite{hall} to prove that this assumption implies that a perfect matching exists in $G$. Hall's marriage theorem implies that we only need to prove that, for every subset $W$ of $V$, the number of neighbors $N_G(W)$ in $G$ of the nodes in $W$  is at least $|W|$.

Let $\mathrm{rc}(W) := \{  \ell(u_j)  ~|~ u_j \in N_G(W) \}$, that is, $\mathrm{rc}(W)$ contains all labels of $\fT$ of nodes that are neighbors of nodes in $W$.
Observe that $\mathrm{rc}(W)$ is $\fT$-right-closed, since, whenever a node $v \in V$ can be matched with a node $u \in U$, it can also be matched with all nodes $u' \in U$ satisfying that $\ell(u')$ is reachable by $\ell(u)$ in $D$. Also, observe that, if a node in $N_G(W)$ has label $\ell$, then all other nodes in $U$ that have label $\ell$ are also in $N_G(W)$. This implies that $T(R) = |N_G(W)|$.
 Let $M$ be the subset of nodes in $V$ that have a label in $M(R)$, that is, the nodes in $V$ that can only be matched with elements in $N_G(W)$. Clearly, the nodes in $W$ are all contained in $M$. Also, by assumption, $X(R) \le T(R)$. Hence, we obtain that
\[
 |W| \le |M| = X(R) \le T(R) = |N_G(W)|,
\]
implying that Hall's condition is satisfied.

We define $\mathcal{P}_{i,2}$ as the set consisting of the inequalities $X(R) \le T(R)$ for all right-closed cuts $R$ of $\fT_i$.
We define $\mathcal{P}_{i}$ as $\mathcal{P}_{i,1} \cup \mathcal{P}_{i,2}$. We obtain that $\fT_i$ dominates $\fC$ if and only if all inequalities of $\mathcal{P}_{i}$ are satisfied.

\paragraph{Automatic checking.}
We define $\mathcal{A}$ as $\mathcal{A}_{1} \cup \mathcal{A}_{2} \cup \mathcal{A}_{3} \cup \mathcal{A}_{4}$.
By abusing notation, by $\mathcal{A}$ we also denote the conjunction of all the inequalities contained in $\mathcal{A}$. We do the same for $\mathcal{P}_i$, that is, $\mathcal{P}_i = p_{i,1} \land \ldots \land p_{i,|\mathcal{P}_i|}$, where the $p_{i,j}$ denote the inequalities in $\mathcal{P}_i$.
We consider the case $F = \{f_1,f_2\}$, the other cases are analogous.

A pair $((\fL_1,\fL_2,\L_1,\L_2),\{(\fT_1,\mathrm{expr}_1), \ldots, (\fT_h,\mathrm{expr}_h)\})$ is valid if and only if the following holds: if all the inequalities in $\mathcal{A}$ are satisfied, then all the configurations obtained by line $\fC$ are dominated by at least one target line, or in other words, for all choices of $\Delta, f_1, f_2, x_{1,1},\ldots,x_{s,t}$, there exists an $i$ for which $\mathcal{P}_i$ is satisfied. As a formula, the statement that the pair is valid is equivalent to the following statement.
\[
\forall ~ \Delta, f_1, f_2,x_{1,1},\ldots,x_{s,t} \in \mathbb{N} ~ (\mathcal{A} \implies \bigvee_{1 \le i \le h} \mathcal{P}_i)
\]
This, in turn, is equivalent to
\begin{align*}
	& \forall \Delta, f_1, f_2,x_{1,1},\ldots,x_{s,t} \in \mathbb{N} ~ (\lnot \mathcal{A} \lor \bigvee_{1 \le i \le h} \mathcal{P}_i) & \iff\\
	& \lnot\lnot (\forall  \Delta, f_1, f_2,x_{1,1},\ldots,x_{s,t} \in \mathbb{N} ~ (\lnot \mathcal{A} \lor \bigvee_{1 \le i \le h} \mathcal{P}_i)) & \iff\\
	& \lnot \exists \Delta, f_1, f_2,x_{1,1},\ldots,x_{s,t} \in \mathbb{N} ~ \lnot(\lnot \mathcal{A} \lor \bigvee_{1 \le i \le h} \mathcal{P}_i) & \iff\\
	& \lnot \exists \Delta, f_1, f_2,x_{1,1},\ldots,x_{s,t} \in \mathbb{N} ~ (\mathcal{A} \land \lnot \bigvee_{1 \le i \le h} \mathcal{P}_i) & \iff\\
	& \lnot \exists \Delta, f_1, f_2,x_{1,1},\ldots,x_{s,t} \in \mathbb{N} ~ (\mathcal{A} \land \bigwedge_{1 \le i \le h} \lnot\mathcal{P}_i) & \iff\\
	& \lnot \exists \Delta, f_1, f_2,x_{1,1},\ldots,x_{s,t} \in \mathbb{N} ~ (\mathcal{A} \land \bigwedge_{1 \le i \le h} \lnot \bigwedge_{1 \le j \le |\mathcal{P}_i|} p_{i,j}) & \iff\\
	& \lnot \exists \Delta, f_1, f_2,x_{1,1},\ldots,x_{s,t} \in \mathbb{N} ~ (\mathcal{A} \land \bigwedge_{1 \le i \le h} \bigvee_{1 \le j \le |\mathcal{P}_i|} \lnot p_{i,j}) & \iff\\
	& \bigwedge_{ P \in  \mathcal{P}_1 \times \ldots \times  \mathcal{P}_h } \lnot \exists \Delta, f_1, f_2,x_{1,1},\ldots,x_{s,t} \in \mathbb{N} ~ ( \mathcal{A} \land \bigwedge_{p \in P} \lnot p )&
\end{align*}
The last statement states that checking whether $\fT_1,\ldots,\fT_h$ dominates $\fC$ can be reduced to checking whether a finite number of systems of inequalities have no solution over the integers. This concludes our proof.
\end{proof}
We now describe our automatic procedure. 
 At first, for each given pair $((\fL_1,\fL_2,\L_1,\L_2),\allowbreak \{(\fT_1,\mathrm{expr}_1), \ldots, \allowbreak (\fT_h,\mathrm{expr}_h)\})$ we generate a set $S$  of systems of inequalities as described in \Cref{lem:auto}. 
In general, a problem that can be solved efficiently is checking whether a given system of inequalities has no solution over the reals. Unfortunately, it is false that the systems of inequalities in $S$ have no solutions over the reals for the list $\Psi$ that we computed. For this reason, we use the fact that our variables are integers as follows.
Consider some inequality $p$ of the form $\mathrm{expression} \le \emph{value}$. The inequality $\lnot p$ would be $\mathrm{expression} > \emph{value}$, but since all the variables are integers, we instead write $\lnot p$ as $\mathrm{expression} \ge \emph{value} + 1$.
The second step of our automatic procedure consists in computing the set $S'$ of systems of inequalities obtained according to the described transformation of inequalities. The third step consists in verifying that all systems of inequalities in $S'$ have no solution over the reals, which in particular implies that they have no solutions over the integers.
By using computer tools, we applied this automatic procedure on the list $\Psi$ that we computed, and all systems of inequalities turned out to have no solution \cite{checker}.

\paragraph{An example.}
We provide an example by considering the following two lines (that are the first of case 1 and the fifth of case 5):
\begin{itemize}
	\item $\fL_1 = \bbempty^{d+1} \s \bbCX^{d} \s \bbACX^{\Delta - 2d - 1}$
	\item $\fL_2 = \bbempty \s \bbXY^{f} \s \bbAXYp^{d-f} \s \bbBXYp^{d - f} \s \bbABXYp^{\Delta-2d-1+f}$
\end{itemize}
The example that we are going to provide is for the case in which $\diagsup{\cdot}{\cdot}$ is taken on $\bbACX$ and $\bbXY$. 
Observe that $\fL_2$ is actually not just a configuration, but it is a \emph{set} of configurations (i.e., a line), since it depends on the parameter $f$.

\paragraph{The target lines.}
We are going to show how the automatic procedure proves that all configurations obtained by combining $\fL_1$ and $\fL_2$ are dominated, for some $k$ that may depend on how the combination is performed, by at least one of the following lines:
\begin{itemize}
	\item $\fT_1 = \bbempty^{d + 2} \s \bbCXY^{k} \s \bbACXYp^{d - k} \s \bbBCXYp^{d - k} \s \bbABCXYp^{\Delta - 3d - 2 + k}$
	\item $\fT_2 = \bbX^{d + 1} \s \bbAXYp^{d} \s \bbABXYp^{\Delta - 2d - 1}$
\end{itemize}
We are going to use $k = d-x_{2,1}-x_{2,2}$.

\paragraph{Restrictions on $\Delta$ and $d$.}From the definition of $\nodeconst_{\Pi_{\Delta}}$, we obtain the following set of inequalities, which we denote by $\mathcal{A}_1$.
\begin{align*}
	d &\ge 1 \\
	\Delta &\le 2d + 4  \\
	\Delta &\ge 2d + 3 \\
	\Delta &\ge 5
\end{align*}

\paragraph{Restrictions on the input lines.}
Configuration $\fL_2$ has some restrictions on its exponents. Let $\mathcal{A}_2$ be the set of such inequalities, that we report here.
\begin{align*}
	f &\ge 0 \\
	d - f & \ge 0\\
	\Delta - 2d - 1 + f &\ge 0
\end{align*}

\paragraph{The obtained line.}
By combining $\fL_1$ and $\fL2$ we obtain the following line:
\begin{align*}
	\fC &&=&&\diagsup{\bbACX}{\bbXY} \s & \diaginf{\bbempty}{\bbempty}^{x_{1,1}} \s \diaginf{\bbempty}{\bbXY}^{x_{1,2}} \s \diaginf{\bbempty}{\bbAXYp}^{x_{1,3}} \s \diaginf{\bbempty}{\bbBXYp}^{x_{1,4}} \\ 
	& && && \diaginf{\bbempty}{\bbABXYp}^{x_{1,5}} 
	 \s \diaginf{\bbCX}{\bbempty}^{x_{2,1}} \s \diaginf{\bbCX}{\bbXY}^{x_{2,2}} \s \diaginf{\bbCX}{\bbAXYp}^{x_{2,3}} \\
	 & && && \diaginf{\bbCX}{\bbBXYp}^{x_{2,4}} \s \diaginf{\bbCX}{\bbABXYp}^{x_{2,5}} 
	 \s \diaginf{\bbACX}{\bbempty}^{x_{3,1}} \s \diaginf{\bbACX}{\bbXY}^{x_{3,2}} \\
	 & && && \diaginf{\bbACX}{\bbAXYp}^{x_{3,3}} \s \diaginf{\bbACX}{\bbBXYp}^{x_{3,4}} \s \diaginf{\bbACX}{\bbABXYp}^{x_{3,5}} \\
	 &&=&&\bbX \s & \bbempty^{x_{1,1}} \s \bbXY^{x_{1,2}} \s \bbAXYp^{x_{1,3}} \s \bbBXYp^{x_{1,4}} \s \bbABXYp^{x_{1,5}} \\
	 & && &&  \bbCX^{x_{2,1}} \s \bbCXY^{x_{2,2}} \s \bbACXYp^{x_{2,3}} \s \bbBCXYp^{x_{2,4}} \s \bbABCXYp^{x_{2,5}} \\
	 & && &&  \bbACX^{x_{3,1}} \s \bbACXYp^{x_{3,2}} \s \bbACXYp^{x_{3,3}} \s \bbABCXYp^{x_{3,4}} \s \bbABCXYp^{x_{3,5}} 
\end{align*}
Note that the variable $x_{i,j}$ represents how many $\diaginf{\cdot}{\cdot}$ have been taken with the $i$th set in line $\fL_1$ and the $j$th set in line $\fL_2$.
The exponents of the obtained line satisfy the following constraints, and we call the set containing them $\mathcal{A}_3$.
\begin{align*}
 x_{i,j} &\ge 0 ~(\forall ~ i,j)\\
 x_{1,1} + x_{1,2} + x_{1,3} + x_{1,4} + x_{1,5} &= d + 1\\
 x_{2,1} + x_{2,2} + x_{2,3} + x_{2,4} + x_{2,5} &= d\\
 x_{3,1} + x_{3,2} + x_{3,3} + x_{3,4} + x_{3,5} &= \Delta -2d -2\\
 x_{1,1} + x_{2,1} + x_{3,1} &= 1\\
 x_{1,2} + x_{2,2} + x_{3,2} &= f - 1\\
 x_{1,3} + x_{2,3} + x_{3,3} &= d - f\\
 x_{1,4} + x_{2,4} + x_{3,4}  &= d - f\\
 x_{1,5} + x_{2,5} + x_{3,5}  &= \Delta - 2d -1 + f
\end{align*}
As an example, the equality $ x_{3,1} + x_{3,2} + x_{3,3} + x_{3,4} + x_{3,5} = \Delta -2d -2$ says that all possible $\diaginf{\bbACX}{\cdot}$ must be exactly $\Delta - 2d -2$, since the exponent of $\bbACX$ in $\fL_1$ is $\Delta -2d -1$ and one copy of $\bbACX$ has been used for computing $\diagsup{\bbACX}{\bbXY}$. The other equations are obtained in the same way.

\paragraph{Free variable on the target lines.}
Let us now consider the two target lines. Observe that the line $\fT_1$ has a parameter $k$. We are free to choose our preferred value of $k$, and we set it as $k = d-x_{3,2}-x_{3,3}$.
Let $\mathcal{A}_4$ be the set containing the following equation.
\begin{align*}
	&k = d-x_{3,2}-x_{3,3}
\end{align*}

\paragraph{Exponents on the target lines.}
In order to use line $\fT_1$ as a target, we need to satisfy its requirements, that is, the exponents must not be negative. For $\fT_2$ there are no such requirements, since its exponents have no free variables. We call $\mathcal{P}_{1,1}$ the inequalities that must be satisfied for targeting $\fT_1$, and these are as follows.
\begin{align*}
	&k \ge 0 \\
	&d - k \ge 0\\
	&\Delta - 3d + k -2 \ge 0\\
\end{align*}

\paragraph{A perfect matching between the obtained line and a target line.}
In order, for a combination of $\fL_1$ and $\fL_2$ to be dominated by $\fT_1$, there must exist a perfect matching between the labels of the combination and the labels of $\fT_1$.
As discussed in the proof of \Cref{lem:auto}, a sufficient condition for a perfect matching to exist is that, for every \emph{right-closed cut} $R$ of $\fT_1$, it holds that the number of labels of $\fC$ that can only be matched with elements of $R$ is at most the number of times the elements of $R$ appear in $\fT_1$. The right-closed cuts of $\fT_1$ are $\{\{\}, \allowbreak \{\bbempty\},\allowbreak \{\bbempty,\bbCXY\},\allowbreak \{\bbempty,\bbCXY,\bbACXYp\},\allowbreak \{\bbempty,\bbCXY,\bbBCXYp\},\allowbreak \{\bbempty,\bbCXY,\bbACXYp,\bbBCXYp\},\allowbreak \{\bbempty,\bbCXY,\bbACXYp,\bbBCXYp,\bbABCXYp\}\}$.

Let us consider the case $R = \{\bbempty\}$. The labels $S= \{\bbX , \bbempty  , \bbXY , \bbAXYp  , \bbBXYp  , \bbABXYp  ,\allowbreak \bbCX  ,\allowbreak \bbACX\}$ can only be matched with elements in $R$, and they appear in $\fC$ for $1 + x_{1,1} + x_{1,2} + x_{1,3} + x_{1,4} + x_{1,5} + x_{2,1} + x_{3,1}$ times, while $\bbempty$ appears in $\fT_1$ for $d+2$ times. Hence, in order for a perfect matching to exist, it must hold that  $1 + x_{1,1} + x_{1,2} + x_{1,3} + x_{1,4} + x_{1,5} + x_{2,1} + x_{3,1} \le d+2$. 

Another example is $R = \{\bbempty,\bbCXY\}$.  The labels $S= \{\bbX , \bbempty  , \bbXY , \bbAXYp  , \bbBXYp  , \bbABXYp  , \bbCX  ,\allowbreak \bbCXY,\allowbreak \bbACX\}$ can only be matched with elements in $R$, and they appear in $\fC$ for $1 + x_{1,1} + x_{1,2} + x_{1,3} + x_{1,4} + x_{1,5} + x_{2,1} + x_{2,2} + x_{3,1}$ times, while the labels in $R$ appear in $\fT_1$, in total, for $d+2 +k$ times. Hence, in order for a perfect matching to exist, it must hold that  $1 + x_{1,1} + x_{1,2} + x_{1,3} + x_{1,4} + x_{1,5} + x_{2,1} + x_{2,2} + x_{3,1} \le d+2 +k$. If we consider all right-closed cuts, we obtain the following inequalities, that we call $\mathcal{P}_{1,2}$. 
\begin{align*}
	&x_{1,1} + x_{1,2} + x_{1,3} + x_{1,4} + x_{1,5} + x_{2,1} + x_{3,1} + 1 \le d + 2\\
	&x_{1,1} + x_{1,2} + x_{1,3} + x_{1,4} + x_{1,5} + x_{2,1} + x_{2,2} + x_{3,1} + 1 \le d + k + 2\\
	&x_{1,1} + x_{1,2} + x_{1,3} + x_{1,4} + x_{1,5} + x_{2,1} + x_{2,2} + x_{2,3} + x_{3,1} + x_{3,2} + x_{3,3} + 1 \le 2d + 2\\
	&x_{1,1} + x_{1,2} + x_{1,3} + x_{1,4} + x_{1,5} + x_{2,1} + x_{2,2} + x_{2,4} + x_{3,1} + 1 \le 2d + 2\\
	&x_{1,1} +x_{1,2} + x_{1,3} + x_{1,4} + x_{1,5} + x_{2,1} + x_{2,2} + x_{2,3} + x_{2,4} + x_{3,1} + x_{3,2} + x_{3,3} + 1 \le 3d - k + 2\\
	& x_{1,1} +x_{1,2} + x_{1,3} + x_{1,4} + x_{1,5} + x_{2,1} + x_{2,2} + x_{2,3} + x_{2,4} + x_{2,5} \\
	&~~~~~~~~~~~~~~~~~~~~~~~~~~~~~~~~~ + x_{3,1} + x_{3,2} + x_{3,3} + x_{3,4} + x_{3,5} + 1 \le \Delta 
\end{align*}
Similarly, for $\fT_2$, we obtain the following inequalities, that we call $\mathcal{P}_{2}$. An interesting example here is the case $R = \{\}$. In fact, observe that the label $\bbempty$ of $\fC$ cannot be mapped to any label of $\fT_2$, and hence we get the inequality $x_{1,1} \le 0$.
\begin{align*}
	&x_{1,1} \le 0\\
	&x_{1,1} + x_{1,2} + x_{1,4} + x_{2,1} + x_{2,2} + x_{2,4} + x_{3,1} + 1 \le d + 1\\
	&x_{1,1} + x_{1,2} + x_{1,3} + x_{1,4} + x_{2,1} + x_{2,2} + x_{2,3} + x_{2,4} + x_{3,1} + x_{3,2} + x_{3,3} + 1 \le  2d + 1\\
	&x_{1,1} + x_{1,2} + x_{1,3} + x_{1,4} + x_{1,5} + x_{2,1} + x_{2,2} + x_{2,3} + x_{2,4} + x_{2,5}\\
	& ~~~~~~~~~~~~~~~~~~~~~~~~~~~~~~~~~  + x_{3,1} + x_{3,2} + x_{3,3} + x_{3,4} + x_{3,5} + 1 \le  
	\Delta
\end{align*}

\paragraph{Automatic checking.}
We define $\mathcal{P}_{1}$ as $\mathcal{P}_{1,1} \cup  \mathcal{P}_{1,2}$ and $\mathcal{A}$ as $\mathcal{A}_{1} \cup \mathcal{A}_{2} \cup \mathcal{A}_{3} \cup \mathcal{A}_{4}$. By abusing notation, by $\mathcal{A}$ we also denote the conjunction of all the inequalities contained in $\mathcal{A}$. We do the same for $\mathcal{P}_1$ and $\mathcal{P}_2$, that is, $\mathcal{P}_1 = p_{1,1} \land \ldots \land p_{1,9}$ (since $\mathcal{P}_1$ contains 9 inequalities) and $\mathcal{P}_2 = p_{2,1} \land \ldots \land p_{2,4}$.
The lines $\fT_1$ and $\fT_2$ dominate $\fC$ if the following statement holds.
\[
  \forall \Delta, f,x_{1,1},\ldots,x_{3,5} ~ (\mathcal{A} \implies \mathcal{P}_1 \lor \mathcal{P}_2)
\]
The statement says that, if our assumptions are true, then all the configurations given by the obtained line can be mapped in $\fT_1$ or in $\fT_2$. The statement is equivalent to
\begin{align*}
	& \lnot \exists \Delta, f,x_{1,1},\ldots,x_{3,5} ~ (\mathcal{A} \land (\lnot p_{1,1} \lor \ldots \lor \lnot p_{1,9}) \land (\lnot p_{2,1} \lor \ldots \lor \lnot p_{2,4}))) & \iff\\
	& \lnot \exists \Delta, f,x_{1,1},\ldots,x_{3,5} ~ ((\mathcal{A} \land \lnot p_{1,1} \land \lnot p_{2,1}) \lor \ldots \lor (\mathcal{A} \land \lnot p_{1,9} \land \lnot p_{2,4})) & \iff\\
	& \lnot (\exists \Delta, f,x_{1,1},\ldots,x_{3,5} ~ (\mathcal{A} \land \lnot p_{1,1} \land \lnot p_{2,1}) \lor \ldots \lor \exists \Delta, f,x_{1,1},\ldots,x_{3,5} ~ (\mathcal{A} \land \lnot p_{1,9} \land \lnot p_{2,4})) & \iff\\
	&  \lnot\exists \Delta, f,x_{1,1},\ldots,x_{3,5} ~ (\mathcal{A} \land \lnot p_{1,1} \land \lnot p_{2,1}) \land \ldots \land \lnot\exists \Delta, f,x_{1,1},\ldots,x_{3,5} ~ (\mathcal{A} \land \lnot p_{1,9} \land \lnot p_{2,4}) &
\end{align*}
The last statement implies that proving that $\fT_1$ and $\fT_2$ dominate $\fC$ can be reduced to show that many systems of inequalities have no solution. For an example, here we consider only the first system of inequalities (among the total of $36$ cases), that is, $\mathcal{A} \land \lnot p_{1,1} \land \lnot p_{2,1}$, and we manually show that it has no solution. Since our variables are integers, $ \lnot p_{1,1}$ and $\lnot p_{2,1}$ can be written as $	k \le -1$ and $x_{1,1} \ge 1$.
Hence, $\mathcal{A} \land \lnot p_{1,1} \land \lnot p_{2,1}$ can be written as the conjunction of the following inequalities.
\begin{align*}
		d &\ge 1 \\
	\Delta &\le 2d + 4  \\
	\Delta &\ge 2d + 3 \\
	\Delta &\ge 5\\
		f &\ge 0 \\
	d - f & \ge 0\\
	\Delta - 2d - 1 + f &\ge 0\\
		x_{i,j} &\ge 0 ~(\forall ~ i,j)\\
	 x_{1,1} + x_{1,2} + x_{1,3} + x_{1,4} + x_{1,5} &= d + 1\\
	x_{2,1} + x_{2,2} + x_{2,3} + x_{2,4} + x_{2,5} &= d\\
	x_{3,1} + x_{3,2} + x_{3,3} + x_{3,4} + x_{3,5} &= \Delta -2d -2\\
	x_{1,1} + x_{2,1} + x_{3,1} &= 1\\
	x_{1,2} + x_{2,2} + x_{3,2} &= f - 1\\
	x_{1,3} + x_{2,3} + x_{3,3} &= d - f\\
	x_{1,4} + x_{2,4} + x_{3,4}  &= d - f\\
	x_{1,5} + x_{2,5} + x_{3,5}  &= \Delta - 2d -1 + f\\
	k &= d-x_{3,2}-x_{3,3}\\
	k &\le -1 \\
	x_{1,1} &\ge 1
\end{align*}
We prove that this system of inequalities has no solution over the reals. By combining $k = d-x_{3,2}-x_{3,3}$ and $k \le -1$ we obtain that $x_{3,2}+x_{3,3} \ge d+1$. By combining 
	$x_{1,2} + x_{2,2} + x_{3,2} = f - 1$ and $x_{1,3} + x_{2,3} + x_{3,3} = d - f$ we obtain $x_{1,2} + x_{2,2} + x_{3,2} + x_{1,3} + x_{2,3} + x_{3,3}  = d -1$, that, since $x_{i,j} \ge 0$ for all $i,j$, implies $x_{3,2} + x_{3,3}  \le d -1$, which is a contradiction.

\section{Open Questions}\label{sec:open}
We conclude with some open questions. 

\paragraph{Defective coloring.}
Recent results showed that $(\bar{\Delta} + 1)$-edge coloring, where $\bar{\Delta}$ is the maximum degree of the line graph, can be solved in 
$O(\log^{12} \Delta + \log^* n)$ rounds \cite{Balliu0KO22a}. However, whether a polylogarithmic-in-$\Delta$ algorithm for computing a $(\Delta+1)$-vertex coloring exists, is a major open question. The edge coloring algorithm has been obtained by providing a subroutine that is able to $2$-color the edges with defect $(1+\varepsilon)\bar{\Delta}/2$, and understanding whether similar subroutines exist for vertex coloring is an interesting open question.
 \Cref{sec:2col} states that a similar result cannot be achieved for vertex coloring by using $2$ colors, and \Cref{sec:3col} states that even with $3$ colors this is not doable, since we would need a subroutine for computing a $3$-coloring with defect $(1+\varepsilon)\Delta /3$, but the lower bound that we provided implies that we cannot obtain a defect smaller than $\Delta/2 - O(1)$. 
  As a starting point for understanding defective colorings in general, we point out that our lower bound for the case of $3$ colors does not match the best known upper bound, which requires a defect of at least $\frac{2\Delta - 4}{3}$ \cite{BalliuHLOS19}. 
\begin{oq}
	For which values of $d$ is the $d$-defective $3$-coloring problem solvable in $O(f(\Delta) \cdot \log^* n)$ rounds, for some function $f$, in graphs of maximum degree $\Delta$?
\end{oq}

In general, it would be interesting to understand for what values of $d$, $c$ and $\Delta$, $d$-defective $c$-coloring is solvable in $O(f(\Delta) \cdot \log^* n)$ rounds, for some function $f$, in graphs of maximum degree $\Delta$. In particular, it is known that $d$-defective $c$-coloring can be solved in $O(\log^* n)$ rounds if $c = O((\frac{\Delta}{d+1})^2)$, while it is known to require $\Omega(\log_\Delta n)$ rounds if $c \le \frac{\Delta}{d+1}$, so there is a wide gap between the lower and the upper bound on the number of colors that makes the problem solvable in $O(f(\Delta) \cdot \log^* n)$ rounds. Since variants of defective coloring are at the core of many coloring algorithms, we believe that understanding defective coloring is an interesting open question.
\begin{oq}
	For which values of $d$, $c$, and $\Delta$, is the  $d$-defective $c$-coloring problem solvable in $O(f(\Delta) \cdot \log^* n)$ rounds, for some function $f$, in graphs of maximum degree $\Delta$?
\end{oq}

\paragraph{Fixed points.}
While our fixed point procedure is able to provide most of the fixed points present in the literature, there is one notable exception. In \cite{Balliu0KO23}, a fixed point relaxation for a problem called \emph{hypergraph colorful $(r-1)\Delta$-coloring} is shown, and the fixed point is obtained by first proving that a solution for hypergraph colorful $\Delta(r-1)$-coloring can be converted in $0$ rounds into a solution for the more standard hypergraph $\Delta$-coloring, and by then providing a fixed point relaxation for this latter problem. By running procedure $\fpp$ on hypergraph $\Delta$-coloring, we would indeed obtain the same fixed point relaxation shown in \cite{Balliu0KO23}. However, if we run it on 
hypergraph colorful $(r-1)\Delta$-coloring, it would fail (the obtained problem would be $0$-round-solvable). We would like to understand whether there exists a \emph{universal} procedure for finding fixed point relaxations.
\begin{oq}
	Assume that there exists a fixed point relaxation $\Pi'$ for a problem $\Pi$. Is there a procedure that is able to find $\Pi'$ automatically?
\end{oq}

\paragraph{Simpler proofs.}
While we have been able to provide a non-trivial fixed point relaxation for $\lfloor (\Delta - 3)/2 \rfloor$-defective $3$-coloring, the proof required to automate a case analysis by using computer tools. While such a result is very interesting, as it shows that proofs based on round elimination can sometimes be automated (and hence answering affirmatively Open Question 9 in \cite{hideandseek}), we would like to understand whether there is a simpler, shorter, and more natural proof. This could help in understanding more complicated cases, such as defective colorings with more colors.
\begin{oq}
	Is there a simple and compact proof for the fact that the fixed point relaxation that we provided for  $\lfloor (\Delta - 3)/2 \rfloor$-defective $3$-coloring is indeed a fixed point?
\end{oq}

\urlstyle{same}
\bibliographystyle{alpha}
\bibliography{fixed-point}

\end{document}